\def\llncs{0}  
\def\draft{0}  
\def\anon{0}   
\def\fullversion{1} 
\def\quantinfo{0}
\ifnum\anon=1
  \def\effectivellncs{1}
\else
  \def\effectivellncs{\llncs}
\fi

\ifnum\effectivellncs=1
  \documentclass[runningheads]{llncs}
\else
  \documentclass[11pt,letterpaper]{article}
  \ifnum\fullversion=1
    \usepackage{fullpage}
  \else
    \usepackage[margin=1in]{geometry}
  \fi
\fi

\usepackage[T1]{fontenc}
\usepackage{silence}
\usepackage{mathtools,amssymb}
\mathtoolsset{showonlyrefs=false,showmanualtags}
\usepackage{braket,comment}
\usepackage{physics}
\usepackage{xcolor}
\usepackage{array}
\usepackage{colortbl}
\usepackage{multirow}
\makeatletter
\let\prismorignewcounter\newcounter
\newcommand{\prismpreparecounter}[1]{%
  \ifnum\pdfstrcmp{#1}{Hpclinenumber}=0 \global\let\theHpclinenumber\relax\fi%
  \ifnum\pdfstrcmp{#1}{H@pclinenumber}=0 \global\expandafter\let\csname theH@pclinenumber\endcsname\relax\fi%
  \ifnum\pdfstrcmp{#1}{Hpcgamecounter}=0 \global\let\theHpcgamecounter\relax\fi%
  \ifnum\pdfstrcmp{#1}{Hpcrlinenumber}=0 \global\let\theHpcrlinenumber\relax\fi%
}
\renewcommand{\newcounter}[1]{\@ifnextchar[{\prismnewcounter@i{#1}}{\prismnewcounter@ii{#1}}}
\newcommand{\prismnewcounter@i}[2]{\prismpreparecounter{#1}\prismorignewcounter{#1}[#2]}
\newcommand{\prismnewcounter@ii}[1]{\prismpreparecounter{#1}\prismorignewcounter{#1}}
\makeatother
\usepackage{cryptocode}
\makeatletter
\let\newcounter\prismorignewcounter
\makeatother
\usepackage[hypertexnames=false]{hyperref}
\usepackage[capitalize,nameinlink,noabbrev]{cleveref}
\ifnum\effectivellncs=0
  \usepackage{amsthm}
  \newtheorem{theorem}{Theorem}
  \newtheorem{lemma}{Lemma}
  \newtheorem{claim}{Claim}
  \newtheorem{proposition}{Proposition}
  \newtheorem{corollary}{Corollary}
  \theoremstyle{definition}
  \newtheorem{definition}{Definition}
  \theoremstyle{remark}
  \newtheorem{remark}{Remark}
  \theoremstyle{plain}
  
\else
  \spnewtheorem{openproblem}{Open Problem}{\bfseries}{\itshape}
\fi
\crefname{openproblem}{Open Problem}{Open Problems}
\Crefname{openproblem}{Open Problem}{Open Problems}
\ifnum\draft=1
  \newcommand{\authornote}[3]{\textcolor{#2}{[#1: #3]}}
  \newcommand{\amit}[1]{\authornote{Amit}{blue}{#1}}
  \newcommand{\vipul}[1]{\authornote{Vipul}{red}{#1}}
\else
  \newcommand{\authornote}[3]{}
  \newcommand{\amit}[1]{}
  \newcommand{\vipul}[1]{}
\fi
\newcommand{\Esq}{E_{\mathrm{sq}}}

\newcommand{\qpvb}{\ensuremath{\mathsf{QPV\text{-}base}}}
\newcommand{\qpva}{\ensuremath{\mathsf{QPV\text{-}amplified}}}
\newcommand{\TD}[1]{\operatorname{TD}\!\left[#1\right]}
\newcommand{\RR}{\mathbb{R}}

\newcommand{\EE}{\mathbb{E}}

\newcommand{\cvhull}{\mathrm{Convex\text{-}Hull}}
\newcommand{\relocate}{\mathrm{Relocate}}

\renewcommand{\secparam}{\ensuremath{\lambda}}

\renewcommand{\secpar}{\lambda}
\newcommand{\bit}{\{0,1\}}

\begin{document}

\title{Breaking the Bounded Entanglement Barrier for Quantum Position Verification}
\ifnum\effectivellncs=1
  \titlerunning{Quantum Position Verification}
  \ifnum\anon=1
    \author{}
    \authorrunning{}
    \institute{}
  \else
    \author{Amit Behera \and Vipul Goyal}
    \authorrunning{A. Behera and V. Goyal}
    \institute{NTT Research\\
    \email{amitbehera1767@gmail.com, }\email{vipul@vipulgoyal.org}}
  \fi
\else
  \ifnum\anon=1
    \author{}
  \else
    \author{Amit Behera\\NTT Research\\\texttt{amitbehera1767@gmail.com}
    \and
    Vipul Goyal\\NTT Research\\\texttt{vipul@vipulgoyal.org}}
  \fi
  \date{}
\fi
\maketitle

\begin{abstract}
Position verification, introduced by Chandran et al. (SIAM J. Computing 2014), allows verifiers to test a prover's claimed position by an interactive protocol. Classical position verification is impossible. Even for Quantum Position Verification (QPV), there always exists an LOCC (Local Operations and Classical Communication) attack if the adversary can hold an exponentially large amount of preshared entanglement. 

Somewhat surprisingly, we show that we can circumvent this barrier in the idealized continuous-time model, where time is represented by a real-valued parameter and challenge messages can be sent at a time sampled uniformly from a real interval. In this model, we give a BB84-based protocol that remains secure against any finite coalition of LOCC adversaries with arbitrary finite (possibly exponential) quantum storage and entanglement.

Our construction can also be instantiated in the discrete-time model. Even though the previous impossibility results apply, we are able to obtain an information-theoretic QPV protocol in which the honest parties' total resources (communication and storage) can be significantly smaller than the adversarial resource bound. In fact, we can achieve any desired polynomial gap between honest parties and adversarial resources. To our knowledge, all prior constructions that do not rely on idealized oracles required resources of the honest parties to be at least as large as the proven adversarial entanglement bound.

Interestingly, the resource gap between the honest and adversarial party resources in our protocol depends on how precisely time can be measured. As the precision of the best-known clock improves with further research, the resource gap in our protocol keeps increasing. In particular, the adversarial resource bound keeps increasing while the honest parties' resources remain largely the same. We call this the "I sleep, you work" paradigm. Finally, we show that our techniques extend to an information-theoretic, position-based message authentication protocol, which in turn extends to an information-theoretic, position-based quantum key exchange protocol with the same guarantees as our position-verification protocol in the continuous-time and discrete-time models.
\end{abstract}
    
\makeatletter
\let\l@title\@gobbletwo
\let\l@author\@gobbletwo
\makeatother
\setcounter{tocdepth}{2}
\clearpage
\tableofcontents
\clearpage

\section{Introduction}
Position verification, introduced by Chandran et al.~\cite{CGM+14}, allows verifiers to check a prover's claimed location by an interactive protocol subjected to timing constraints. Position verification forms the foundation for position-based cryptography, where a party's location acts as its credential. However, Chandran et al.~\cite{CGM+14} showed that position verification is impossible using only classical information, based on the fact that adversaries surrounding the claimed location but not at it can copy, relay, and process the verifiers' messages to mimic an honest prover's actions.
In the quantum setting, Quantum position verification (QPV), introduced by Buhrman et al.~\cite{BCF+14}, attempts to circumvent such a generic attack by using quantum information, which cannot be copied in general, thanks to the no-cloning principle of quantum mechanics~\cite{WZ82,Die82}. Unfortunately, even QPV cannot be unconditionally secure, as Ref.~\cite{BCF+14} shows that any QPV protocol is insecure against a coalition of adversaries with LOCC strategies and a sufficiently large amount of preshared entanglement. Thus, QPV protocols can only exist under some restriction of bounded adversarial entanglement for LOCC adversaries. From now on, by adversaries we only mean LOCC adversaries (see \Cref{sec:adversary-model} for more details). 

\subsection{Our Results}

Somewhat surprisingly, we show that in the idealized continuous-time model, where (similar to prior work) all computations are assumed to be instantaneous, the previous impossibility result of Buhrman et al. can be bypassed. We provide an unconditional, information-theoretically secure construction that achieves security even against an adversary holding any finite, possibly exponential, amount of entanglement. To understand why we can bypass the previous impossibility result, consider the following construction idea.

\paragraph{Construction} Consider the following simple protocol with two verifiers on opposite sides of the claimed location on a 1-D line, equidistant from the claimed location.
\begin{itemize}
    \item Suppose the two verifiers prepare a random BB84 state and send it to the honest prover at the beginning of the protocol.
    \item At a fixed time during the protocol, one of the verifiers, selected uniformly at random, asks the prover to return it to them.
    \item The honest prover only needs to store the state and send it to the selected verifier.
    \item The selected verifier first checks if the returned qubit arrived on time, and, if so, checks if the returned qubit is indeed the BB84 state they prepared, and accepts if both checks pass.
\end{itemize}

  There is a natural LOCC attack for the above construction. Consider two adversaries $A_0$ and $A_1$ located between the two verifiers but on opposite sides of the claimed location. Before the protocol, $A_0$ receives the qubit sent by the verifiers. During the protocol, $A_0$ first waits for any return instruction sent from its nearby verifier, and if so, returns the qubit. If no instruction arrives (meaning the verifier closer to $A_1$ has asked to return), $A_0$ sends the qubit to $A_1$ using quantum teleportation, and $A_1$ applies the required correction on its half of the Bell pair to recover the BB84 state and then sends it to its nearby verifier. The teleportation process consumes exactly one pre-shared Bell pair, but the attack succeeds with probability $1$. Buhrman et al.~\cite{BCF+14} show that this attack idea can be extended against any general QPV protocol with multiple rounds of communication. In particular, adversaries can use (potentially an exponential amount of) preshared entanglement to perform teleportation heavily in order to break the security of any protocol.

A crucial assumption in the impossibility result of Buhrman et al.~\cite{BCF+14} is that there is only a finite number of time slots where the messages or instructions can possibly be sent by the verifier. In their generic attack, in every such time slot, the adversaries use quantum teleportation to send quantum states among themselves, which consumes a (possibly exponential) amount of quantum entanglement. 

\paragraph{Randomizing the Challenge Message Time} 

Our key idea to bypass the impossibility result of Buhrman et al. in the continuous-time model is to randomize the time at which the selected verifier sends the challenge to the prover in the above protocol. Since we are in the continuous-time model, the number of possibilities for the time when the challenge message can appear is no longer finite. In particular, the upgraded protocol is as follows.

\begin{enumerate}
    \item As before, the verifiers prepare a BB84 qubit and send it to the prover.
    \item Since we are in the continuous-time model, we view a protocol session as a fixed time interval of, say, one second, given by $[0,1]$. The verifiers select a time $t\in [0,1]$ uniformly at random, along with one of the two verifiers who will issue a return request to the prover.
    \item The selected verifier sends a return request to the prover \emph{at the selected time $t$}.
    \item The prover, on receiving the return request, sends the state immediately to the selected verifier, who first checks if the returned qubit arrived on time, and if so checks if the returned qubit is indeed the BB84 state they prepared.
\end{enumerate}

To perform the attack given by the impossibility result of Buhrman et al.~\cite{BCF+14}, in every possible time slot where the challenge message might arrive, the adversaries would have to perform quantum teleportation (thus consuming a constant amount of entanglement) in order to be prepared to respond in time. However, now the number of such possible time slots is not finite and hence any finite amount of preshared entanglement will not be sufficient to perform the attack. 

Note that in the above protocol, the adversary can always guess with probability half which verifier will send the challenge message and win if the guess is correct. To remedy this, we show how to amplify the security using sequential repetition. Overall, we get the following result in the continuous-time model. 
\begin{theorem}[Informal result in the continuous-time model]
\label{thm:continuous-main-result-informal}
In the continuous-time model, there exists a QPV protocol in 1-D such that every finite LOCC adversarial coalition, with an arbitrary finite quantum-register size and any finite amount of preshared entanglement, has a negligible success probability. 
\end{theorem}
 The formal statement and proofs are given in \Cref{sec:continuous-time-model}.

\paragraph{Moving to the Discrete-time Model.}

Even though treating time as continuous is a standard practice in various areas of physics, including classical and quantum mechanics, electrodynamics, and general relativity (see Refs.~\cite{GoldsteinClassicalMechanics,JacksonClassicalElectrodynamics,WaldGeneralRelativity,SakuraiModernQuantumMechanics}), implementing the continuous-time model in the real world is impossible. Indeed, the Margolus--Levitin bound~\cite{ML98} limits how quickly a physical system with bounded average energy can pass through distinguishable states.

Fortunately, our ideas for the continuous-time model carry over seamlessly to the discrete-time model. In the discrete-time model, the previous impossibility of Buhrman et al. does apply, and hence the best one can hope for is to obtain a gap between honest and adversarial resources, where by resources we mean total communication and storage combined. To our knowledge, all prior results in the literature which do not rely on idealized oracles require the honest parties' resources to be at least as large as the assumed adversarial resource bound. We break this barrier and provide a construction in which adversarial resources can be significantly larger than the honest party resources, and security still holds. In fact, we can achieve any desired polynomial gap between adversarial and honest party resources by choosing the protocol parameters appropriately.

We start by noting the main changes in our protocol in the discrete-time model.
\begin{enumerate}
    \item In each of the sequentially repeated rounds there are \emph{finitely many} possible points in time referred to as slots, in which the verifiers might possibly send the return instruction (but don't have to). Concretely, we work with $m$ sequentially repeated rounds, and each round has $s$ slots.
    \item As in the continuous-time model, the verifiers send a random BB84 state before the round, and then one of the verifiers is selected at random. But now the selected verifier samples a random slot $i\in [s]$ and sends the return instruction in this slot. 
\end{enumerate}

The key conceptual insight is that, among all possible slots in a round, there is a single real challenge slot when the verifiers send the challenge, and the rest are dummy slots. The honest parties operate only during the real challenge slot. Hence, the honest parties’ resources remain largely independent of the number of slots $s$. On the contrary, since the real challenge slot is hidden from the adversaries, any adversarial coalition must remain prepared, and intuitively should keep teleporting the BB84 state for the majority of the slots to ensure a sufficiently large constant success probability.
We call this the "I sleep you work" paradigm because, in an overwhelming majority of the slots, the honest parties are simply sleeping while the adversaries need to actively work, thereby consuming entanglement proportional to the number of slots.

\paragraph{Another Source of Resource Gap} We also generalize our protocol to multiple BB84 states, rather than a single BB84 state, which provides an even stronger resource gap. In particular, the final protocol with $m$ rounds, each consisting of $s$ slots, has an additional parameter $R_H$, representing the total number of qubits an honest prover can store, and the protocol is as follows.  
 \begin{enumerate}
     \item The verifiers send $R_H$ qubits at the beginning of the first round, which the honest prover stores as an ordered list of qubits.
     \item Each round consists of $s$ possible challenge slots, and in each round, the selected verifier, along with sampling a random slot $i\in [s]$, also samples a \emph{random index} out of the $R_H$ options and sends a return request for the chosen index at slot $i$. Upon receiving the returned qubit, the selected verifier checks if the qubit was returned on time and if it is indeed the BB84 qubit at the chosen index. 
     \item After each round ends, the verifiers provide the prover with a fresh, random BB84 state to replenish the qubit used up in the last round.
 \end{enumerate}

Note that the extension to $R_H$ qubits adds an additive factor of $\widetilde{O}(R_H)$ for the honest parties' resources, but note that now the adversaries would need to keep teleporting $R_H$ qubits as opposed to a single qubit in every slot, which should increase the bound on the adversarial resources by an additive factor of $\tilde{\Omega}(mR_Hs)$. We show that this is indeed the case; thus, for large $R_H$, we get an even stronger resource gap. In particular, setting the parameters $m,s,R_H$ as polynomial functions of the security parameter $\secparam$, we get the following result.
\begin{theorem}[Informal main result in the discrete-time model]\label{thm:main-result-informal}
Let $s(\secparam)\geq 9, R_H(\secparam)\geq 1$ and $m(\secparam)\in \omega(\log^2(\secparam))$ be polynomially bounded parameters. Then there exists an efficient QPV protocol in 1-D based on BB84 states\footnote{The protocol only requires very simple quantum states called BB84 states~\cite{BB84}, whose preparation and transmission have been studied extensively in quantum key distribution; see Refs.~\cite{BBB+92,BCD+09,LMP+20,ABS+25}. }, requiring total communication and storage of $\widetilde{O}(R_H+m)$ for the honest prover and verifiers such that any LOCC coalition of finite adversaries, with non-negligible success probability against our protocol, must have total quantum-register size $\Omega(sR_Hm/\log_2^2(\secparam))=\widetilde{\Omega}(sR_Hm)$. 
\end{theorem}

The formal result is given in \Cref{thm:amplified_security-generalized}.

\paragraph{How Many Slots Can We Have?} Note that the number of slots is mainly limited by how precisely one can measure time. Our construction allows overlapping slots (i.e., a slot does not need to be long enough for the protocol in that slot to finish before the next slot starts). Hence, by using each tick of the best-known clocks as a slot, our protocol can be implemented such that the security of this implementation is directly tied to the precision of the best available clock, as follows.

\begin{theorem}[Concrete Security in the discrete-time model]
\label{thm:clock-main-result-informal}
Let $m,s,R_H,k$ with $1\leq k<m,s\geq 9, R_H\geq 1$ be integer protocol parameters. Suppose there is a clock with $s$ reliably distinguishable ticks in a unit interval of time. Then there exists a QPV protocol in 1-D such that every finite LOCC adversarial coalition with success probability at least $\nu^k$ must have total quantum-register size at least $csR_Hm/k$, for some universal constants $c,\nu$ where $c>0,\nu=(1-10^{-7})$; see \Cref{rem:concrete-security} for more details.
\end{theorem}

Therefore, our construction has the distinctive feature that future improvements in clocks and timing systems automatically yield better protocol instantiations with a stronger resource gap. Modern optical atomic clocks already provide very precise measurements of time. For example, the single-ion optical clock reported in 2025 uses an optical reference oscillating at approximately $1.12\times10^{15}$ optical ticks (cycles) per second~\cite{MarshallClock25}.\footnote{The number of challenge slots that the complete implementation can reliably distinguish also depends on its synchronization, control, transmission, and detection systems.} Hence, even the current clock-based implementation of our scheme already provides a strong and concrete resource gap. 

\paragraph{Applications to Position-Based Authentication and Key Exchange}
Our techniques extend to a position-based authentication protocol, where a prover at the claimed location sends a classical message to the verifiers, and the adversaries, even with the help of the honest prover, should not be able to make the verifiers accept a message that the prover did not send. Following the template used in previous works; see~\cite{BCF+14,Unr14}, our position-based authentication protocol in turn yields a position-based Quantum Key Exchange scheme.
\begin{theorem}[Informal applications to authentication and key exchange]
\label{thm:authentication-qkd-informal}
Let $s\geq s_0$, $R_H\geq 4$, and $m(\secparam)\in\omega(\log^2(\secparam))$ be polynomially bounded parameters, where $s_0$ is some suitable constant. Then there exists a position-based authentication protocol in 1-D, requiring total communication and storage of $\widetilde{O}(R_H+m)$ for the honest prover and verifiers, such that any LOCC coalition of finite adversaries which, with non-negligible probability, makes the verifiers accept a message not sent by the honest prover, must have total quantum-register size $\Omega(sR_Hm/\log_2^2(\secparam))=\widetilde{\Omega}(sR_Hm)$. 
Further, in the continuous-time model, the protocol is secure against every finite LOCC adversarial coalition. 

Moreover, the protocol gives position-based key exchange with the same guarantees and resource bounds.
\end{theorem}
The formal results are given in \Cref{pba:thm:auth,pba:cor:qkd}.

\subsection{Discussion on Our Model} 

In our adversarial model, adversaries may start with any preshared quantum state, but during the protocol, they are allowed to exchange only classical messages with one another. Their joint operations are therefore LOCC, and for simplicity, we assume that their local quantum operations are instantaneous. This restriction applies only to communication among the adversaries; the adversaries (and the honest prover) are allowed to return quantum responses to the verifiers. It is easy to see that the model includes the standard teleportation and port-based teleportation attacks used in generic impossibility results for QPV; see Refs.~\cite{BCF+14,BK11}.

We measure security for adversaries by the total size of their quantum registers. A protocol is called $n(\lambda)$-qubit secure if every coalition with no member at the claimed location and at most $n(\lambda)$ qubits of quantum registers convinces the verifiers only with negligible probability. Classical storage is free and unrestricted, so transcripts, measurement results, and verifier information revealed in earlier rounds do not count toward this bound. Thus the model limits the adversaries' quantum storage, and not their classical computation or communication.

\paragraph{Why Bounded Quantum Storage Instead of Bounded Preshared Entanglement}\label{pg:bounded-register-size} In order to obtain the polynomial resource gap in \Cref{thm:main-result-informal}, we parametrize security in terms of total quantum-register size, whereas a natural alternative would be to formulate security, using a suitable entanglement measure, by saying that a protocol is secure whenever the adversaries initially share at most a specified amount of entanglement. To see the problem with the alternative, consider the following example.
\begin{enumerate}
    \item Suppose a QPV protocol has an attack $\mathcal A$ that uses $N:=\secparam^cR_H$ Bell pairs and succeeds with non-negligible probability $p$, where $R_H$ is the honest prover's quantum storage.
    \item The adversaries can instead start with a mixture of $N$ Bell pairs with probability $1/N$ and a product state with the rest of the probability. They run $\mathcal A$ and if the initial state is indeed the entangled state, they would succeed with probability $p$. 
    \item Their overall success probability is therefore at least $p/N$, which is still non-negligible because $N$ is polynomial in $\secparam$. For a convex entanglement measure, however, the entanglement of the mixture that the adversaries start with is only $O(1)$.
\end{enumerate}

In contrast, the adversaries' total quantum-register size still has to be at least $N$, as their fixed quantum registers must be large enough to hold all $N$ Bell pairs whenever the entangled state occurs. The discrepancy arises because quantum-register size is a worst-case resource measure whereas convex entanglement measures are average-case measures. This is why we formulate security in terms of total quantum-register size. Entanglement will be used only as an intermediate quantity for lower-bounding the adversaries' quantum registers.

\paragraph{Non-instantaneous Computation} The assumption that the computation is instantaneous is made in all prior works on position-based cryptography. We can account for computation time by giving additional time to the prover to respond to the verifiers. However, if the adversary can compute and respond faster than the honest prover, this would introduce a precision error. That is, the difference between the allowed time (for an honest prover) and the actual time taken by an adversary for an operation will determine the precision with which positioning can be achieved. This argument applies to our constructions as well. However, in our protocol, this difference will also dictate the number of slots. As the adversary gets an additional time buffer, it can respond to messages received in a slot even in a later time slot. Thus, each slot must be correspondingly longer to account for this extra buffer. 

Fortunately, the concrete protocol that we build in this paper requires almost no computation. In our basic protocol, the prover stores a single qubit and has to simply send it back to the selected verifier without performing any operation on it. In the amplified protocol, the prover stores multiple qubits and has to select and return the relevant one.

\paragraph{Moving to Higher Dimensions}

Beyond 1-D, we also show that our protocol remains secure in higher dimensions, provided the verifiers are mobile, i.e., they are allowed to move according to the claimed location. We show a reduction that, in both continuous-time and discrete-time models, upgrades the security of our protocol to security with respect to mobile verifiers in higher dimensions, without any loss in the security bounds. Thus, we obtain mobile-verifier versions of the same results in higher dimensions. The formal results are given in \Cref{thm:generic-reduction,cor:continuous-time-higher-dimensional}.

\subsection{Open Questions}\label{sec:open-questions}
In the discrete-time model, we provide a family of QPV protocols, such that for every polynomial $p(\secparam)$ there is a protocol in the family that achieves that resource gap $p(\secparam)$. However, this is not tight, and in particular, does not match the impossibility result of Buhrman et al.~\cite{BCF+14}, which only rules out an arbitrary (exponential) resource gap. This leaves open the question of constructing a single QPV protocol that achieves an arbitrary polynomial resource gap, or even a fixed exponential resource gap.

\subsection{Related Works}\label{sec:related-works}

Chandran et al.~\cite{CGM+14} introduced position verification and proved that it is impossible classically. They also constructed position-verification protocols in the bounded-retrieval model, showing position verification is possible when the adversaries' storage bound exceeds the honest storage. Early quantum variants, called \emph{quantum tagging}, were proposed by Kent, Munro, and Spiller~\cite{KMS11}; teleportation-based attacks on these and related protocols appear in Refs.~\cite{KMS11,LL11,BCF+14,BK11}.

Buhrman et al.~\cite{BCF+14} formally introduced Quantum Position Verification (QPV) and gave the BB84-based protocol $\mathsf{QPV}_{\mathrm{BB84}}$, which is secure against adversaries without preshared entanglement. Repeating the protocol in parallel remains secure when the adversarial state contains at most $\alpha n$ qubits for a sufficiently small constant $\alpha>0$~\cite{TFKW13}. When adversaries may communicate only classically, Ribeiro and Grosshans~\cite{RG15} obtain a nearly tight local-register lower bound of $n-O(\log n)$ qubits.\footnote{Their theorem is stated as a lower bound on the max-relative entropy of entanglement of the adversarial preshared state. Since this quantity is at most the logarithm of either local Hilbert-space dimension, their result also implies the same asymptotic lower bound on the local quantum-register size.} The latter two results give adversarial lower bounds linear in $n$, of the same asymptotic order as the honest parties' total communication and storage.

Two important families of one-dimensional protocols are $f$-routing~\cite{BFS+13,KMS11} and $f$-BB84~\cite{BCS22}, both defined using a Boolean function $f:\bit^n\times\bit^n\to\bit$. Garden-hose complexity provides attacks and lower bounds for $f$-routing~\cite{BFS+13}, and later work relates attacks on $f$-routing and $f$-BB84 and improves several constructions~\cite{CM23,ABM+24,BHM+25}. For a random $f$, Bluhm, Christandl, and Speelman~\cite{BCS22} show that adversaries with fewer than $n/2-O(1)$ local qubits on each side fail with constant probability. Ref.~\cite{ACM25} gives function-dependent rank lower bounds for a restricted class of attacks,\footnote{Namely, it lower bounds the logarithm of the Schmidt rank for one-sided-perfect $f$-routing attacks and perfect $f$-BB84 attacks.} including linear bounds for several explicit functions. The honest communication in these protocols is linear in $n$, while their total storage also includes the space needed to represent and evaluate $f$.

Gao, Liu, and Wen~\cite{GLW16} studied non-simultaneous QPV under a bound on the frequency of adversarial attack operations, whereas our proof derives an entanglement-consumption bound without imposing such a restriction.

Another line of work studies robustness and practical implementation. BB84-based protocols have been analyzed with channel loss, weak-laser sources, and arbitrary transmission loss~\cite{QS15,LXS+16,ES23,ABB+25}. Recent experiments and proposals study QPV based on single-photon sources and two-photon interference, as well as device-independent QPV~\cite{KPB+25,KGZ+25}.

Beyond BB84 protocols, Junge et al.~\cite{JKP+22} propose a protocol whose honest total communication and storage are $O(n^2)$. They obtain polynomial lower bounds on adversarial quantum resources for a regular class of attacks and extend these bounds to general attacks under a conjecture from Banach-space theory. In the quantum random-oracle model (QROM), Unruh~\cite{Unr14} constructs QPV secure against polynomially many oracle queries. Liu, Liu, and Qian~\cite{LLQ22} introduce classically verifiable QPV (CVPV), with QROM constructions that are secure against polynomial-time adversaries and a construction based on subexponentially hard Learning with Errors (LWE) that fixes in advance a bound on the adversaries' quantum communication during preprocessing.

We summarize these results in \Cref{fig:related-work-comparison}.
    
\begin{figure}[p]
\centering
\normalsize
\linespread{0.92}\selectfont
\setlength{\tabcolsep}{1.4pt}
\setlength{\arrayrulewidth}{0.6pt}
\renewcommand{\arraystretch}{1.02}
\setlength{\extrarowheight}{0.5pt}
\makebox[\linewidth][c]{%
\begin{tabular}{|>{\raggedright\arraybackslash}p{0.14\linewidth}|>{\columncolor{cyan!25}\raggedright\arraybackslash}p{0.09212\linewidth}|>{\raggedright\arraybackslash}p{0.1034\linewidth}|>{\raggedright\arraybackslash}p{0.10152\linewidth}|>{\columncolor{cyan!25}\raggedright\arraybackslash}p{0.10152\linewidth}|>{\raggedright\arraybackslash}p{0.14664\linewidth}|>{\columncolor{cyan!25}\raggedright\arraybackslash}p{0.17672\linewidth}|>{\raggedright\arraybackslash}p{0.10528\linewidth}|}
\hline
\textbf{Work} & \textbf{Assum-\newline ption} & \textbf{Commu\-nication} & \textbf{Storage} & \textbf{Total} & \textbf{Adv. model} & \textbf{Bound} & \textbf{Success} \\
\hline
$\mathsf{QPV}_{BB84}$\newline\cite{BCF+14} & None & $\Theta(n)$ & $\Theta(n)$ & $\Theta(n)$ & Preshared entanglement & $0$~\cite{TFKW13} & $\mathrm{negl}$ \\
\arrayrulecolor{cyan!25}\cline{2-2}\cline{5-5}\arrayrulecolor{black}\cline{6-8}\noalign{\vskip\arrayrulewidth}
& & & & & Preshared state size & Total register size\newline $\leq \alpha n$~\cite{TFKW13} & $\mathrm{negl}$ \\
\arrayrulecolor{cyan!25}\cline{2-2}\cline{5-5}\arrayrulecolor{black}\cline{6-8}\noalign{\vskip\arrayrulewidth}
& & & & & LOCC; preshared entanglement & $E_{\max}(\widetilde\Phi)$\newline $\leq n-\omega(\log^2 n)$~\cite{RG15} & $\mathrm{negl}$ \\
\hline
\mbox{$f$-routing}\newline\cite{BFS+13,KMS11} & None & $\Theta(n)$ & $\Theta(n+{}$\newline$\mathrm{Space}(f))$ & $\Theta(n+{}$\newline$\mathrm{Space}(f))$ & Quantum registers & Per side:\newline $\leq n/2-O(1)$, random $f$~\cite{BCS22} & $\leq1-\Omega(1)$ \\
\arrayrulecolor{cyan!25}\cline{2-2}\cline{5-5}\arrayrulecolor{black}\cline{6-8}\noalign{\vskip\arrayrulewidth}
& & & & & $\log_2$ Schmidt rank (OSP) & $\Omega(n)$ required for specified $f$~\cite{ACM25} & No OSP attack \\
\hline
$f$-BB84\newline\cite{BCS22} & None & $\Theta(n)$ & $\Theta(n+{}$\newline$\mathrm{Space}(f))$ & $\Theta(n+{}$\newline$\mathrm{Space}(f))$ & Quantum registers & Per side:\newline $\leq n/2-O(1)$, random $f$~\cite{BCS22} & $\leq1-\Omega(1)$ \\
\arrayrulecolor{cyan!25}\cline{2-2}\cline{5-5}\arrayrulecolor{black}\cline{6-8}\noalign{\vskip\arrayrulewidth}
& & & & & $\log_2$ Schmidt rank (perfect) & $\Omega(n)$ required for specified $f$~\cite{ACM25} & No perfect attack \\
\hline
\cite{JKP+22} & Banach-space conjecture & $\Theta(n^2)$ & $\Theta(n^2)$ & $\Theta(n^2)$ & Quantum registers; regular/general & Per side:\newline $\leq n^\alpha$, constant $\alpha>0$ & $\leq1-\Omega(1)$ \\
\hline
QROM QPV\newline\cite{Unr14} & QROM & $O(\lambda)$ & $O(\lambda)$ & $O(\lambda)$ & QROM queries & $\operatorname{poly}(\lambda)$ queries & $\mathrm{negl}$ \\
\hline
QROM CVPV\newline\cite{LLQ22} & QROM + LWE & $\operatorname{poly}(\lambda)$ & $\operatorname{poly}(\lambda)$ & $\operatorname{poly}(\lambda)$ & Poly-time & Unbounded preshared entanglement & $\mathrm{negl}$ \\
\hline
CVPV\newline\cite{LLQ22} & Subexp. LWE & $\operatorname{poly}$\newline$(\lambda,L)$ & $\operatorname{poly}$\newline$(\lambda,L)$ & $\operatorname{poly}$\newline$(\lambda,L)$ & Subexp.-time & Preprocessing qubits sent $\leq L$ & $\mathrm{negl}$ \\
\hline
\rowcolor{pink!35}
This work & None & $\widetilde O(m+{}$\newline$R_H)$ & $\widetilde O(m+{}$\newline$R_H)$ & $\widetilde O(m+{}$\newline$R_H)$ & Finite LOCC; quantum registers & Total size:\newline $\leq c\,\frac{sR_Hm}{\log_2^2\lambda}$ & $\mathrm{negl}$ \\
\hline
\end{tabular}%
}
\par\smallskip
\begin{minipage}{\linewidth}
\normalsize
\textbf{Notation.} Honest communication and storage count bits and qubits; ``Total'' sums both. $\mathrm{Space}(f)$ is classical space to represent and evaluate $f$; OSP means one-sided-perfect. ``None'' excludes computational, oracle, and conjectural assumptions. $\mathsf{QPV}_{BB84}$ uses parallel repetition (state-size bound: $0<\alpha<-\log_2(\cos^2(\pi/8))$); $f$-protocols show one execution. In~\cite{BCS22}, sequential repetition gives negligible error with proportionally more communication; in~\cite{Unr14}, both input lengths are $\Theta(\lambda)$. ``This work'' satisfies the parameter conditions of Theorem~\ref{thm:amplified_security-generalized}, with universal constant $c>0$. $\widetilde O$ absorbs logarithmic classical factors. The preprocessing bound~\cite{LLQ22} does not bound total quantum-register size.
\end{minipage}
\setlength{\abovecaptionskip}{4pt}
\caption{Prior QPV results and our finite-slot instantiation.}
\label{fig:related-work-comparison}
\end{figure}
    
\paragraph{AI Acknowledgment} All the main ideas of the paper were produced by the human authors, but ChatGPT 5.5 Pro, 5.6 Sol, 6 Astra, and Claude Opus 5.5 were used for proofreading, editorial changes, literature review for quantum position verification, position-based authentication and key exchange, as well as the standard quantum information tools used in the work, and for assisting with formalizing and verifying the arguments for the security proofs in \Cref{sec:position-authentication-appendix} and the corresponding parts of the technical overview.

\clearpage

\section{Technical Overview}
\label{sec:technical_overview}

In this section, we will recall the protocol and state the main ideas behind the security proof, i.e., the proof of \Cref{thm:main-result-informal}.

\paragraph{Protocol Reminder} We next recall the complete protocol before continuing with the proof. The protocol has $m$ rounds, and each round is divided into $s$ slots. 
\begin{itemize}
\item \textbf{Initialization} Verifier $V_0$ sends the prover an ordered list of $R_H$ independently sampled BB84 states. The verifiers share the labels and the classical description of these states among themselves.
\item \textbf{Choosing the challenge} In each of the $m$ rounds, the verifiers uniformly choose a challenge slot $i\in[s]$, a verifier $V_b$, and a state in the current list. Both verifiers remain idle in every slot other than $i$. In slot $i$, verifier $V_b$ asks the prover to return the selected state, while $V_{1-b}$ remains idle.
\item \textbf{Returning and testing the state} The honest prover removes the selected state from its list and immediately returns it to $V_b$. If the state at that index was $H^\theta\ket{x}$, the verifier, using the information $\theta,x$, measures the returned qubit in the basis $\theta$ and accepts the round if and only if the response arrives at the prescribed time and the measurement outcome is $x$.
\item \textbf{Replenishing the list} After each round, $V_0$ sends one fresh BB84 state and the verifiers add its classical description to their respective lists. Before the next round begins, the prover receives and appends this state, restoring a list of $R_H$ live qubits. The verifiers accept the protocol if and only if all $m$ rounds are accepted.
\end{itemize}

Clearly, the total quantum communication in an honest execution of the protocol is $R_H+2m$ quantum messages ($R_H+m$ qubits sent by the verifiers and $m$ qubits sent in response by the honest prover). Including classical communication and storage, the honest parties overall use $\widetilde{O}(R_H+m)$ total resources.

\paragraph{A Natural Attack} Consider two adversaries $A_0$ and $A_1$ on opposite sides of the claimed location. At the beginning of a slot, one side may hold all the stored states. That side waits until it either receives an instruction from its nearby verifier or can determine that this verifier is idle in the slot. In the idle case, the adversaries teleport these states to the other side so that they can still answer a possible instruction from the opposite verifier. Moving all $R_H$ states in this way consumes $R_H$ Bell pairs. However, since the adversaries do not know the selected slot in advance, the adversaries would have to do the teleportation in most of the slots if not all slots. Hence, they would need to start out with at least $\Omega(sR_H)$ Bell pairs before each round, which would overall require $\Omega(msR_H)$ Bell pairs when accounting for all $m$ rounds. Our proof will show that this is the best the adversaries can do in terms of total quantum register size, up to a polylogarithmic factor.

In contrast, during the idle slots, the honest parties sleep and do nothing, consistent with the ``I sleep, you work'' paradigm. The adversaries must keep moving the stored states only because they do not know which slot and verifier were selected.

\paragraph{Purified Execution for the Analysis} In our security proof, we use the following purified formulation of the protocol, which does not change the adversaries' view or success probability.
\begin{itemize}
\item \textbf{Initialization} Verifier $V_0$ prepares $R_H$ labeled Bell pairs, keeps one qubit from each pair as a reference qubit, and sends the other qubits toward the prover in the same order.
\item \textbf{Choosing the challenge} The verifiers choose the challenge slot, verifier, and index in the list as before. They remove the selected label $r$ from their list, and the selected verifier asks the prover to return the qubit labeled $r$. Both verifiers remain idle in every other slot.
\item \textbf{Returning and testing the state} The prover removes the requested qubit and returns it to the selected verifier. At the prescribed deadline, the verifiers measure the returned qubit and its reference qubit in the same basis, chosen uniformly from $\{X,Z\}$, and accept if and only if the outcomes agree. 
If no response arrives, they discard the reference qubit without measuring it and reject.
\item \textbf{Replenishing the list} After each round, $V_0$ prepares a fresh labeled Bell pair, keeps its reference qubit, and sends the other qubit toward the prover. The new pair is appended to the lists before the next round begins.
\end{itemize}
\paragraph{Proof Idea for the Special Case of One Qubit and Perfect Adversaries} 
 For simplicity, set $R_H=1$, let $O$ denote the verifier's reference qubit, and suppose the adversaries succeed with probability $1$. We group the adversarial registers to the left and right of the claimed location into $A_0$ and $A_1$, respectively.

Since the adversaries have perfect success probability, when either verifier asks for the qubit inside a slot, the adversaries on the corresponding side should be able to recover a qubit whose joint state with $O$ is a Bell state. Note that local recoverability on one side may crucially depend on the classical messages from the other side. We therefore use the notion of \emph{extended states} of the sides. The extended state of side $A_b$ at some $t$ is the local state of the corresponding side at time $t$, together with the unread classical messages already generated by the other side $A_{1-b}$ at time $t$, that would arrive before $A_b$'s deadline to send a timely response.

For this intuitive argument, suppose that every slot, regardless of the verifiers' instructions, contains two points in time, one at which $A_0$ and the other at which $A_1$ can recover a qubit maximally entangled with $O$ from its extended states. In \Cref{sec:invoking-entaglement-consumption-theorem}, we explain how to circumvent this local recoverability assumption. Relabel the sides if necessary, so that $A_0$'s local recoverability precedes $A_1$'s local recoverability. Call $A_0$'s earlier recovery state the \emph{initial state} and $A_1$'s later recovery state the \emph{final state}. 
We can conclude the following about the entanglement of the initial and final states, respectively.

\begin{itemize}\label{pg:proof-idea}
\item \textbf{At the initial state,} $O$ is maximally entangled with $A_0$, i.e., $O$ forms a Bell pair with one of the qubits in $A_0$ (up to a local recovery map on $A_0$), and hence $O$ and $A_1$ must be uncorrelated due to monogamy of entanglement. Therefore, adding $O$ on the $A_0$ side of the cut $A_0:A_1$ does not change the squashed entanglement across the cut, and hence $\Esq(OA_0:A_1)_{initial}=\Esq(A_0:A_1)_{initial}$.
\item \textbf{During the evolution,} the squashed entanglement across $OA_0:A_1$ cannot increase because the evolution is LOCC and does not act on $O$.
\item \textbf{At the final state,} monogamy of entanglement implies that
\[\Esq(OA_0:A_1)_{final}\geq \Esq(A_0:A_1)_{final}+\Esq(O:A_1)_{final},\] which is equal to $\Esq(A_0:A_1)_{final} + 1$ as the squashed entanglement across $O:A_1$ is one, because $O$ is now maximally entangled with $A_1$.
\end{itemize}

Hence, combining the three constraints, it is easy to see that the entanglement across $A_0:A_1$ must have decreased by at least one unit from the initial to the final state.

\paragraph{Extending to the General Case} In the perfect adversary case, even if we had $R_H$ qubits and all $R_H$ Bell-pair halves are simultaneously recoverable in the initial and final states, the same idea as above can be extended to yield an entanglement decrease of $R_H$.
Since entanglement is always bounded by the size of the quantum registers, this lower bound on the entanglement consumption forms a lower bound for the size of quantum storage. 

For the proof of the most general case, however, we need an imperfect version of the above arguments in which only one uniformly random qubit index is tested, and more importantly, the local recovery maps succeed with high probability \emph{on average}, rather than perfectly. This required lower bound on entanglement consumption can be abstracted out as the following quantum-information guarantee.

\begin{theorem}[Entanglement consumption theorem]
\label{thm:informal-entanglement-consumption}
Suppose there are $R_H$ reference qubits $O_1,\ldots,O_{R_H}$, and initially\footnote{We also need that the initial and final global states should be classically flagged mixtures of pure states, the need for which will be made explicit in \Cref{sec:entanglement-consumption-thoerem-informal}.} one side of the cut, register $A_0$, can locally recover, for a uniformly random index $r\in[R_H]$, a qubit whose joint state with $O_r$ has high average Bell-pair fidelity\footnote{Here, high average Bell-pair fidelity means that the joint state of the recovered qubit and the corresponding reference qubit has high fidelity with a Bell pair, averaged over the random index, and it represents high success probability of the adversaries in the purified version of the protocol. By high, we mean at least some constant threshold greater than $1/2$, sufficiently close to $1$. The threshold needs to be higher than $1/2$ because, by measuring the qubit in the computational basis and sharing the classical outcomes among themselves, two unentangled adversaries can satisfy the theorem's hypothesis with an average Bell-pair fidelity of $1/2$. See the formal version, \Cref{cor:k-extraction-entanglement-consumption-distributed} for the exact constant threshold.}, averaged over $r$. After an LOCC evolution, suppose the other side, $A_1$, can locally recover such a qubit for a uniformly random $r$ with similar average fidelity. Then the entanglement between $A_0$ and $A_1$ decreases by $\Omega(R_H)$.
\end{theorem}
The theorem is formally shown in \Cref{cor:k-extraction-entanglement-consumption-distributed}.

\paragraph{Choice of Entanglement Measure}
To formally state and prove the entanglement consumption theorem (\Cref{thm:informal-entanglement-consumption}), we need an entanglement measure satisfying the properties used implicitly in the perfect-case argument. Namely, it should satisfy
\begin{enumerate}
    \item Monogamy of entanglement inequality~\cite{CKW20},
    \item Monotonicity under LOCC.
\end{enumerate}
Most common entanglement measures violate at least one of the first two criteria. For instance, entanglement of formation does not satisfy the above monogamy inequality~\cite{CKW20}, while the R\'enyi-$2$ entanglement entropy of pure states can increase on average under LOCC~\cite{KPM20}. We therefore use squashed entanglement, a conditional-mutual-information-based measure that satisfies both criteria, which is defined for any bipartite state $\rho_{AB}$ as
\begin{equation*}
E_{\rm sq}(A:B)_\rho:=\frac12\inf_{\rho_{ABE}}I(A:B\mid E)_\rho,
\end{equation*}
where the infimum is over all extensions $\rho_{ABE}$ of $\rho_{AB}$.

\subsection{Using the Entanglement-Consumption Theorem in the Protocol}\label{sec:invoking-entaglement-consumption-theorem}

We first complete the security proof of our protocol assuming the entanglement consumption theorem (\Cref{thm:informal-entanglement-consumption}), and then return to the proof of the theorem in \Cref{sec:entanglement-consumption-thoerem-informal}. 
Without loss of generality, we group all adversarial registers on the left side of the claimed location into $A_0$ and those on the right side into $A_1$ (if a side contains no adversary, the size of the corresponding register is treated as $0$). We also use the convention that when a verifier does not send any instruction in a slot, we call it the idle instruction. 

\paragraph{Obstacle to Local Recoverability Inside a Single Slot} 
It might appear that once the entanglement consumption theorem (\Cref{thm:informal-entanglement-consumption}) is proved, we can invoke it inside every slot of a round that the adversaries win with high probability \footnote{i.e., with probability at least some constant threshold determined by the formal statement of the entanglement consumption theorem given in \Cref{cor:k-extraction-entanglement-consumption-distributed}.}. To do so, in any such slot, we need to identify two distinct time points: one at which $A_0$'s extended states can locally recover the qubit at a random index with high probability, and the other at which $A_1$'s extended state can do the same. Even in the preceding proof of the special case of perfectly successful adversaries, we made this assumption to complete the proof. However, this intuition is not entirely correct, as shown in the following two scenarios.

\begin{enumerate}
\item \textbf{Only one adversarial side may recover in a slot}\label{it:mutual-dependency} 
 Suppose the adversarial strategy is such that when $A_0$ receives the idle instruction from $V_0$, it measures its local state and sends the classical outcome to $A_1$. This classical message allows $A_1$ to locally recover if $V_1$ requests the qubit at a random index, but the measurement ruins the local recoverability of $A_0$'s extended state for the entire slot. However, if $V_0$ instead requests the qubit at a random index, the roles are reversed, i.e.,  $A_0$ preserves its state while $A_1$ measures and sends the information needed by $A_0$. Thus, the operation enabling recovery on one side may destroy recoverability on the other. Hence, under this strategy, which could still be perfect, the local recoverability guarantees for $A_0$ and $A_1$ never arise inside the same slot. 
\item \textbf{There may exist a side that can never locally recover in any dummy slot}\label{it:no-recovery-guarantee-at-physical-time} Suppose both verifiers are idle in a slot, i.e., it is a dummy slot. Under the adversarial strategy mentioned in the previous scenario, each side may measure its local state and send the classical outcome to the other side, since that outcome would be needed if the other verifier had sent the return instruction. Suppose $A_0$ is closer to $V_0$ compared to the distance between $A_1$ and $V_1$. Hence, before $A_1$ receives the idle instruction from $V_1$ and measures its state, $A_0$ receives its instruction from $V_0$ and measures its state. Before $A_0$'s measurement, since $A_1$ has not yet received its instruction and measured, $A_0$'s extended state lacks the classical information needed for recovery, and after $A_0$'s measurement, $A_0$ has ruined its local quantum state. Hence, under this strategy, there is no physical time in a dummy slot at which $A_0$ can recover locally from its extended state. Note that showing entanglement consumption across or inside dummy slots is the main source of our polynomial separation, as they are exactly those slots in which adversaries work but the honest parties sleep.
\end{enumerate}

\paragraph{Resorting to Entanglement Consumption Across Multiple Slots} Since the issue of having only a single adversarial side per slot that can locally recover (see \Cref{it:mutual-dependency}) hinders showing entanglement consumption across the same slot, we switch to entanglement consumption across different slots. 
Assume for now that the slot intervals do not overlap, i.e., the next slot starts after the previous one has ended. We will later show a reduction to remove this assumption.\footnote{We note that allowing for overlapping slot intervals is essential for the clocks-based instantiation of our result in the discrete time model (\Cref{thm:clock-main-result-informal}), as well as the instantiation in the continuous time model (\Cref{thm:continuous-main-result-informal}) because both these instantiations rely on the fact that consecutive slots can begin arbitrarily close to one another.} Under the non-overlapping slot assumption, the classical messages needed by $A_0$ in the first slot do not require the classical messages from $A_1$ arising in the next slot. Thus, by considering two different slots, we break the barrier of mutual dependency of the two sides for local recoverability as described in \Cref{it:mutual-dependency} in
the preceding paragraph.

To identify the suitable slots, we next define and count the set of \emph{good slots} inside a highly successful round.
\begin{itemize}
\item We call a round \emph{high-success} if the adversaries win it with probability at least $1-10^{-7}$.\label{pg:high-success-round}
\item Within such a round, we call the pair $(i,b)$ \emph{good} if, conditioned on the verifiers choosing slot $i$ and verifier $V_b$, the adversaries win with probability at least $1-\gamma$ where $\gamma=10^{-5}$.
\item We call slot $i$ \emph{good} if both $(i,0)$ and $(i,1)$ are good.
\end{itemize}
A direct application of Markov's inequality shows that every high-success round contains at least $N\geq0.98s$ good slots. Write the good slots $t_1<t_2<\cdots <t_N$. Next, in order to establish local recoverability, we note that for any good slot $t_j$ with $j<N$ and $b\in \{0,1\}$, the farther adversarial side $A_{1-b}$ does not have enough time to learn whether $V_b$ sent a non-idle instruction and send a response to $V_b$ that reaches in time. Hence, we show that conditioned on verifier $V_b$ in slot $t_j$ sending a non-idle return request, the nearby adversarial side $A_b$ must send an accepting response with large probability; in particular, $A_b$ must recover and send a qubit that is accepted with probability at least $1-2\gamma$, see the proof of \Cref{thm:amplified_security-generalized} (Step 3) for more details.\footnote{For the analysis, if no response is provided by $A_b$, we define its response qubit to be the maximally mixed state $I/2$, which can always be prepared locally, irrespective of $A_b$'s input state.}   
Though shifting to analyzing entanglement consumption across consecutive good slots circumvents the first issue (\Cref{it:mutual-dependency}), the second issue (\Cref{it:no-recovery-guarantee-at-physical-time}) discussed in the previous paragraph still persists. In particular, there may exist an adversarial side $A_b$ that never locally recovers at any physical time of any dummy slot.

\paragraph{Reordered Adversarial Evolutions} To address the remaining issue (\Cref{it:no-recovery-guarantee-at-physical-time}), consider a thought experiment where we pause one side while allowing the other side to produce the classical messages needed for recovery. This reordering is only for the analysis and need not respect the original timing constraints. Fix a good slot $t_j$ with $j<N$, and suppose that the actual challenge slot is later than $t_j$, so both verifiers are idle in this slot.
\begin{itemize}
\item We first let the adversaries perform the operations and classical communication that do not depend on either instruction in the current slot. Fix $b\in\{0,1\}$. We pause $A_b$ just before it receives its instruction $\ell_b$, while $A_{1-b}$ continues: it receives its instruction $\ell_{1-b}$ (which is idle) from $V_{1-b}$ and performs the operations that do not need new information from the paused side, potentially generating classical messages for $A_b$.
\item Let $\tau_{j,1-b}$ denote this timestamp in the reordered evolution, when $A_{1-b}$ has acted on seeing its instruction, potentially generating the classical messages needed for $A_b$'s local recovery, while $A_b$ is still paused. The operations we have reordered act on disjoint registers and do not use each other's outputs. Hence, after resuming the paused operations with the original instructions, we obtain the same final state as in the original execution.
\item At $\tau_{j,1-b}$, $A_{1-b}$ has acted on the idle instruction, while $A_b$ has not yet received its instruction. If we instead give $A_b$ a return instruction for any qubit index, it can produce the corresponding response using only its extended state at this timestamp. By the goodness of slot $t_j$ and the arguments (about the response of the nearby side $A_b$) established in the preceding paragraph, the success probability of $A_b$'s response in the purified protocol's test is at least $1-2\gamma$, averaged over the uniformly selected qubit index. Translating the success probability to Bell-pair fidelity, we conclude that the recovered qubit and its corresponding reference qubit have average Bell-pair fidelity (with the same averaging) at least $1-4\gamma$, thereby establishing the required local recoverability hypothesis of  the entanglement consumption theorem (\Cref{thm:informal-entanglement-consumption}) for the extended state of $A_b$ at $\tau_{j,1-b}$.
\item Reversing the roles of the two sides gives an alternate reordered evolution and a timestamp $\tau_{j,b}$ at which the extended state of $A_{1-b}$ satisfies the same local recoverability guarantee, and hence the local recoverability hypothesis of \Cref{thm:informal-entanglement-consumption}. We note that the timestamps $\tau_{j,1-b}$ and $\tau_{j,b}$ belong to two different reordered evolutions of the same slot and are chronologically incomparable, but across slots, they can be compared, i.e.,  for every $j\in [N-2]$, $\tau_{j,b}< \tau_{j+1,0}$ and $\tau_{j,b}<\tau_{j+1,1}$. 
\end{itemize}

\paragraph{Entanglement Consumption Over a High-success round} Since we have local recoverability for the extended states of $A_0$ and $A_1$ in two different timestamps of two different good slots, we can invoke the entanglement consumption theorem (\Cref{thm:informal-entanglement-consumption}) that ensures that there must be a drop of $\Omega(R_H)$ in squashed entanglement across the two timestamps, and hence across the corresponding good slots. Note that the local recoverability guarantees hold for the intermediate states only conditioned on the real challenge slot occurring after the two consecutive good slots.

\begin{itemize}
\item We use the first $N-1$ good slots $t_1<t_2<\cdots<t_{N-1}$. We choose the reordering with timestamp $\tau_{j,0}$ when $j$ is odd and $\tau_{j,1}$ when $j$ is even. Since the slots do not overlap, these choices give a single evolution with a chronologically ordered list of timestamps $\{\tau_{1,0}< \tau_{2,1}<\cdots <\tau_{(N-1,N\mod 2)}\}$ with $N-2$ transitions between consecutive chosen timestamps.\footnote{We omit the last good slot because the preceding same-side response argument requires a later good slot.} Moreover, the locally recoverable side alternates between $A_1$ and $A_0$, i.e., for odd $j\in [N-1]$, $A_1$ can locally recover at timestamp $\tau_{j,0}$ and for even $j\in [N-1]$, $A_0$ can locally recover at timestamp $\tau_{j,1}$.
\item For any $j\in [N-2]$, if the real challenge slot occurs after slot $t_{j+1}$, the evolution across the transition from the chosen state in $t_j$ to that of $t_{j+1}$ is LOCC and leaves the reference qubits and their order unchanged. Hence, by \Cref{thm:informal-entanglement-consumption}, we get an $\Omega(R_H)$ entanglement drop across each such transition, conditioned on the challenge slot occurring after the two consecutive good slots.
\item There are at least $N-j-1$ good slots after $t_{j+1}$. Hence, averaging over the choice of the real challenge slot, the transition from slot $t_j$ to slot $t_{j+1}$ contributes to an overall drop in entanglement of at least $(N-j-1)/s$ times $\Omega(R_H)$. Summing over these $N-2$ transitions, we get an entanglement drop of at least $\Omega(sR_H)$ across all the slots of every high-success round. 
\item There is a final missing detail: we need to account for possible increases caused by the verifier's action in between the rounds, such as testing the returned qubit. However, since the verifiers have at most $R_H$ reference qubits, we show that this only incurs an additive loss of $O(R_H)$ units in entanglement consumption (see \Cref{lem:verifier-entanglement-accounting}). 
 Hence, despite the loss, we still get a lower bound of at least $\Omega(sR_H)$ units of squashed entanglement.
\end{itemize}

\paragraph{From Adversaries with Many High-success Rounds to General Adversaries} The preceding argument would complete the proof if the original execution contained sufficiently many high-success rounds. However, non-negligible overall success probability does not imply the existence of even one such round. This is the same issue that motivated our use of quantum-register size as the adversarial resource measure on \Cpageref{pg:bounded-register-size}.
\begin{enumerate}
    \item The adversaries may begin with a mixture that contains a highly entangled state with small inverse-polynomial probability and a product state otherwise.
    \item Conditioned on the highly entangled state, every round may be won with high probability, while the unconditional success probability of each round can remain far below the high-success threshold.
\end{enumerate}
  We therefore cannot apply the preceding argument directly to any round of the original execution.

\paragraph{Solution: A Block of Consecutive High-success Rounds via Post-selection} Our solution arises from the same motivating example. The key observation is that although the original state does not have large entanglement, there should be a post-selected state (in fact, after the first round itself in the motivating example) that does, and the adversarial registers must be large enough to hold even this state. Formally, we do the following.
\begin{enumerate}
    \item Let $L:=\lceil\log_2^2(\secparam)\rceil$ and partition the $m$ rounds into $L$ consecutive blocks\footnote{whose lengths differ by at most one.}. Assuming the adversaries win the full protocol with probability at least $(1-10^{-7})^L$, there must exist a block $\mathcal{B}$ such that the probability of winning all rounds in $\cal B$, conditioned on success in all rounds preceding $\cal B$, is at least $1-10^{-7}$.
    \item We consider the post-selected state $\rho_{\rm post}$ obtained after conditioning on success on all rounds preceding $\cal B$, and run the protocol rounds in $\cal B$ on such a post-selected state.
    \item Since we keep all classical outcomes, the post-selected state at the beginning of the block is also a classically flagged mixture of pure states. 
    \item Also, since the probability of winning any individual round in $\cal B$ is at least the probability of winning all rounds in $\cal B$, every round in this execution, starting from the post-selected state, is a high-success round. 
    \item Hence, by the argument in the preceding paragraphs, the decrease in squashed entanglement in each round must be more than $\Omega(sR_H)$. 
    \item Note that the block $\cal B$ contains at least $\lfloor m/L\rfloor$ rounds, where $L:=\lceil\log_2^2(\secparam)\rceil$. Summing the decreases over all rounds in the block and using the non-negativity of squashed entanglement yields
\begin{equation}
E_{\rm sq}(A_0:A_1)_{\rho_{\rm post}}>\left\lfloor\frac{m}{\lceil\log_2^2(\secparam)\rceil}\right\rfloor \Omega(sR_H)=\Omega\left(\frac{smR_H}{\log_2^2(\secparam)}\right).
\end{equation}

\end{enumerate}

Since squashed entanglement is at most the number of qubits in either adversarial group, the fixed quantum registers supporting the post-selected state must contain at least this many qubits. Finally, since post-selection does not change the number or dimensions of the adversarial coalition's quantum registers, the total adversarial quantum-register size is $\Omega(smR_H/\log_2^2(\secparam))=\widetilde{\Omega}(smR_H)$.

\paragraph{From Non-overlapping to Overlapping Slots} The arguments in the preceding paragraph assume that consecutive slot intervals do not overlap, i.e., operations associated with each slot begin after all operations of the preceding slot have ended.\footnote{The arguments assume that every intermediate auxiliary state obtained from an earlier good slot comes before an intermediate auxiliary state obtained from a later good slot.} However, as described before, for both \Cref{thm:clock-main-result-informal,thm:continuous-main-result-informal}, it is crucial that consecutive slots can begin arbitrarily close to one another, which may potentially result in overlapping slot intervals.

To extend the result to overlapping slots, we consider a non-overlapping execution obtained by increasing the time between the start times of consecutive slots by a fixed positive amount large enough to make the slots non-overlapping, postponing later rounds and the actions between them accordingly.
For the reduction, we simulate the original attack against the overlapping slot execution, as follows.
\begin{itemize}
    \item The reduction runs the original strategy and stores classical information from earlier slots until its scheduled use, and delays each operation until its required information and quantum registers are available. It uses the same local and outgoing quantum registers, preserving their order of use and release.
    \item Crucially, we show that responses in an originally accepted execution can be reproduced at the new deadlines of the non-overlapping execution (\Cref{sec:proof-overlap-to-nonoverlap}). This is based on the fact that an adversary that has learned the challenge knows that the verifiers' later slots in that round are idle. An adversary that has not yet learned the challenge may still obtain new information from the other verifier's later slots, but that information cannot affect a response that reaches the requesting verifier by the original deadline.
\end{itemize}

Clearly, the rescheduled strategy simulates the original strategy as per the time constraints of the non-overlapping execution, preserving every accepting execution, while using exactly the same quantum registers. Hence, the rescheduled strategy in the non-overlapping execution has success probability at least that of the original strategy in the original overlapping execution.  See \Cref{prop:overlap-to-nonoverlap} and its proof in \Cref{sec:proof-overlap-to-nonoverlap} for more details.

\subsection{Proving the Entanglement-Consumption Theorem}
\label{sec:entanglement-consumption-thoerem-informal}
We now explain the proof of \Cref{thm:informal-entanglement-consumption}. As before, we will begin with $R_H=1$ but now in the imperfect adversaries case. Let $O$ be the reference qubit, and suppose $A_0$ can \emph{approximately} recover the other Bell-pair half initially, while $A_1$ can approximately recover it after an LOCC evolution. To extend the proof template from the perfect adversaries case on \Cpageref{pg:proof-idea}, we need three relations, up to some small additive error.
\begin{enumerate}
    \item An upper bound on $E_{\rm sq}(OA_0:A_1)_{\rm initial}$ for the initial state in terms of $\Esq(A_0:A_1)_{\rm initial}$,
    \item Monotonicity of $E_{\rm sq}(OA_0:A_1)$ from the initial to final states,
    \item A lower bound on $E_{\rm sq}(OA_0:A_1)$ for the final state in terms of $\Esq(A_0:A_1)$.
\end{enumerate}

It is easy to see that the last two relations hold.
In particular, LOCC monotonicity of squashed entanglement gives the second relation, i.e.,
\begin{equation*}
E_{\rm sq}(OA_0:A_1)_{\rm final}\leq E_{\rm sq}(OA_0:A_1)_{\rm initial}.
\end{equation*} 
Next, monogamy of squashed entanglement gives
\begin{equation*}
E_{\rm sq}(OA_0:A_1)_{\rm final}\geq E_{\rm sq}(A_0:A_1)_{\rm final}+E_{\rm sq}(O:A_1)_{\rm final}.
\end{equation*}
Since $A_1$ can locally recover a qubit (nearly) maximally entangled with $O$, the state of $O$ and the recovered qubit has high\footnote{By high, we mean at least some constant threshold, sufficiently close to $1$. See the formal version, \Cref{cor:k-extraction-entanglement-consumption-distributed}.} fidelity with a Bell pair, and hence $E_{\rm sq}(O:A_1)_{\rm final}$ is close to one. Combining it with the last equation gives the last relation. 

Thus, if we can show that the first relation also holds, i.e., 
\begin{equation*}
E_{\rm sq}(OA_0:A_1)_{\rm initial}\leq E_{\rm sq}(A_0:A_1)_{\rm initial} + \textrm{error term},
\end{equation*}
with a (small enough) constant error term, then the three inequalities give an $\Omega(1)$ decrease in $E_{\rm sq}(A_0:A_1)$.

\paragraph{Obstacle in Proving the Upper Bound with the Initial State} The initial recovery guarantee says that $A_0$ can produce a qubit that is nearly maximally entangled with $O$. Intuitively, this means that the Bell-pair half paired with $O$ is already present on the $A_0$ side, so adjoining $O$ to that side should add very little entanglement across the cut. The difficulty is to turn this intuition into a dimension-independent upper bound on $E_{\rm sq}(OA_0:A_1)$. In particular, 

\begin{enumerate}
    \item \textbf{Lack of chain rule for squashed entanglement} Note that if squashed entanglement satisfied a chain rule, then getting this upper bound would be immediate.\footnote{Indeed, if we had measures satisfying the chain rule, such as mutual information, then such an upper bound would be immediate, but these measures capture both classical correlations and entanglement, and are therefore unsuitable for our purpose because the adversaries may exchange and store arbitrarily large classical records without sharing any additional entanglement.} However, squashed entanglement, like most suitable entanglement measures, does not satisfy a chain rule. 
    \item \textbf{Difficulty in obtaining the required upper bound by continuity} A natural approach would be to use the high Bell-pair fidelity to compare the squashed entanglement of the actual state with that of an ideal state containing a perfectly recovered Bell pair, via a continuity bound. However, converting this fidelity into a bound on squashed entanglement, using the Fannes--Audenaert continuity bound~\cite{audenaert2007sharp}, introduces an error proportional to the logarithm of the adversarial-register dimension. Since this dimension is exactly what we are trying to lower-bound at the end, this would lead to a circular argument.
    \item \textbf{Locking of quantum correlations}  In fact, in the most general mixed-state setting, no dimension-independent upper bound of the required form is true. This is due to a phenomenon called the \emph{locking of quantum correlations}~\cite{christandl2005uncertainty}. In particular, there are states called \emph{flower states} on arbitrarily large registers, such that removing a single qubit can change the squashed entanglement by an amount that grows with the dimensions of the surrounding registers. Translated to our setting, removing $O$ from the entanglement bipartition $OA_0:A_1$ may introduce an additive error that depends on the register sizes of $A_0$ and $A_1$. In other words, the difference
\begin{equation*}
E_{\rm sq}(OA_0:A_1)_{\rm initial}-E_{\rm sq}(A_0:A_1)_{\rm initial}
\end{equation*}
can grow linearly in the total number of qubits in the registers $A_0$ and $A_1$. 
\end{enumerate}

\paragraph{Solution: Classically Flagged Mixtures of Pure States} Our way around this problem begins with the observation that pure states (i.e., the global state on $OA_0A_1$ is pure) do not exhibit the locking phenomenon. Indeed, we show that the required dimension-independent upper bound holds when the state on $OA_0A_1$ is pure (formally shown in \Cref{lemma:pure-state-flow}) by exploiting the fact that squashed entanglement can be expressed in terms of local von Neumann entropies for pure states.

Unfortunately, this does not entirely solve the problem because the states that appear during the protocol are not necessarily pure. The adversaries measure quantum registers, exchange classical messages, and store transcript information. One might try to purify all these operations coherently, but doing so would introduce additional quantum registers and, more importantly, would no longer preserve the LOCC model whose monotonicity is essential for adding the entanglement decreases throughout the security proof.

\paragraph{Classically Flagged Mixture of Pure States} We therefore take an intermediate route between pure and arbitrary mixed states. We keep the classical information explicit and work with classically flagged mixtures of pure states. In particular, we show that the global state has the form
\begin{equation*}
\rho_{\rm global}=\sum_c p_c\ket{c}\bra{c}_C\otimes\ket{\psi_c}\bra{\psi_c}_Q,
\end{equation*}
where $C$ and $Q$ denote the classical and quantum registers, respectively.\footnote{A register is classical when it is classical in the complete joint state: it represents a classical random variable that may be classically correlated with the other registers but is not quantumly entangled with them; see \Cref{def:classically-flagged-mixture-pure-states}.} After fixing the classical value $c$, the state on the quantum registers is the pure state $\ket{\psi_c}$. We may therefore apply the pure-state entanglement upper bound separately to each $\ket{\psi_c}$. The direct-sum identity for squashed entanglement in \Cref{lemma:direct-sum} then allows us to average these inequalities over $c$, yielding the desired upper bound for the complete initial state; see \Cref{lemma:bounding-initial-state-k-extraction} for more details.

\paragraph{Transferring Several Qubits With an Average Recovery Guarantee} The discussion above only considers the case of one reference qubit, whereas \Cref{thm:informal-entanglement-consumption} concerns $R_H$ reference qubits and a uniformly random qubit index $r\in[R_H]$. The hypotheses of the theorem do not provide a joint procedure to recover all $R_H$ qubits in both the initial and final states.\footnote{Even if we manage to identify a single index with sufficiently good recovery at both the initial and final states and then invoke the one-qubit case proven above, the resulting lower bound would still lose the factor of $R_H$.} Instead, each index may have a different local recovery map. The maps need not commute, and recovery is guaranteed only on average over random $r$. 

To extend our arguments for the single-qubit case to the above setting, we had to develop additional tools such as generalized additivity of mutual information (up to a term due to the correlations among the reference qubits, as formalized in \Cref{lemma:superadditivity}) to prove the upper bound required for the stronger theorem. In particular, we do the following.
\begin{itemize}
    \item We start with the previous technique of combining the pure-state upper bound in \Cref{lemma:pure-state-flow} for each classical record and the averaging argument using the direct-sum identity (formally shown in \Cref{lemma:direct-sum}).
    \item We then use the generalized additivity of mutual information (formally shown in \Cref{lemma:superadditivity}) to relate the mutual information involving $O=O_1\cdots O_{R_H}$ to a sum of separate mutual information terms for each individual subsystem up to a correlation term.
    \item Finally, the additional correlation term cancels with the entropy terms from the pure state bound, while the errors arising from comparing each approximately recovered pair with an ideal Bell pair can be proven to be small by the Fannes-Audenaert continuity bound~\cite{audenaert2007sharp} since each individual system involves only two qubits.
    \item Averaging and summing over all indices yield the required upper bound (\Cref{lemma:bounding-initial-state-k-extraction}).
\end{itemize}

\subsection{Lifting the Protocol from 1-D Static Verifiers to Higher-Dimensional Mobile Verifiers}

The remaining step is to lift the one-dimensional static-verifier result to mobile verifiers in higher dimensions. We show a generic reduction, for which we first translate and rotate the execution so that the claimed location is at the origin and the two verifiers lie at $(-R,\vec0)$ and $(R,\vec0)$. A natural attempt at reduction would be to project every adversarial location onto the verifier line. Projection does not increase distances, so messages could be delivered early and stored until the original use times. The problem is that an adversary away from the claimed location may project exactly onto it.

We circumvent this problem by slightly modifying the projection as follows.
\begin{enumerate}
    \item We map any adversarial location $(x,\vec y)\in\mathbb R\times\mathbb R^{d-1}$ to $F_\varepsilon(x,\vec y):=x+\varepsilon\|\vec y\|_2^2$.
    Since the adversarial coalition is finite, we can choose a sufficiently small nonzero $\varepsilon$ such that no adversary is mapped to the origin and none of the distances used in the simulated execution increase. 
    \item We place a one-dimensional adversary at the image of each original adversary and let it perform the same local operations. Messages received early by an adversary are kept until their original use times. 
    \item If the distance from a responding adversary to the recipient verifier becomes shorter, the response remains in the same outgoing quantum register and is sent later, so that it reaches the verifier at the original time. 
\end{enumerate}

Clearly, the resulting 1-D adversarial coalition uses the same quantum registers and reproduces every response that affects acceptance at the right time; hence, it has at least the same success probability as the original higher-dimensional coalition. Thus the 1-D theorem (\Cref{thm:main-result-informal}) gives the higher-dimensional lower bound $\Omega(smR_H/\log_2^2(\secparam))=\widetilde{\Omega}(smR_H)$.

\subsection{Position-Based Authentication and Quantum Key Exchange}
\label{sec:position-authentication-overview}
 
In position-based authentication, an honest prover at the claimed location sends a classical message over a public channel in the presence of adversaries, and verifiers should accept it if it is indeed a message sent by the prover. Hence, unlike our QPV protocol, the adversaries can now tamper with communication between the verifiers and an honest party at the claimed location, and use the prover's responses to convince the verifier to accept their message. 
 
\paragraph{A Natural Template}
A natural template for extending our ``I sleep you work'' paradigm to a position-based authentication protocol is the following.
\begin{itemize}
  \item \emph{The prover sends its message $M$ to the verifiers over an untrusted classical channel.}
  \item As before, the verifiers send the prover independent BB84 states, which the prover stores \emph{in $W$ groups of a few qubits each, called tuples and indexed by $r\in[W]$}. 
  \item As before, the protocol runs in rounds of $s$ possible challenge slots. In each round, the verifiers choose a random slot, a random verifier, and a random \emph{tuple index $r$}, and the chosen verifier sends $r$ in that slot.
  \item The prover returns \emph{one qubit of the $r^{th}$ tuple, chosen by its current message bit, and discards the rest of the qubits in the $r^{th}$ tuple}.
  \item As before, the verifiers test the reply, now against the qubit \emph{that the message they received prescribes}. 
  \item After each round, the verifiers then replace \emph{the entire $r^{th}$ tuple of qubits}.
\end{itemize}
The intuition is that if the message $M'$ the verifiers receive differs from the prover's message $M$, the prover may return a different qubit from what the verifiers expect, and the test may fail. The remaining task is to choose the response rule, i.e., which qubit of the $r^{th}$ tuple the prover returns on an instruction $r$ from a verifier, for each message-bit value.
 
\paragraph{Obstacle: Attacks by Swapping Qubits}
Consider tuples of two qubits $(D_{0,r},D_{1,r})$, $r\in[W]$. The most natural rules for responding to an
instruction $r$ from a verifier are the following.
\begin{enumerate}
  \item\label{pba:rule:fixed} \emph{Fixed:} return $D_{0,r}$ if the message bit is zero and $D_{1,r}$ if it is one.
  \item\label{pba:rule:random} \emph{Random:} the request contains a fresh random bit $e$; return $D_{e,r}$ if the message bit is zero and $D_{1-e,r}$ if it is one.
\end{enumerate}
However, both rules can be attacked for every $W$, without entanglement. In particular, the adversaries can change the prover's message $M$ to its bitwise complement $\overline{M}$ in the first step, and still convince the verifiers to accept $\overline{M}$ as follows.  For Rule~\ref{pba:rule:fixed}, they can just swap $D_{0,r}$ and $D_{1,r}$ on their way to the prover, for every $r$, so that the prover's reply for a zero is exactly the state the verifiers expect for a one. Similarly, for Rule~\ref{pba:rule:random}, the adversaries can flip the verifier's request $e$, which turns the prover's reply for message bit zero into the reply for message bit one. Note that larger tuples of qubits or making the requests $e$ longer do not help, since suitable permutations of the incoming qubits or request values can swap the replies for zero and one.
 
\paragraph{Solution: Combining the Two Rules}
\label{pba:sec:rationale}
Our observation is that using both rules together thwarts these swapping attacks. In particular, we use two consecutive testing rounds for each message bit, each with an independently sampled verifier,
slot, tuple index $r$, and bit $e$. In the first round, message bit zero follows Rule~\ref{pba:rule:fixed}
and message bit one follows Rule~\ref{pba:rule:random}; in the second round, the roles are exchanged. The intuition is that by carefully choosing the exact fixed (Rule~\ref{pba:rule:fixed}) and random rules (Rule~\ref{pba:rule:random}), swapping qubits can no longer make both
rounds of a bit pass. In particular, with this response rule, if a prover's message bit is flipped, one of its corresponding two rounds will have the prover following Rule~\ref{pba:rule:fixed} while the verifier follows
Rule~\ref{pba:rule:random}, which should intuitively result in rejection assuming the adversaries have no preshared entanglement.\footnote{With preshared entanglement, the adversaries can just store the BB84 qubits sent by the verifier and keep teleporting among themselves to convince the verifiers; however, we will show, as is the case in the QPV protocol, this would force the adversaries to have large quantum storage, proportional to $s$, the number of slots per round.}
 
Concretely, each tuple has four qubits $(D_{00,r},D_{01,r},D_{10,r},D_{11,r})$, so $R_H=4W$. On an
instruction $(r,e)$, the prover returns the qubit of the $r^{th}$ tuple according to its message bit, given by the table below, where we omit
the index $r$ and write $\bar e=1-e$.

\begin{center}
\begin{tabular}{@{}ccc@{}}
\hline
Bit & First round & Second round\\
\hline
$0$ & $D_{00}$ & $D_{0e}$\\
$1$ & $D_{1\bar e}$ & $D_{11}$\\
\hline
\end{tabular}\label{fig:rule}
\end{center}
For example, if the prover's bit is $0$ and the verifiers' bit is $1$, then in the first round the
prover can only return $D_{00}$, while the verifier may ask for $D_{10}$ or $D_{11}$.
 
\paragraph{Encoding the Message}
If the two messages differ in a single bit, only the two rounds corresponding to that bit are affected. The adversaries can pass all other rounds by forwarding the prover's replies, and for the two rounds, even sending a maximally mixed state as a reply would pass these two rounds with constant probability. Thus, they can make the verifiers accept a different message without needing any shared entanglement. To remedy this, we encode the message with a code of length $m_0$ and insertion--deletion distance\footnote{We work with the insertion--deletion distance since the adversaries can suppress or add dummy requests to make the prover run behind or ahead of the verifiers on the codeword} $\Omega(m_0)$~\cite{SZ99}, so that even if the adversaries flip a single message bit of the prover, it should intuitively cause a constant fraction of mismatched rounds between the prover's response and the verifier's expectation. Assuming the adversaries have no preshared entanglement, this should lead to rejection, except with negligible probability.
 
However, this distance need not survive when each code bit is expanded into two response rules. For instance, write $F$ and $R$ for rounds with Rule~\ref{pba:rule:fixed} and Rule~\ref{pba:rule:random}, respectively. The all-zero and all-one codewords, though far apart in insertion--deletion distance, expand to rounds of the form $FRFR\cdots FR$ and $RFRF\cdots RF$, respectively, which have an insertion--deletion distance of only two, i.e. they align exactly after a one-round shift. 
To remedy this issue, we insert a zero before every code bit. For example, the all-zero and all-one codewords now use the rules $FRFR\cdots FRFR$ and $FRRF\cdots FRRF$. It is easy to see that no shift can line these patterns up, and we show that zero-insertion still preserves at least half of the original code's distance (\Cref{pba:lem:count}).
 
\paragraph{The Authentication Protocol}
Putting everything together, our position-based authentication protocol fixes a code $\mathcal C$ of length $m_0$ and insertion--deletion distance $\Omega(m_0)$. Let $m=2m_0$ and $\mathcal C'(M)\in\{0,1\}^{m}$ be the same as the codeword $\mathcal C(M)$ but with a zero inserted before every bit. The protocol is as follows.
\begin{itemize}
  \item The prover sends its message $M$ to the verifiers over an untrusted classical channel. The
        prover uses the codeword $\mathcal C'(M)$, and the verifiers use $\mathcal C'(M')$ for the
        message $M'$ they receive.
  \item The verifiers send the prover $W$ tuples of four independent BB84 states. Hence the honest prover's storage $R_H$ is $4W$.
  \item There are two consecutive rounds per code bit, so $2m$ rounds in total, each with $s$ slots.
        In each round, the verifiers choose a random slot, a random verifier, a random tuple index $r$
        and a random bit $e$, and the chosen verifier sends $(r,e)$ in that slot.
  \item The prover returns the qubit of the $r^{th}$ tuple given by the table above (on~\Cpageref{fig:rule}) for its current code bit and round, and discards the rest of the tuple.
  \item The verifiers check that the reply arrives on time and is the qubit prescribed by $\mathcal C'(M')$.
  \item After each round, the verifiers then replace the $r^{th}$ tuple; the verifiers accept if all rounds pass.
\end{itemize}
 
\paragraph{Proof Idea}
The proof follows the same template as the QPV proof, i.e., we show that the adversaries with non-negligible success probability must consume entanglement across the slots of a block of rounds, thus forcing large adversarial quantum registers.
 
We call a round \emph{mismatched} if, conditioned on the history before that round, the probability that the round is successful and the prover either gives no reply or follows Rule~\ref{pba:rule:fixed} while the verifiers' codeword prescribes Rule~\ref{pba:rule:random} is at least a sufficiently small fixed positive constant. As we discuss next, mismatched rounds with high success probability play the same role as high-success rounds in the QPV proof.
 
\paragraph{Entanglement Consumption in a Mismatched Round}
As in the QPV proof, we compare states by which one side can recover the requested qubit, and then the other side across the claimed location can recover, to prove entanglement consumption
via \Cref{thm:informal-entanglement-consumption}. There are two cases:
\begin{enumerate}
  \item \emph{Fixed reply against a choice.} Across all tuples, the prover can help with only $W$
        qubits, while the verifier may request any of $2W$ states.\label{it:fixed-reply}
  \item \emph{No reply.} The adversaries answer from their own registers, which may hold earlier
        replies, and may decide after seeing the request whether to obtain a new reply.\label{it:no-reply}
\end{enumerate}
In the first case, unlike in the QPV proof, the side that recovers the qubit at the end may use the prover's
qubits, which lie on the other side of the cut. The final recovery is then not local to one side, so
the hypothesis of \Cref{thm:informal-entanglement-consumption} fails. To restore it, we merge the
registers needed for recovery to the recovering side using state merging~\cite{HOW07}, applied to many independent copies of the state, and then apply the theorem. Since the prover can help with only
$W$ registers, the adversarial entanglement by \Cref{thm:informal-entanglement-consumption} still drops by $2W-W=W$; see \Cref{pba:lem:auth-consumption}, part~\ref{pba:case:helper}. In the second case, the drop in adversarial entanglement depends on the probability of the event in \Cref{it:no-reply}, i.e., answering without a new prover reply; see \Cref{pba:lem:auth-consumption}, part~\ref{pba:case:noresp}. 
The inter-slot argument of the QPV proof, based on reordered adversarial evolutions, carries over with the honest prover included in the reordering.\footnote{This reordering also handles overlapping slots directly, while preserving causal dependencies and physical delivery times.} Summing over the $s$ slots again gives an expected contribution proportional to $sR_H$ in a mismatched round, before taking into account the increase in entanglement due to recovery errors and replenishment costs, as in the QPV proof. 
 
\paragraph{From Mismatched Rounds to General Adversaries}
As in the QPV proof, we split the $2m$ rounds into $L=\lceil\log_2^2(\secparam)\rceil$ blocks and select one block
after conditioning on the earlier history. However, a highly successful block is no longer enough, since not all its rounds need to be mismatched, i.e., in some of its successful rounds the honest
prover may help the adversaries with the expected replies; thus there need not be an adversarial entanglement drop in these rounds. Hence, we need a block that is both highly successful \emph{and} has a constant expected fraction of mismatched rounds.
 
To see that such a block exists, let $P_{\mathcal B}$ indicate that block $\mathcal B$ passes and $Z_{\mathcal B}\in[0,1]$ be the fraction of its rounds in which either of the two cases above (\Cref{it:fixed-reply,it:no-reply}) actually occurs, and consider the weight $P_{\mathcal B}e^{tZ_{\mathcal B}}$ for a small constant $t>0$. A successful forgery passes every block, so $P_{\mathcal B}=1$ for all $\mathcal B$, and by the code's distance it has many such occurrences overall, so $\sum_{\mathcal B}Z_{\mathcal B}=\Omega(L)$. Hence the product of the weights, $\prod_{\mathcal B}P_{\mathcal B}e^{tZ_{\mathcal B}}$, is $e^{\Omega(L)}$. If every block had a small expected weight given its earlier history, the forgery probability would be negligible. Therefore, assuming non-negligible forgery probability, there must exist a block and an earlier successful history with conditional expected weight at least $1+\theta$, for some constant $\theta>0$. Therefore, this block passes with high probability and, by averaging over the histories before individual rounds, has a constant expected fraction of mismatched rounds, i.e., has $\Omega(m/L)$ mismatched rounds in expectation. Averaging their contributions over the histories within the block and subtracting the recovery errors and replenishment costs, as in the QPV proof, gives the required $\Omega(sR_Hm/L)$ entanglement bound. Hence, as in the QPV proof, the total adversarial quantum-register size is $\Omega(sR_Hm/L)=\widetilde{\Omega}(smR_H)$. 
 
\paragraph{From Authentication to Quantum Key Exchange}
Position-based authentication naturally leads to position-based key exchange by replacing the authenticated classical channel required for existing Quantum Key Distribution (QKD) protocols with a position-based authentication protocol; see~\cite{BCF+14,Unr14}.
We extend our construction of position-based
authentication to position-based key exchange by taking the compiler of~\cite[Section 4]{TNT26},
with our position-based authentication, to convert the
BB84 state-based QKD protocol of~\cite[Section 7]{TL17} into a position-based key exchange protocol by authenticating the entire public discussion only once, at the end. The verifiers release their key only if the authenticated discussion matches their own record.\footnote{We also append fresh random bits to each classical message during the QKD protocol to prevent the adversaries from correctly predicting an honest message before it is sent, except with negligible probability; see the proof of \Cref{pba:cor:qkd}.}

\paragraph{Paper Organization} \Cref{sec:quantum-information} reviews the quantum-information tools, \Cref{sec:def} defines quantum position verification and the adversarial model, and \Cref{sec:protocol} describes the protocol. \Cref{sec:quantum-info-lower-bound-entanglement-consumption} proves the entanglement-consumption result, \Cref{sec:security-proof} proves one-dimensional static-verifier security and lifts it to higher-dimensional mobile verifiers, \Cref{sec:continuous-time-model} formalizes the continuous-time model, and \Cref{sec:position-authentication-appendix} extends our results to position-based authentication and key exchange. \Cref{sec:appendix-basic-proofs,sec:appendix-proofs-of-entanglement-lemmas,sec:appendix-security-helpers} contain the supporting proofs.

\section{Preliminaries: Quantum Information}
\label{sec:quantum-information}
For quantum states $\rho$ and $\sigma$ on the same finite-dimensional Hilbert space, their normalized trace distance is $\TD{\rho,\sigma}:=\frac{1}{2}\|\rho-\sigma\|_1$.

We use squared Uhlmann fidelity $F(\rho,\sigma):=\|\sqrt{\rho}\sqrt{\sigma}\|_1^2$. If $\sigma=\ket{\psi}\bra{\psi}$ is pure, then $F(\rho,\sigma)=\bra{\psi}\rho\ket{\psi}$, and the Fuchs--van de Graaf inequalities read $1-\sqrt{F(\rho,\sigma)}\leq\TD{\rho,\sigma}\leq\sqrt{1-F(\rho,\sigma)}$.

For qubit registers $U$ and $V$, let $|\Phi^+\rangle_{UV}:=(|00\rangle+|11\rangle)/\sqrt{2}$ denote the maximally entangled Bell state and let $\Phi^+_{UV}:=|\Phi^+\rangle\langle\Phi^+|_{UV}$ denote its density matrix.

The binary entropy function is $h(\epsilon):=-\epsilon\log_2\epsilon-(1-\epsilon)\log_2(1-\epsilon)$ for $\epsilon\in[0,1]$, with $h(0)=h(1)=0$. We also define $g(\epsilon):=\epsilon\log_2 3+h(\epsilon)$ and $f(\epsilon):=\epsilon\log_2 3+2h(\epsilon)$. We use the following concave, non-decreasing extensions on $[0,1]$:
\begin{align}
g_{safe}(\epsilon)&:={}
\begin{cases}
g(\epsilon),&\epsilon\leq3/4,\\
g(3/4),&\epsilon>3/4,
\end{cases}\\
f_{safe}(\epsilon)&:={}
\begin{cases}
f(\epsilon),&\epsilon\leq1/2,\\
f(1/2),&\epsilon>1/2.
\end{cases}
\end{align}

For a bipartite state $\omega_{AB}$, the squashed entanglement is $E_{sq}(A:B)_\omega:=\frac{1}{2}\inf_{\omega_{ABE}}I(A:B|E)_\omega$, where the infimum is over all extensions $\omega_{ABE}$ of $\omega_{AB}$ and $I(A:B|E)_\omega:=S(AE)_\omega+S(BE)_\omega-S(E)_\omega-S(ABE)_\omega$ is the conditional mutual information.

We also use the coherent information $I_c(A\rangle B)_\omega:=S(B)_\omega-S(AB)_\omega$.

\vspace{0.3cm}
\noindent \textbf{Standard Facts} We use the following standard properties throughout the paper.
\begin{enumerate}
\item[(i)] \textbf{Pure states} If $\psi_{AB}$ is pure, then $E_{sq}(A:B)_\psi=S(A)_\psi=S(B)_\psi$~\cite{christandl2004squashed}.
\item[(ii)] \textbf{LOCC monotonicity} Squashed entanglement cannot increase under Local Operations and Classical Communication (LOCC)~\cite{christandl2004squashed}.
\item[(iii)] \textbf{Monogamy} For every tripartite state $\omega_{ABC}$, $E_{sq}(A:BC)_\omega\geq E_{sq}(A:B)_\omega+E_{sq}(A:C)_\omega$~\cite{christandl2004squashed}.
\item[(iv)] \textbf{Coherent-information lower bound} For every bipartite state $\omega_{AB}$, $E_{sq}(A:B)_\omega\geq I_c(A\rangle B)_\omega$. Indeed, squashed entanglement upper-bounds one-way distillable entanglement $E_D^\rightarrow$~\cite{christandl2004squashed}, while the hashing inequality gives $E_D^\rightarrow(A:B)_\omega\geq I_c(A\rangle B)_\omega$~\cite{devetak2005distillation}.
\item[(v)] \textbf{Fannes--Audenaert continuity} If $\rho$ and $\sigma$ act on a $d$-dimensional Hilbert space and $\TD{\rho,\sigma}=T\leq1-1/d$, then $|S(\rho)-S(\sigma)|\leq T\log_2(d-1)+h(T)$~\cite{audenaert2007sharp}. The right-hand side is non-decreasing in $T$ on $[0,1-1/d]$.
\end{enumerate}

\begin{definition}\label{def:classically-flagged-mixture-pure-states}
A register $C$ is classical in a joint state if, writing $Q$ for all remaining registers, the state has the form $\rho_{CQ}=\sum_c p_c\ket{c}\bra{c}_C\otimes\rho_Q^c$ in the standard computational basis. Thus $C$ may be classically correlated with $Q$, but it is not entangled with $Q$; requiring only the reduced state of $C$ to be diagonal would not be enough.

A state is a classically flagged mixture of pure states if it has the form $\rho_{FQ}=\sum_f p_f\ket{f}\bra{f}_F\otimes\ket{\psi_f}\bra{\psi_f}_Q$, where $F$ is a collection of classical flag registers. In other words, once the classical value $f$ is fixed, the remaining state is pure.
\end{definition}

The following observation records the simple operations that preserve this form.
\begin{remark}\label{remark:observation-classical-registers}
Adding fresh all-zero classical registers, copying classical values by bitwise CNOT, and computing classical functions into fresh registers preserve the classically flagged form. Projective measurements also preserve this form when their outcomes are recorded in fresh classical registers and the existing flags are retained. Conditioning on any value in the support of one of the classical flag registers also leaves a classically flagged mixture of pure states on the remaining registers.
\end{remark}

\subsection{Classical records and Quantum-register Size}

Our storage bound counts quantum registers but not classical records. The following standard identity lets us evaluate squashed entanglement by averaging over a classical record held on one side.
\begin{lemma}[Classical-Record Direct Sum]\label{lemma:direct-sum}
Let $\omega_{XAB}=\sum_x p_x|x\rangle\langle x|_X\otimes\omega_{AB}^x$, where $\{|x\rangle\}_x$ is an orthonormal basis and Alice holds the classical register $X$. Then $E_{sq}(AX:B)_\omega=\sum_x p_xE_{sq}(A:B)_{\omega^x}$.
\end{lemma}
The proof is given on \Cpageref{pf:lemma:direct-sum}.

The next lemma converts an entanglement lower bound into a lower bound on the number of physical qubits, even when the classical record is arbitrarily large.
\begin{lemma}[Quantum-Register Dimension Bound]\label{lem:quantum_dim_bound}
Let $\rho_{A_1A_2}$ be bipartite, with $A_1=C\otimes A_1'$, where $C$ is classical and $A_1'$ is a quantum register. If $\rho_{A_1A_2}=\sum_c p_c|c\rangle\langle c|_C\otimes\rho_{A_1'A_2}^{(c)}$, then $E_{sq}(A_1:A_2)_\rho\leq\log_2\dim(\mathcal H_{A_1'})$. In particular, if $A_1'$ consists of $Q$ qubits, then $E_{sq}(A_1:A_2)_\rho\leq Q$. By symmetry, the same statement holds with $A_1$ and $A_2$ interchanged.
\end{lemma}
The proof is given on \Cpageref{pf:lem:quantum_dim_bound}.


\ifnum\quantinfo=0

\section{Definitions of Quantum Position Verification}\label{sec:def}

We adapt the definition of a QPV protocol from~\cite{BCF+14} to the bounded quantum-register setting. The adversary model and the precise quantum-register accounting used in the security definitions are stated immediately after the definition. 
\begin{definition}\label{def:qpv}
An $\ell$-verifier Quantum Position Verification (QPV) protocol $\{\textnormal{QPV}_\secparam\}_{\secparam}$ in $d$-dimensional space is an interactive protocol between honest Quantum Polynomial-Time (QPT) verifiers $V_1,\dots,V_\ell$, a QPT prover $P$, and a deterministic QPT algorithm $\relocate$. A claimed location $x\in\mathbb R^d$ is \emph{admissible} for verifier positions $(x_1,\ldots,x_\ell)$ if $x\in\cvhull(x_1,\ldots,x_\ell)\setminus\{x_1,\ldots,x_\ell\}$. The protocol satisfies the following properties:
\begin{itemize}
\item \textbf{Correctness} For every set of verifier locations $x_1,\ldots,x_\ell$ and every admissible location $x$ with respect to them, an honest prover at $x$ succeeds with certainty in the protocol, i.e.,
\begin{equation*}
\Pr[\textnormal{QPV}_\secparam(V_1(x_1),\dots,V_\ell(x_\ell)\leftrightarrow P(x))\textnormal{ accepts}]=1.
\end{equation*}
\item \textbf{Static-verifier LOCC security under bounded quantum registers} For any $r$, the protocol has $r$-qubit static-verifier security if the following holds. Let $(cons_1,\ldots,cons_\ell)=\relocate(0)$ be the fixed verifier locations. For every admissible claimed location $x\in\cvhull(cons_1,\ldots,cons_\ell)$ and every finite coalition of adversaries $(P_1,\ldots,P_n)$ having some fixed locations in $\RR^d\setminus \{x\}$, if the coalition has at most $r(\lambda)$ qubits in total, excluding classical registers, and communicates internally only through classical messages, then the verifiers accept with negligible probability. Equivalently, there is a negligible function $\textnormal{negl}(\lambda)$ such that
\begin{equation*}
\Pr[\textnormal{QPV}_\lambda(V_1(cons_1),\dots,V_\ell(cons_\ell)\leftrightarrow P_1(1^\secparam),\dots,P_n(1^\secparam))\textnormal{ accepts}]\leq\textnormal{negl}(\lambda).
\end{equation*}
\item \textbf{Mobile-verifier LOCC security under bounded quantum registers} For any $r$, the protocol has $r$-qubit mobile-verifier security if the following holds. For every admissible claimed location $x$ with verifier locations $(z_1,\ldots,z_\ell)=\relocate(x)$ and every finite coalition of adversaries $(P_1,\ldots,P_n)$ having some fixed locations in $\RR^d\setminus \{x\}$, if the coalition has at most $r(\lambda)$ qubits in total and communicates internally only through classical messages, then the verifiers accept with negligible probability. Equivalently, there is a negligible function $\textnormal{negl}(\lambda)$ such that
\begin{equation*}
\Pr[\textnormal{QPV}_\lambda(V_1(z_1),\dots,V_\ell(z_\ell)\leftrightarrow P_1(1^\secparam),\dots,P_n(1^\secparam))\textnormal{ accepts}]\leq\textnormal{negl}(\lambda).
\end{equation*}

\end{itemize}

For any functions $t<r$, an $\ell$-verifier QPV protocol achieves a $t$-vs-$r$ static-verifier, respectively mobile-verifier, resource gap if the sum of the honest parties' total quantum and classical storage and total quantum and classical communication is at most $t(\lambda)$, while the protocol has $r$-qubit static-verifier, respectively mobile-verifier, security.
\end{definition}

\paragraph{Adversary Model}\phantomsection\label{sec:adversary-model}
The adversarial coalitions used in the preceding security definition satisfy the following conditions.
\begin{enumerate}
\item The coalition is finite. Each adversary and its local hardware remain at one fixed location throughout the protocol, and no adversary is located at the claimed location.
\item For each fixed $\secparam$, every adversary's complete strategy is a finite circuit, which may be arbitrarily large and need not be efficient. Joining these circuits by all fixed classical wires that may be used in any execution gives one finite circuit for the coalition. Every wire goes from an earlier operation to a later one, every quantum wire within the coalition remains local to one adversary, and every wire between adversaries is classical. Earlier messages and measurement outcomes may determine which later operations occur, when they occur, and how long an adversary waits, but they do not change the circuit. An unused operation acts as the identity, and an unused message wire carries a fixed null value at its assigned delivery time.
\item Each adversary's circuit, and hence the coalition's circuit, has a fixed set of quantum registers, which are never added or removed during the execution. When an adversary discards the contents of a register, we represent this by measuring that register in the computational basis, keeping the outcome only in a classical register used for the analysis and never read by the adversarial strategy, and resetting the same quantum register so that it may be reused. If an incoming verifier qubit is kept, we represent this by swapping the incoming register with an available reset register of the circuit.
\item No information or physical system travels faster than speed $1$. Local quantum operations are instantaneous, but the adversaries may wait to delay their operations or messages.
\item The adversaries may share an arbitrary initial quantum state and perform arbitrary local quantum operations. They may exchange arbitrary classical messages but cannot send quantum messages to one another.\footnote{They may still receive quantum messages from the verifiers and send quantum messages to the verifiers as allowed by the protocol.}
\item Classical computation, communication, and storage are unrestricted and are not counted. The adversarial coalition's quantum-register size is the total number of qubits that its fixed local and outgoing quantum registers can hold.\footnote{Following the discussion in the introduction, an appropriate version of \Cref{def:qpv} would be to count both classical and quantum storage as well as classical and quantum communication as adversarial resources  when bounding the adversarial coalition's resources. The present definition is stronger as it clearly implies such a definition.} Every outgoing quantum register from an adversary remains included in the sending adversary's quantum-register size and unavailable for reuse until the verifier receives or discards its state.
\end{enumerate}

\begin{definition}[Inter-slot-delay independence]\label{def:qpv-slot-based}
A one-dimensional two-verifier QPV protocol is called slot-based if each round uses logical times $T_1<T_2<\cdots<T_s$ with fixed spacing $\Delta:=T_{j+1}-T_j$, and each verifier instruction and honest-response deadline is fixed relative to its corresponding time $T_j$. The time intervals associated with different slots may overlap.

For a constant $c>0$, let $c-\Pi$ be the protocol obtained from $\Pi$ by multiplying the time between consecutive logical times by $c$, without changing the timing of any instruction or response within a slot. We say that $\Pi$ is inter-slot-delay independent with security bound $B(\lambda)$ if $c-\Pi$ has the same $B(\lambda)$-qubit static-verifier security guarantee as $\Pi$ for every constant $c>0$.
\end{definition}

\begin{definition}\label{def:qpv-d-dimensional-generalization}
Let $\Pi$ be a one-dimensional static-verifier QPV protocol with claimed location $0$ and verifier locations $-R$ and $R$. Its $d$-dimensional mobile-verifier generalization $\Pi'$ places the claimed location at $z\in\mathbb R^d$, chooses a unit vector $u$, and places the verifiers at $z-Ru$ and $z+Ru$, so that $z$ is the midpoint of the line segment between them. The entire one-dimensional configuration, including the directed paths of the verifier inputs, the logical times, and the response deadlines, is preserved by this translation and rotation. Conversely, $\Pi$ is called the one-dimensional version of $\Pi'$.
\end{definition}

\section{QPV Protocol}\label{sec:protocol}

First, we describe the one-round base protocol, which is a simplified protocol; see \Cref{fig:qpv-protocol-single-round-single-qubit}.

\begin{figure}[t]
\centering
\fbox{%
\begin{minipage}{0.98\linewidth}
\footnotesize
\setlength{\tabcolsep}{3pt}
\renewcommand{\arraystretch}{1.2}
\begin{tabular}{p{0.08\linewidth}|p{0.28\linewidth}|p{0.28\linewidth}|p{0.28\linewidth}|}
\textbf{Slot} & \textbf{$V_0$} & \textbf{$P$} & \textbf{$V_1$} \\
\hline
$t_0$ & Sample $x,\theta,b\xleftarrow{\$}\{0,1\}$ and $i\xleftarrow{\$}[s]$; prepare $\ket{\psi}=H^\theta\ket{x}$ and send $\ket{\psi}$ to $P$. & Receive and store $\ket{\psi}$. & Receive the sampled values $(x,\theta,b,i)$ from $V_0$. \\
\hline
$1\leq t<i$ & Idle. & Keep the stored qubit and wait. & Idle. \\
\hline
$t=i$ ($b=0$) & Send $\mathsf{challenge}$ to $P$. & & \\
& & Immediately send $\ket{\psi}$ to $V_0$. & Idle. \\
& If a response arrives at the prescribed deadline, perform the projection test onto the state $\ket{\psi}$ using its classical description. If the projection succeeds, set the output to be $1$, else $0$. If no response arrives at the prescribed deadline, abort and output $0$. & & \\
\hline
$t=i$ ($b=1$) & & & Send $\mathsf{challenge}$ to $P$. \\
& Idle. & Immediately send $\ket{\psi}$ to $V_1$. & \\
& & & If a response arrives at the prescribed deadline, perform the projection test onto the state $\ket{\psi}$ using its classical description. If the projection succeeds, set the output to be $1$, else $0$. If no response arrives at the prescribed deadline, abort and output $0$. 
\\
\end{tabular}
\end{minipage}%
}
\caption{Synchronized interactive view of the base protocol $\qpvb(V_0,V_1,P)$.}
\label{fig:qpv-protocol-single-round-single-qubit}
\end{figure}

To make the probability of a successful attack negligible, the amplified protocol repeats the base protocol $m$ times. The full protocol is given in \Cref{fig:qpv-protocol-multi-round-unified}.

\begin{figure}[p]
\centering
\fbox{%
\begin{minipage}{0.98\linewidth}
\scriptsize
\setlength{\tabcolsep}{1.8pt}
\renewcommand{\arraystretch}{0.9}
\begin{tabular}{p{0.08\linewidth}|p{0.28\linewidth}|p{0.28\linewidth}|p{0.28\linewidth}|}
\textbf{Slot} & \textbf{$V_0$} & \textbf{$P$} & \textbf{$V_1$} \\
\hline
$t_0$ & For every $j\in[R_H]$, assign the distinct label $r_j=j$, sample $x_{r_j},\theta_{r_j}\xleftarrow{\$}\{0,1\}$, prepare $\ket{\psi}_{r_j}=H^{\theta_{r_j}}\ket{x_{r_j}}$, and send the states to $P$ in label order. Initialize the ordered list $S:=((r_j,x_{r_j},\theta_{r_j}))_{j\in[R_H]}$. & Receive and store the labeled states in the same order. & Receive the labeled records from $V_0$. The verifiers maintain the same list $S$.\\
\hline
\multicolumn{4}{c}{$\vdots$}\\
\hline
\multicolumn{4}{c}{\textbf{$k^{th}$ round for $k\in[m]$}}\\

\multicolumn{4}{p{0.91\linewidth}|}{\centering $V_0$ samples $b_k\xleftarrow{\$}\{0,1\}$ and $i_k\xleftarrow{\$}[s]$, and samples a label $r_k$ uniformly from the $R_H$ distinct labels currently stored in $S$.}\\
\multicolumn{4}{p{0.91\linewidth}|}{\centering Remove the record labeled $r_k$ from $S$ without changing the order of the remaining records.}\\
\hline
$1\leq t<i_k$ & Idle. & Idle. & Idle. \\
\hline
$t=i_k$ (if $b_k=0$) & Send $\mathsf{challenge},r_k$ to $P$. & & \\
& & Immediately remove $\ket{\psi_{r_k}}$ from the list without changing the order of the remaining states and send it to $V_0$. & Idle. \\
& If a response arrives at the prescribed deadline, perform the projection test onto $H^{\theta_{r_k}}\ket{x_{r_k}}$. If the projection succeeds, continue and set the outcome of the round to be $1$, else abort. If no response arrives at the prescribed deadline, abort. & & \\
\hline
$t=i_k$ (if $b_k=1$) & & & Send $\mathsf{challenge},r_k$ to $P$. \\
& Idle. & Immediately remove $\ket{\psi_{r_k}}$ from the list without changing the order of the remaining states and send it to $V_1$. & \\
& & & If a response arrives at the prescribed deadline, perform the projection test onto $H^{\theta_{r_k}}\ket{x_{r_k}}$. If the projection succeeds, continue and set the outcome of the round to be $1$, else abort. If no response arrives at the prescribed deadline, abort. 
\\
\hline
$i_k<t\leq s$ & Idle. & Idle. & Idle. \\
\hline
\multicolumn{4}{p{0.91\linewidth}|}{\centering After all the slots of round $k$ have ended, 
$V_0$ chooses a fresh label, samples a fresh pair $(x,\theta)$ for that label, and appends the labeled record to $S$, and shares it with $V_1$.}\\
\multicolumn{4}{p{0.91\linewidth}|}{\centering $V_0$ sends the corresponding state $H^\theta\ket{x}$ to $P$, which receives and appends it to its list.}\\
\hline
\multicolumn{4}{c}{$\vdots$}\\
\hline
\multicolumn{4}{p{0.91\linewidth}|}{\centering\textbf{Output of the experiment} is $1$ if and only if the outcome of  every round is $1$.}\\
\end{tabular}
\end{minipage}%
}
\caption{Synchronized Interactive view of the protocol $\qpva(V_0,V_1,P)$.}
\label{fig:qpv-protocol-multi-round-unified}
\end{figure}

For polynomial-time-computable, polynomially bounded integer-valued functions $R_H,s,m$ of the security parameter $\secpar$, and assuming perfect quantum communication, correctness of the protocol in \Cref{fig:qpv-protocol-multi-round-unified} is immediate.

\paragraph{Protocol Rules}\phantomsection\label{sec:protocol-rules}
The following rules are part of the protocol.
\begin{enumerate}
\item\label{cond:private-verifier-coordination} The verifiers coordinate through an information-theoretically private and authenticated classical channel. They send classical information to the adversarial coalition only as prescribed by the protocol. In particular, the descriptions of the BB84 qubits are never explicitly revealed to the adversarial coalition. The adversaries nevertheless do learn whether the round was accepted or rejected since the verifiers abort in the latter event. 
\item\label{cond:dynamic-input-path} Each initialization or replacement qubit sent by $V_0$ is sent along the fixed one-way path from $V_0$, through the claimed location, to $V_1$. Before beginning each round, the verifiers wait until the time at which the qubits sent by $V_0$ would reach $V_1$. If $V_1$ receives any of these qubits, the verifiers abort, since the honest prover should have intercepted and stored them.

\item\label{cond:slot-timing} For each slot $j$, the logical time $T_j$ is the time at which the selected verifier's instruction reaches the claimed location. The honest prover immediately returns the requested state, and the corresponding deadline is the time at which this honest response reaches the selected verifier.
\item\label{cond:nonoverlap-response} If the response does not reach the respective verifier at the respective deadline, the verifiers abort. The verifiers also abort if either verifier receives any additional qubit at the deadline or an unwarranted response arrives at a time other than the correct deadline. 
\item\label{cond:nonoverlap-ordering} The verifiers complete all operations of a round before sending a fresh BB84 state as replenishment for the next round.
\end{enumerate}

In \Cref{sec:security-analysis-protocol}, we prove that any adversarial coalition against the amplified protocol in \Cref{fig:qpv-protocol-multi-round-unified} with non-negligible success probability requires total adversarial quantum-register size $\Omega(smR_H/\log^2\lambda)$, while the honest prover uses $R_H$ qubits and the protocol sends $R_H+2m$ honest quantum messages; see \Cref{thm:amplified_security-generalized}. The protocol is slot-based in the sense of \Cref{def:qpv-slot-based}. \Cref{thm:amplified_security-generalized} also proves that its security is independent of the delay between consecutive slots.
Then, by \Cref{thm:generic-reduction}, we show that the one-dimensional static-verifier security result also gives higher-dimensional mobile-verifier security.

However, before proving the security for our protocol, we take a detour to prove the entanglement consumption theorem, which, as discussed in the technical overview, will be essential for our security proof.

\section{Lower Bound on Entanglement Consumption}\label{sec:quantum-info-lower-bound-entanglement-consumption}
Recall the functions $g,f,g_{safe}$, and $f_{safe}$ defined in \Cref{sec:quantum-information}.


Recall that in each round of our protocol, the final random Pauli test essentially estimates the closeness of the recovered pair to a Bell pair.  The following lemma converts this closeness into a (dimension-independent) lower bound on squashed entanglement.
\begin{lemma}[Bell Closeness Implies Entanglement]\label{lemma:lower_bound-closeness}
Let $\tau_{OW}$ be a bipartite state on two qubits ($O$ and $W$) such that its trace distance to a maximally entangled pure Bell state $|\Phi^+\rangle$ is bounded by $\epsilon'$, i.e., $\TD{\tau_{OW}, \Phi^+} \le \epsilon'$ for $0\leq\epsilon' \leq 1/2$. Then the squashed entanglement of $\tau_{OW}$ is lower-bounded by:
\begin{equation}
E_{sq}(O:W)_\tau \ge 1-f(\epsilon')= 1 - \epsilon' \log_2(3) - 2h(\epsilon').
\end{equation}
where $h(\epsilon') = -\epsilon' \log_2 \epsilon' - (1-\epsilon') \log_2(1-\epsilon')$ is the binary entropy function.
\end{lemma}

The proof of the lemma is given on \Cpageref{pf:lemma:lower_bound-closeness}.

The protocol supplies BB84 test-success probabilities, whereas the main entanglement theorem uses Bell-pair fidelity.  The next lemma converts one quantity into the other after averaging over the relevant classical choices.
\begin{lemma}[BB84 Success Implies Bell Fidelity]
\label{lem:bb84_fidelity}
Let $\{q_x, \rho_{AB}^{(x)}\}$ be a probabilistic ensemble of bipartite quantum states. Suppose the average success probability of this ensemble passing a random two-qubit BB84 verification test (measuring both qubits in a randomly chosen $XX$ or $ZZ$ basis) is $p = \sum_x q_x p_x$. Then the average fidelity of the ensemble to the maximally entangled Bell pair $\Phi^+$ is lower bounded by:
\begin{equation}
    F_{avg} := \sum_x q_x \langle \Phi^+ | \rho_{AB}^{(x)} | \Phi^+ \rangle \ge 2p - 1.
\end{equation}
\end{lemma}
The proof of the lemma is given on \Cpageref{pf:lem:bb84_fidelity}.

At the initial endpoint, we compare the entanglement before and after adjoining the verifier reference to one side of the cut.  For a pure conditional state, the following inequality gives this comparison without a dimension factor.
\begin{lemma}[Pure-State Entanglement Comparison]\label{lemma:pure-state-flow}
For any pure state $\omega_{O A_1 A_2}$,
\begin{equation}
    E_{sq}(O A_1 : A_2)_\omega \le E_{sq}(A_1 : A_2)_\omega + I(O : A_2)_\omega.
\end{equation}
\end{lemma}
The proof of the lemma is given on \Cpageref{pf:lemma:pure-state-flow}.

For several reference qubits, the one-qubit estimates must be combined even when those qubits are correlated.  The next lemma records the exact correlation term needed for this summation.
\begin{lemma}[General additivity of Mutual-Information]
\label{lemma:superadditivity}
Let $\rho_{OA}$ be a bipartite quantum state where system $O$ decomposes into $k$ subsystems $O_1, O_2, \dots, O_k$. For any system $A$, the mutual information satisfies:
\begin{equation}
    I(O : A)_\rho \ge \sum_{j=1}^k I(O_j : A)_\rho - \Delta_O,
\end{equation}
where the correlation penalty $\Delta_O$ is defined as $\Delta_O := \sum_{j=1}^k S(O_j)_\rho - S(O)_\rho$.
\end{lemma}
The term $\Delta_O$ can be seen as a penalty based on how far the state on $O$ is from being a tensor-product state of the form $\bigotimes_{i\in [k]}\rho_{O_i}$, and vanishes in the special case when it is indeed a tensor product state. The proof of the lemma is given on \Cpageref{pf:lemma:superadditivity}.

Combining the preceding two lemmas gives the estimate required at the initial recovery endpoint.  This is the last technical input to the main entanglement-loss theorem.

\begin{lemma}[Initial Recovery Bound]\label{lemma:bounding-initial-state-k-extraction}
Let $O = O_1 \cdots O_k$ be a quantum register consisting of $k$ qubits. Let $\rho_{O A_1 A_2}$ be a classically flagged mixture of pure states (see \Cref{def:classically-flagged-mixture-pure-states}) across registers $O$, $A_1$, and $A_2$, where the classical flag $C$ is a subregister of $A_1$ (so $A_1 = C \otimes A_1'$), i.e.,
\begin{equation}
    \rho_{O A_1 A_2} = \sum_{c} p_c |c\rangle\langle c|_C \otimes |\psi_c\rangle\langle \psi_c|_{O A_1' A_2}.
\end{equation}
Suppose for each $j \in [k]$, there exists a local unitary $U_j = \sum_c |c\rangle\langle c|_C \otimes U_c^{(j)}$ acting on $A_1$, where $U_c^{(j)}$ acts strictly on $A_1'$ and yields a target register $R_j \subset A_1'$. Let $\tilde{\rho}^{(j)}$ be the resulting state, such that the expected trace distance of the conditional states from the Bell state is at most $\delta_j\in[0,1]$, i.e., $\mathbb{E}_c [TD(\tilde{\rho}_{O_j R_j}^{(c,j)}, \Phi^+)] \le \delta_j$. Then,
\begin{equation}
    E_{sq}(O A_1 : A_2)_\rho \le E_{sq}(A_1 : A_2)_\rho + 2 \sum_{j=1}^k g_{safe}(\delta_j).
\end{equation}
\end{lemma}

The proof of the lemma is given on \Cpageref{pf:lemma:bounding-initial-state-k-extraction}.

We can now quantify the entanglement lost when LOCC operations transfer the ability to recover $k$ Bell-pair halves from one side of a bipartition to the other.  We first state the version with locally controlled unitaries.

\begin{theorem}[Entanglement Consumed due to Transferring Entangled Bell Halves]\label{theorem:k-extraction-entanglement-consumption}
Let $O = O_1 \cdots O_k$ be a quantum register consisting of $k$ qubits. Let $\rho_{O A_1 A_2}$ be a classically flagged mixture of pure states (see \Cref{def:classically-flagged-mixture-pure-states}) across registers $O$, $A_1$, and $A_2$, where the classical flag $C$ is a subregister of $A_1$ (so $A_1 = C \otimes A_1'$), i.e.,
\begin{equation}
    \rho_{O A_1 A_2} = \sum_{c} p_c |c\rangle\langle c|_C \otimes |\psi_c\rangle\langle \psi_c|_{O A_1' A_2}.
\end{equation}

Suppose that initially, there exist $k$ local unitaries $\{U_j\}_{j=1}^k$ acting on $A_1$ such that $U_j$ does not scramble the classical flag register, i.e., $U_j$ can be written as $\sum_{c}\ket{c}\bra{c}_C\otimes U^{(c)}_j$ such that the global fidelity of the recovered marginal states to $\Phi^+$ satisfies:
\begin{equation}
    \frac{1}{k} \sum_{j=1}^k F(\tilde{\rho}_{O_j R_j}^{(j)}, \Phi^+) \ge 1 - \epsilon^2, \quad \text{where } 0\leq\epsilon \le 3/4.
\end{equation}

An LOCC channel $\Lambda$ is applied across $A_1 : A_2$, resulting in $\tau_{O A_1 A_2} = (\text{id}_O \otimes \Lambda)(\rho)$. Finally, suppose there exist $k$ local unitaries $\{V_j\}_{j=1}^k$ acting on $A_2$ such that the global fidelity of the final recovered marginal states satisfies:
\begin{equation}
    \frac{1}{k} \sum_{j=1}^k F(\tilde{\tau}_{O_j W_j}^{(j)}, \Phi^+) \ge 1 - (\epsilon')^2, \quad \text{where } 0\leq\epsilon' \le 1/2.
\end{equation}

Then the squashed entanglement consumed across $A_1 : A_2$ is lower-bounded by:
\begin{equation}
    \Delta E_{sq} \ge k \Big[ 1 - 2g(\epsilon) - f(\epsilon') \Big].
\end{equation}
\end{theorem}

The proof is given on \Cpageref{pf:theorem:k-extraction-entanglement-consumption}. Note that the adversarial operations to recover the Bell-pair halves may be arbitrary quantum channels in general, rather than unitaries. Hence, we use Stinespring dilation to extend \Cref{theorem:k-extraction-entanglement-consumption} to the general setting of arbitrary quantum channels as recovery operations.
\begin{corollary}[Local Recovery by Quantum Channels]\label{cor:k-extraction-entanglement-consumption}
Let $O = O_1 \dots O_k$ be a quantum register consisting of $k$ qubits. Let $\rho_{OA_1A_2}$ be a classically flagged mixture of pure states (see \Cref{def:classically-flagged-mixture-pure-states}) defined exactly as in \Cref{theorem:k-extraction-entanglement-consumption}.

Suppose that initially, there exist $k$ local quantum channels (CPTP maps) $\{\mathcal{E}_j\}_{j=1}^k$ acting on $A_1$ that do not scramble the classical flag register, such that the average fidelity of the recovered marginal states satisfies:
\begin{equation}
    \frac{1}{k} \sum_{j=1}^k F(\tilde{\rho}_{O_jR_j}^{(j)}, \Phi^+) \ge 1 - \epsilon^2, \quad \text{where } 0\leq\epsilon \le 3/4.
\end{equation}

An LOCC channel $\Lambda$ is applied across $A_1 : A_2$, resulting in $\tau_{OA_1A_2} = (\mathcal{I}_O \otimes \Lambda)(\rho_{OA_1A_2})$. Finally, suppose there exist $k$ local quantum channels (CPTP maps) $\{\mathcal{F}_j\}_{j=1}^k$ acting on $A_2$ such that the average fidelity of the final recovered marginal states satisfies:
\begin{equation}
    \frac{1}{k} \sum_{j=1}^k F(\tilde{\tau}_{O_jW_j}^{(j)}, \Phi^+) \ge 1 - (\epsilon')^2, \quad \text{where } 0\leq\epsilon' \le 1/2.
\end{equation}

Then the squashed entanglement consumed across $A_1 : A_2$ due to the LOCC channel $\Lambda$ is lower-bounded by:
\begin{equation}
    \Delta E_{sq} \ge k[1 - 2g(\epsilon) - f(\epsilon')].
\end{equation}
\end{corollary}

The proof is given on \Cpageref{pf:cor:k-extraction-entanglement-consumption}.

The recovery operations at the two endpoints can depend on classical records stored on different adversarial sides. The physical labels of the reference qubits may also depend on those records.
\begin{corollary}[Classical Records on Both Sides]
\label{cor:k-extraction-entanglement-consumption-distributed}
Let $A_0=C_0A_0'$ and $A_1=C_1A_1'$. For every value $c=(c_0,c_1)$ of the two classical records, occurring with probability $p_c$, let
\begin{equation}
    \bigl(r_1(c),\ldots,r_k(c)\bigr)
\end{equation}
be an ordered list of $k$ distinct reference-qubit labels. When the classical records have value $c$, regard the reference spaces in this order as the register $O=O_1\cdots O_k$, with $O_j:=O_{r_j(c)}$.

With this notation, suppose
\begin{equation}
    \rho_{OA_0A_1}
    =
    \sum_{c=(c_0,c_1)}p_c
    \ket{c_0}\!\bra{c_0}_{C_0}
    \otimes\ket{c_1}\!\bra{c_1}_{C_1}
    \otimes
    \ket{\psi_c}\!\bra{\psi_c}_{OA_0'A_1'},
\end{equation}
where every $\ket{\psi_c}$ is pure. Let $\Lambda$ be an LOCC channel across $A_0:A_1$ that leaves $C_0,C_1$ and the reference qubits unchanged, and let
\begin{equation}
    \tau=(\mathcal I_O\otimes\Lambda)(\rho).
\end{equation}

For every $c$ and $j\in[k]$, suppose there is an initial recovery map $\mathcal E_j^c$ on $A_0'$ and a final recovery map $\mathcal F_j^c$ on $A_1'$. These maps may use the complete value of $c$ after the classical record held on the other side is copied for the analysis. Suppose their average recovered fidelities satisfy
\begin{align}
    \frac{1}{k}\sum_c p_c\sum_{j=1}^k
    F\!\left(\widetilde\rho^{\,c,j}_{O_jR_j},\Phi^+\right)
    &\geq1-\epsilon^2,
    &
    0\leq\epsilon&\leq\frac34,
    \\
    \frac{1}{k}\sum_c p_c\sum_{j=1}^k
    F\!\left(\widetilde\tau^{\,c,j}_{O_jW_j},\Phi^+\right)
    &\geq1-(\epsilon')^2,
    &
    0\leq\epsilon'&\leq\frac12.
\end{align}
Then
\begin{equation}
    \Esq(A_0:A_1)_\rho-\Esq(A_0:A_1)_\tau
    \geq k\bigl[1-2g(\epsilon)-f(\epsilon')\bigr].
\end{equation}
\end{corollary}

The proof is given on \Cpageref{pf:cor:k-extraction-entanglement-consumption-distributed}.

\section{Security Proof of the QPV Protocol in the Discrete-Time Model}
\label{sec:security-proof}

In this section, we use the preceding squashed entanglement consumption bounds to prove static-verifier security of our QPV protocol (\Cref{fig:qpv-protocol-multi-round-unified}). First, we state the conventions used throughout this section. We then show that any possibly overlapping execution can be replaced by a nonoverlapping execution with the same verifier geometry; see \Cref{prop:overlap-to-nonoverlap}. Finally, we prove static-verifier security in \Cref{thm:amplified_security-generalized} and lift the result to higher-dimensional mobile verifiers in \Cref{thm:generic-reduction}.

We use the adversary model stated in \Cref{sec:def} (see the list after \Cref{def:qpv}).
\paragraph{Notations for the Proof}
\label{sec:security-analysis-protocol}
\begin{itemize}
\item We translate the claimed location to $0$ and, if necessary, reflect the coordinates without changing the verifier names or roles, so that $V_0=-L$ and $V_1=R$ for $L,R>0$. In particular, $V_0$ remains the sender of the initialization and replacement states.
\item For slot $j$, $T_j$ is the logical time defined in condition~\ref{cond:slot-timing}. Thus $V_0$ sends at time $T_j-L$ when selected and its honest-response deadline is $T_j+L$; the corresponding times for $V_1$ are $T_j-R$ and $T_j+R$. Set $D:=\max\{L,R\}$, use $[T_j-D,T_j+D]$ as the full interval of slot $j$, and let $\Delta:=T_{j+1}-T_j$.
\item Let $A_0$ contain the local and outgoing registers assigned to adversaries strictly to the left of $0$, and define $A_1$ similarly on the right. An outgoing register remains on its sender's side until the verifier receives its state or the state is discarded. This grouping changes neither the adversaries' locations nor any communication time.
\item Give every verifier state a permanent label. The initialization labels are $1,\ldots,R_H$, and the state added before round $k>1$ has label $R_H+k-1$. Let $\xi=(\xi_0,\xi_1)$ denote the earlier classical records and let $\mathbf r_k(\xi)=(r_{k,1}(\xi),\ldots,r_{k,R_H}(\xi))$ be the ordered labels immediately before the variables of round $k$ are sampled. Write $U_k\xleftarrow{\$}[R_H]$ for the requested list index and set $R_k:=r_{k,U_k}(\xi)$. In the Bell-pair description used in the proof, $O_r$ denotes the verifier's reference qubit with label $r$.
\end{itemize}
By conditions~\ref{cond:private-verifier-coordination} and~\ref{cond:dynamic-input-path}, fixing $\xi$ leaves $U_k$ uniform and independent of the earlier quantum state.

\paragraph{Analysis Conventions}
\label{sec:security-analysis-conventions}
\begin{enumerate}
\item The verifiers preserve all classical data from previous rounds, including the classical description of states used up in previous rounds.
\item Additional classical registers used in the analysis begin in the all-zero state. They may store or copy classical values, and the adversarial strategy ignores them.
\item A classical message is represented by copying its value into a fresh classical register of the receiver. If a local operation depends on a classical record $c$, we write it as the classically controlled map $\bigoplus_c\mathcal N^c$.
\item We may copy a classical verifier instruction when it first passes an adversary. For the analysis, if no adversary exists between a verifier and the claimed location, we add an adversary with only classical registers there. It copies each classical instruction as it passes and otherwise leaves the execution unchanged. Information about an instruction from $V_a$ that later reaches $A_{1-a}$ is then treated as classical communication from $A_a$. This changes neither the success probability nor the quantum-register size.
\item\label{cond:nonoverlap-slot-structure} In the security analysis, every round is continued through all $s$ logical slots, even when the selected slot occurs earlier.\footnote{For the analysis, we continue the protocol and the adversaries' strategy through all rounds even after rejection, while keeping the final outcome as rejection.} When we change the delay between consecutive slots, we change only their timing, not the protocol actions, message contents, instruction paths, or timing within an individual slot.
\item After fixing a pure component of the initial state, we keep enough classical outcomes that fixing the complete record leaves the remaining quantum state pure, while tracing out the additional records recovers the original execution.\footnote{Formally, each recorded outcome of a local operation corresponds to one Kraus operator. Whenever the contents of a quantum register are discarded, we first measure the register in a rank-one basis, keep the outcome only as a classical record for the analysis, and reset the same register to a fixed state so that it may be reused. Ignoring the recorded outcome recovers exactly the original discard operation.}
\item For the squashed-entanglement calculation, an input traveling from $V_0$ to $V_1$ is counted with $A_0$ before it crosses $0$ and with $A_1$ afterward.
\end{enumerate}

We are now ready to state the proposition used to prove \Cref{thm:amplified_security-generalized}.
\paragraph{From an overlapping-slot execution to a nonoverlapping-slot execution} The full interval for slot $j$ is $[T_j-D,T_j+D]$, so consecutive slots overlap when their spacing is less than $2D$. The following proposition shows that an attack against any possibly overlapping execution gives an equally strong attack after the delay between consecutive slots is increased.

\begin{proposition}[Reduction to a Nonoverlapping Execution]\label{prop:overlap-to-nonoverlap}
Consider the one-dimensional execution of $\qpva$ with claimed location $0$, verifiers $V_0=-L$ and $V_1=R$, and consecutive slot times separated by $\Delta>0$. Set $D:=\max\{L,R\}$. Suppose that a finite adversarial coalition with no member at $0$ has total quantum-register size $Q$ and succeeds with probability $p$. For every $\Delta^\star>\max\{\Delta,2D\}$, the execution with the same verifier and adversary locations, the same protocol actions, and slot spacing $\Delta^\star$ has a coalition with quantum-register size $Q$ and success probability at least $p$. Moreover, the consecutive slot intervals for the new execution are disjoint.
\end{proposition}

The proof is given in \Cref{sec:proof-overlap-to-nonoverlap}. It uses protocol rules~\ref{cond:dynamic-input-path}--\ref{cond:nonoverlap-ordering} and analysis convention~\ref{cond:nonoverlap-slot-structure}.

\begin{remark}
The same conclusion holds for any one-dimensional slot-based QPV protocol satisfying rules~\ref{cond:dynamic-input-path} to \ref{cond:nonoverlap-ordering} for which convention~\ref{cond:nonoverlap-slot-structure} applies, provided that each round has at most one non-idle return instruction, sent by the selected verifier in the selected slot along a fixed straight path toward the claimed location, while the verifiers are idle in every other slot.
\end{remark}
\paragraph{Security of the Amplified Protocol}

We now combine the reduction from an overlapping execution to a nonoverlapping one with the squashed-entanglement consumption bound from \Cref{sec:quantum-info-lower-bound-entanglement-consumption} to prove static-verifier security of the protocol in \Cref{fig:qpv-protocol-multi-round-unified}.
\begin{theorem}[Static-verifier security of the protocol given in \Cref{fig:qpv-protocol-multi-round-unified}]
\label{thm:amplified_security-generalized}
Let $m:=m(\secparam)$, $R_H:=R_H(\secparam)$ and $s:=s(\secparam)$ be polynomial functions of $\secparam$ such that $s\geq9, R_H\geq 1$ and $m=\omega(\log_2^2\secparam)$. The one-dimensional QPV protocol in \Cref{fig:qpv-protocol-multi-round-unified} is slot-based and inter-slot-delay independent, and satisfies
\begin{equation}
    \left\lfloor\frac{m}{\lceil\log_2^2(\secparam)\rceil}\right\rfloor R_H
    \bigl(2\kappa(s-2.5)-1\bigr)
\end{equation}
-qubit (and hence $\Omega\left(\frac{m}{\log_2^2(\secparam)} sR_H\right)$-qubit) static-verifier security for a universal constant $\kappa>0.08$.

In particular, for any admissible claimed location (and for arbitrary and not necessarily polynomially bounded parameters $s(\secparam),R_H(\secparam),m(\secparam)$ with $s\geq9, R_H\geq 1$ and $m=\omega(\log_2^2\secparam)$), if an adversarial coalition (with no adversaries at the claimed location) in 1-D,
succeeds with probability at least $(1-10^{-7})^{\lceil\log_2^2(\secparam)\rceil}$, then, its total quantum-register size must be strictly greater than
\begin{equation}
    \left\lfloor\frac{m}{\lceil\log_2^2(\secparam)\rceil}\right\rfloor R_H
    \bigl(2\kappa(s-2.5)-1\bigr),
\end{equation}
for large enough $\secpar$. 
\end{theorem}

\begin{proof}
Let $\nu(\secparam):=(1-10^{-7})^{\lceil\log_2^2(\secparam)\rceil}$, which is negligible. Fix a finite coalition $P_0,\ldots,P_n$, with no adversary at the claimed location and success probability $p_{\mathrm{win}}\geq\nu(\secparam)$. We show that its total quantum-register size exceeds the bound in the theorem for large enough $\secparam$. By linearity of acceptance probability in the initial state, we may take that state to be pure without decreasing the success probability.

We use the equivalent Bell-pair description: whenever $V_0$ would send a random BB84 state, it instead prepares a Bell pair, keeps one reference qubit, and sends the other. Measuring the reference qubit in a random basis from $\{X,Z\}$ before transmission prepares the original random BB84 state. Since the basis and outcome are not used before the verifier's test, this measurement may be deferred until then. If no response arrives at the deadline, discarding the reference qubit unmeasured gives the same adversarial state as averaging over the hidden BB84 choice.

\paragraph{Setting up the Framework} Next, we show that the adversarial coalition must contain at least one adversary on each side of $0$. We first suppose that every adversary lies to the left. We fix any earlier transcript and condition on all earlier rounds being won. For challenges from $V_1$, we compare the $s$ executions corresponding to the possible challenge slots with one execution in which both verifiers remain idle throughout the round. Let $G$ be the first slot in which $V_1$ receives a response at the prescribed deadline, and set $G=\infty$ if no such response is received. If $V_1$ sends the instruction in slot $i$, there is not enough time for that instruction to reach the adversaries and for a response based on it to return by the deadline. Therefore, winning when $I_k=i$ requires $G=i$. Since the events $G=i$ are disjoint, the conditional success probability for a challenge from $V_1$ is at most $1/s$. For a challenge from $V_0$, we use the upper bound $1$. Hence the conditional probability of winning any round is at most $(1+1/s)/2\leq5/9$, and therefore $p_{\mathrm{win}}\leq(5/9)^m<\nu(\secparam)$ for all sufficiently large $\secparam$, because $m/\lceil\log_2^2\secparam\rceil\to\infty$. This contradicts the theorem hypothesis. The case in which every adversary lies to the right is symmetric.

By \Cref{prop:overlap-to-nonoverlap}, we may assume nonoverlapping slots without increasing the coalition's quantum-register size or decreasing its success probability. Since the reduction applies to every positive original slot spacing, the bound below also proves inter-slot-delay independence.

After each round, we add a classical copy of the requested label to the analysis record on the side to which it was sent. The reference qubit with label $R_k$ is measured if a response is tested and discarded unmeasured if no response arrives at the deadline.

For the analysis, let $F_k:=(I_k,B_k,U_k)$ be a classical register containing the verifier choices in round $k$. It is independent of the earlier state. We include it with $A_0$ only when calculating squashed entanglement; this does not give its contents to the adversaries. At each prescribed delivery time, $F_k$ and the list record determine exactly the instruction sent by the protocol, including $R_k$ when requested. The adversaries learn nothing else from $F_k$, and we discard it at the end of the round.

By the analysis conventions, fixing the complete classical record leaves the quantum state pure. Reference qubits measured in earlier rounds are then uncorrelated with the remaining systems and may be omitted. Hence, conditioned on the transcripts, the state on
\begin{equation}
    O\bigl(\mathbf r_k(\xi)\bigr)A_0A_1,
    \qquad
    O\bigl(\mathbf r_k(\xi)\bigr):=
    O_{r_{k,1}(\xi)}\cdots O_{r_{k,R_H}(\xi)},
\end{equation}
is pure. Thus every state used below is an average of pure states, one for each classical record, as required by \Cref{cor:k-extraction-entanglement-consumption-distributed}.

\textbf{Step 1: High-success rounds and slots that are good for both sides} We first define a high-success round and show that, in most slots of such a round, the adversaries succeed with high probability. Let $W_k$ indicate that the adversaries pass round $k$, and let $H_k$ be the event $W_1=\cdots=W_k=1$, with $H_0$ the certain event. For $i\in[s]$ and $b\in\{0,1\}$, set $p_{i,b}^k:=\Pr[W_k=1\mid I=i,B=b]$, where the probability also averages over the preceding transcript and the uniform position $U_k$. Thus $p^{(k)}:=\Pr[W_k=1]=\frac{1}{2s}\sum_{i=1}^s\sum_{b\in\{0,1\}}p_{i,b}^k$. We call round $k$ high-success if $p^{(k)}\geq1-10^{-7}$.

Now fix a high-success round and set $\gamma:=10^{-5}$. Call a challenge choice $(i,b)$ bad if $1-p_{i,b}^k>\gamma$, and good otherwise. By Markov's inequality, there are at most $2s\cdot10^{-7}/10^{-5}=0.02s$ bad choices. Call a slot good for both sides if neither of its two challenge choices is bad, i.e., both the challenge choices are good, and let $N$ be the number of such slots. Then $N\geq0.98s\geq0.8s$. Write their indices in increasing order as $\mathcal T=(t_1,\ldots,t_N)$.

\textbf{Step 2: Separating the two sides in one selected slot} We next analyze how the two sides $A_0,A_1$ can act in one of the good slots. Fix a high-success round $k$ and a good slot $t_j$. Let $\eta_j=(\eta_{j,0},\eta_{j,1})$ contain all classical information available to the two sides immediately before this slot, and let $\xi$ be its pre-round part. We fix $\eta_j$ only for the analysis; each side still uses only its own information and the messages it receives. We use the following conventions and timing facts.
\begin{enumerate}
\item For $a\in\{0,1\}$, let $\ell_a$ denote the return instruction sent by $V_a$ in this slot. If $V_a$ sends no return instruction, set $\ell_a=\mathsf{idle}$ and regard this as an empty instruction with the same timing as a return instruction. Thus $(\ell_0,\ell_1)$ covers both possibilities: either both verifiers are idle, or one sends a return instruction while the other is idle.
\item By the copying convention above, we may assume that an instruction from $V_a$ first reaches the coalition through $A_a$, and any information about it that later reaches $A_{1-a}$ can be regarded as a classical message from $A_a$.
\item Set $d_0:=L$ and $d_1:=R$. Information already available before the slot is handled by $\mathcal C_j$. The coalition cannot learn $\ell_b$ before time $T_{t_j}-d_b$. A response to $V_c$ is due at time $T_{t_j}+d_c$, so at most $d_b+d_c$ time is available. Information about $\ell_b$ cannot cross $0$ twice before affecting that response, because doing so would take more than $d_b+d_c$ time since no adversary is at $0$. Hence information learned from the current instructions may cross $0$ at most once before affecting a timely response.
\end{enumerate}
For each $a\in\{0,1\}$, let $S_a=1$ if $A_a$ sends exactly one fixed-size response timed to reach $V_a$ at the deadline for this slot, and let $Q_a$ be that response. Otherwise, let $S_a=0$ and let $Q_a$ be a fixed dummy state. The protocol's rejection rule handles early, late, multiple, and additional responses.

By the adversary model, the coalition's strategy is one fixed finite circuit, with its classical communication represented by classical wires between the adversaries. We consider only the operations that can affect whether either side sends the prescribed response for the current round or what state it sends. We divide these operations once in this circuit according to the information they use, rather than when they occur. Information carried through an earlier operation or stored in a register also counts as information used by a later operation.

\begin{claim}[Operations determining the round response]\label{claim:form-of-a-selected-slot}
The operations determining the prescribed response for the current round can be written as
\begin{equation}
\mathcal C_j\longrightarrow\mathcal M_{0,j}^{\ell_0}\otimes\mathcal M_{1,j}^{\ell_1}\longrightarrow\text{classical communication}\longrightarrow\mathcal U_{0,j}^{\ell_0}\otimes\mathcal U_{1,j}^{\ell_1}.
\end{equation}
The LOCC map $\mathcal C_j$ contains the operations affecting the round response that use only information already available before the slot, together with those that are independent of the current instructions. The local map $\mathcal M_{a,j}^{\ell_a}$ may use $\ell_a$, but it does not use information about the current instructions received from $A_{1-a}$. The two sides then communicate classically, and the local map $\mathcal U_{a,j}^{\ell_a}$ completes the response using the information received. No later communication can affect a response received by the prescribed deadline. Each $\mathcal U$-map outputs $(S_a,Q_a)$.
\end{claim}

\begin{proof}
We first place in $\mathcal C_j$ every operation affecting the round response that uses only information available before the slot or is independent of the current instructions, together with the operations needed to perform them. We do not change when a message is physically delivered: an operation is moved earlier in this description only if it does not use the message being passed over.

On each side, we combine the remaining operations that do not use, directly or indirectly, information about the current instructions received from the other side into $\mathcal M_{a,j}^{\ell_a}$. These two maps act on opposite sides, and neither uses information produced by the other, so they commute. Their classical outputs are then delivered to the other side.

The remaining local operations form the two $\mathcal U$-maps. Suppose some later communication could still affect a timely response. If the information it sends does not use anything received from the other side in the current slot, the sending operation belongs to the corresponding $\mathcal M$-map. If it does use such information, then information about a current instruction would cross $0$ twice before affecting the response, which the timing argument rules out. Thus no later communication is needed. If a timely response was prepared before the final local map, that map carries the already prepared response register unchanged to its output. Hence $(S_a,Q_a)$ is defined for every possible set of classical values, which proves the claim.
\end{proof}

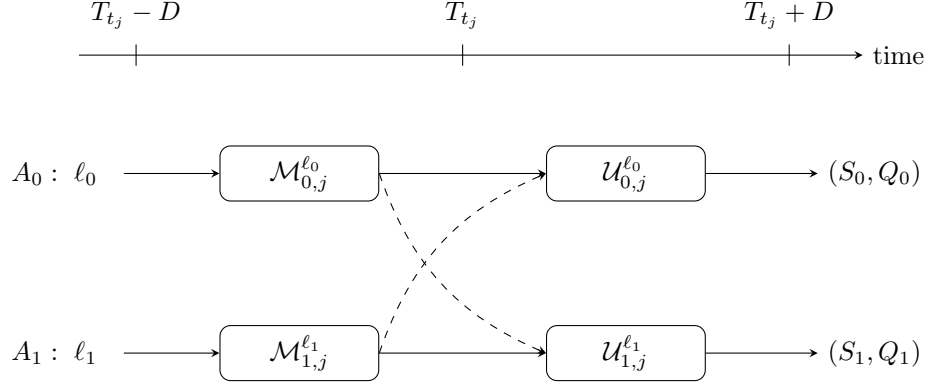
\begin{center}
\begin{tikzpicture}[x=1.08cm,y=1.08cm,>=stealth,every node/.style={font=\small}]
\draw[->] (-4.7,2.55)--(4.9,2.55) node[right] {time};
\foreach \x/\lab in {-4/{T_{t_j}-D},0/{T_{t_j}},4/{T_{t_j}+D}} {\draw (\x,2.42)--(\x,2.68) node[above=2pt] {$\lab$};}
\node[anchor=east] at (-4.35,1.1) {$A_0:\ \ell_0$};
\node[anchor=east] at (-4.35,-1.1) {$A_1:\ \ell_1$};
\node[draw,rounded corners,minimum width=2.1cm,minimum height=0.7cm] (m0) at (-2,1.1) {$\mathcal M_{0,j}^{\ell_0}$};
\node[draw,rounded corners,minimum width=2.1cm,minimum height=0.7cm] (m1) at (-2,-1.1) {$\mathcal M_{1,j}^{\ell_1}$};
\node[draw,rounded corners,minimum width=2.1cm,minimum height=0.7cm] (u0) at (2,1.1) {$\mathcal U_{0,j}^{\ell_0}$};
\node[draw,rounded corners,minimum width=2.1cm,minimum height=0.7cm] (u1) at (2,-1.1) {$\mathcal U_{1,j}^{\ell_1}$};
\node[anchor=west] (q0) at (4.35,1.1) {$(S_0,Q_0)$};
\node[anchor=west] (q1) at (4.35,-1.1) {$(S_1,Q_1)$};
\draw[->] (-4.15,1.1)--(m0.west);
\draw[->] (-4.15,-1.1)--(m1.west);
\draw[->] (m0.east)--(u0.west);
\draw[->] (m1.east)--(u1.west);
\draw[dashed,->] (m0.east) to[bend right=24] (u1.west);
\draw[dashed,->] (m1.east) to[bend left=24] (u0.west);
\draw[->] (u0.east)--(q0.west);
\draw[->] (u1.east)--(q1.west);
\end{tikzpicture}
\par\smallskip
{\small\emph{Before commutation.} The upper line shows the fixed protocol times. The lower drawing shows which information the operations use, not physical locations or exact execution times. Its two rows represent $A_0$ and $A_1$. Solid arrows show local processing, while dashed arrows show classical communication between the two sides.}
\end{center}

\textbf{Step 3: Commuting the two sides and comparing consecutive states}

By Claim~\ref{claim:form-of-a-selected-slot}, the two $\mathcal M$-maps commute. Delivering the classical outputs of either map also commutes with the receiving side's $\mathcal M$-map, because that map does not use those messages. We may therefore apply the maps in either order, while delivering each message before the $\mathcal U$-map that uses it.

Let $\sigma_j$ be the state after $\mathcal C_j$. For $b\in\{0,1\}$, let $\tau_{j,b}$ be the state obtained by starting from $\sigma_j$, giving $A_b$ its instruction $\ell_b$, and applying $\mathcal M_{b,j}^{\ell_b}$. When this map sends a classical message, we place a copy of that message in a new unread register on $A_{1-b}$ at its original delivery time. For this auxiliary state, we do not yet give $A_{1-b}$ its instruction or apply $\mathcal M_{1-b,j}^{\ell_{1-b}}$.

At $\tau_{j,1-a}$, side $A_a$ has not yet received its instruction, while the classical information sent to it by $A_{1-a}$ is already stored in its registers. For each list index $h$, let $E_{j,a,h}^{\xi}$ be the local response map obtained by giving $A_a$ the instruction $\mathsf{return}(r_{k,h}(\xi))$ and applying its remaining response operations. For this auxiliary map, no system is actually sent: the map outputs the unique timely response when $S_a=1$, outputs the dummy state otherwise, and locally discards everything else.\footnote{If the response state was prepared earlier, the map outputs that same register; it does not make a copy. The dummy state is used only to define the map when no unique timely response is produced.} The map may use the local information contained in $\eta_j$, but it does not read the actual value of $U_k$.

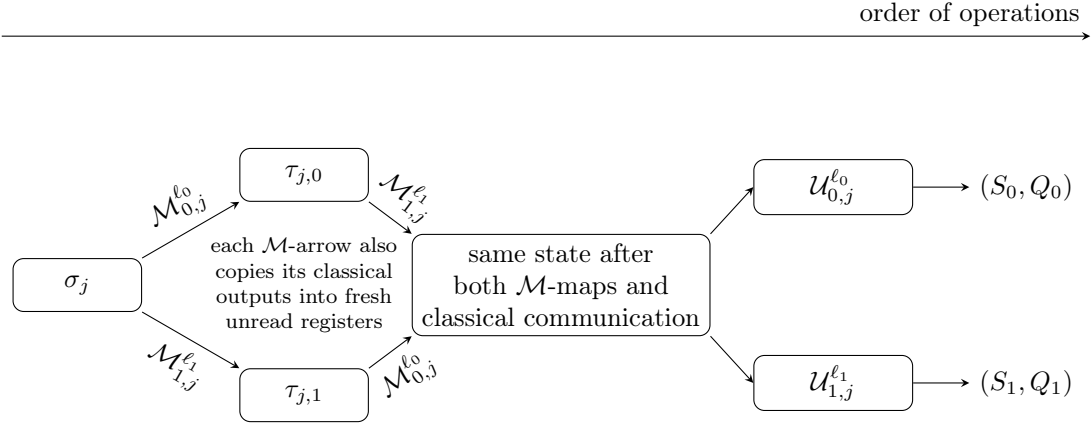
\begin{center}
\begin{tikzpicture}[x=1cm,y=1.08cm,>=stealth,every node/.style={font=\small}]
\draw[->] (-7,3.05)--(7.4,3.05) node[above left] {order of operations};
\node[draw,rounded corners,minimum width=1.7cm,minimum height=0.7cm] (s) at (-6,0) {$\sigma_j$};
\node[draw,rounded corners,minimum width=1.7cm,minimum height=0.7cm] (t0) at (-3,1.35) {$\tau_{j,0}$};
\node[draw,rounded corners,minimum width=1.7cm,minimum height=0.7cm] (t1) at (-3,-1.35) {$\tau_{j,1}$};
\node[draw,rounded corners,align=center,minimum width=3cm,minimum height=1.15cm] (f) at (0.4,0) {same state after\\both $\mathcal M$-maps and\\classical communication};
\node[draw,rounded corners,minimum width=2.1cm,minimum height=0.7cm] (u0) at (4,1.2) {$\mathcal U_{0,j}^{\ell_0}$};
\node[draw,rounded corners,minimum width=2.1cm,minimum height=0.7cm] (u1) at (4,-1.2) {$\mathcal U_{1,j}^{\ell_1}$};
\node[anchor=west] (q0) at (5.8,1.2) {$(S_0,Q_0)$};
\node[anchor=west] (q1) at (5.8,-1.2) {$(S_1,Q_1)$};
\draw[->] (s.north east)--node[above,sloped] {$\mathcal M_{0,j}^{\ell_0}$} (t0.south west);
\draw[->] (s.south east)--node[below,sloped] {$\mathcal M_{1,j}^{\ell_1}$} (t1.north west);
\draw[->] (t0.south east)--node[above,sloped] {$\mathcal M_{1,j}^{\ell_1}$} (f.north west);
\draw[->] (t1.north east)--node[below,sloped] {$\mathcal M_{0,j}^{\ell_0}$} (f.south west);
\node[font=\scriptsize,align=center,text width=3cm] at (-3,0) {each $\mathcal M$-arrow also copies its classical outputs into fresh unread registers};
\draw[->] (f.north east)--(u0.west);
\draw[->] (f.south east)--(u1.west);
\draw[->] (u0.east)--(q0.west);
\draw[->] (u1.east)--(q1.west);
\end{tikzpicture}
\par\smallskip
{\small\emph{After commutation.} The horizontal arrow shows only the order in which the operations are written, not physical time. The two orders of the $\mathcal M$-maps give the same state, after which the two $\mathcal U$-maps produce the responses. The states $\tau_{j,0}$ and $\tau_{j,1}$ are auxiliary states and need not occur at a common physical time.}
\end{center}

We now arrange the previously defined local response maps in the form required by \Cref{cor:k-extraction-entanglement-consumption-distributed}. To do so, we alternate the side on which the local response maps are defined for consecutive good slots. Set $a_j:=1$ when $j$ is odd and $a_j:=0$ when $j$ is even, and let $\tau_j:=\tau_{j,1-a_j}$. The maps $\{E_{j,a_j,h}^{\xi}\}_{h\in[R_H]}$ act locally on $A_{a_j}$ at $\tau_j$, and since $a_{j+1}=1-a_j$, the corresponding maps at $\tau_{j+1}$ act on the opposite side. Let $\rho_j^{i,b}$ denote the state at $\tau_j$ conditioned on the challenge slot being $i$ and the selected verifier being $V_b$. Only $(i,b)$ is fixed in this notation; $U_k$ and the earlier classical records remain random.

We next identify the challenge choices for which the coalition evolves from $\tau_j$ to $\tau_{j+1}$ by LOCC without a verifier test. For $j<N$, let
\begin{equation}
\mathcal S_j:=\{(i,b):i>t_{j+1},\ b\in\{0,1\}\}\cup\{(t_{j+1},a_{j+1})\}.
\end{equation}
These choices either place the challenge after slot $t_{j+1}$ or place it in slot $t_{j+1}$ on the side that has not yet received its instruction at $\tau_{j+1}$. Hence $|\mathcal S_j|=2s-2t_{j+1}+1$.

For every choice in $\mathcal S_j$, no return instruction has arrived before the state being considered. The earlier idle slots therefore reveal no information about $U_k$. Since $U_k$ was sampled independently, it remains uniform after fixing $\eta_j$ or $\eta_{j+1}$.

For every $(i,b)\in\mathcal S_j$, let $\Lambda_j^{i,b}$ denote the evolution from $\rho_j^{i,b}$ to $\rho_{j+1}^{i,b}$ obtained as follows. Starting from $\tau_j$, we finish slot $t_j$, follow the coalition's strategy until slot $t_{j+1}$, apply $\mathcal C_{j+1}$ followed by $\mathcal M_{1-a_{j+1},j+1}^{\mathsf{idle}}$, and stop at $\tau_{j+1}$. Note that for every $(i,b)\in\mathcal S_j$, slot $t_j$ is idle, while at slot $t_{j+1}$ either both sides are idle or we stop before the selected side receives its instruction. Thus $\Lambda_j^{i,b}$ is LOCC, contains no verifier test or reveal, and preserves $U_k$, the reference qubits, and the records determining the list order.

Finally, apart from the fixed local labels $(i,b)$, the states at $\tau_j$ and $\tau_{j+1}$ are the same for every $(i,b)\in\mathcal S_j$. At $\tau_j$, side $A_{1-a_j}$ has received $\mathsf{idle}$ and $A_{a_j}$ has not yet received its instruction, exactly as for the good challenge choice $(t_j,a_j)$. The same reasoning at $\tau_{j+1}$ gives the good challenge choice $(t_{j+1},a_{j+1})$. Hence
\begin{equation}
\rho_j^{i,b}=\rho_j^{t_j,a_j},\qquad \rho_{j+1}^{i,b}=\rho_{j+1}^{t_{j+1},a_{j+1}}.
\end{equation}

We now show that the maps $E_{j,a,h}^{\xi}$ succeed with high probability. We fix $j<N$ and $a\in\{0,1\}$ and compare the two good challenge choices $(t_j,a)$ and $(t_{j+1},a)$, which use the same verifier. Let $W$ be the event that the first challenge is won, and let $K_{j,a}$ be the event that $A_{1-a}$ sends a fixed-size response timed to reach $V_a$ at the deadline of slot $t_j$, whether or not that slot is selected.

In both executions, all slots before $t_j$ are idle, so the record $\eta_j$ has the same distribution. We fix $\eta_j$ and leave all later outcomes random. The probability of $K_{j,a}$ is then the same in both executions: $A_{1-a}$ receives the idle instruction in slot $t_j$ in both cases, and information about the instruction received by $A_a$ cannot travel to the opposite side and back before the deadline. If $K_{j,a}$ occurs in the second execution, $V_a$ did not request that response, so the protocol rejects. Since $(t_{j+1},a)$ is good, it follows that $\Pr[K_{j,a}]\leq\gamma$.

The challenge choice $(t_j,a)$ is also good, so $\Pr[W]\geq1-\gamma$ and therefore $\Pr[W\cap K_{j,a}^c]\geq1-2\gamma$. On this event, the unique accepted response is sent by $A_a$, and $E_{j,a,U_k}^{\xi}$ outputs that response. Let $q_{j,a,h}^{\eta_j}$ denote the acceptance probability when the output of $E_{j,a,h}^{\xi}$ is tested against $O_{r_{k,h}(\xi)}$, where $\xi$ is the pre-round part of $\eta_j$. Since $U_k$ remains uniform and the maps do not read its actual value,
\begin{equation}
\frac{1}{R_H}\sum_{h=1}^{R_H}\mathbb E_{\eta_j}\!\left[q_{j,a,h}^{\eta_j}\right]=\mathbb E_{\eta_j,U_k}\!\left[q_{j,a,U_k}^{\eta_j}\right]\geq1-2\gamma.
\end{equation}
Thus the required local response maps exist at every good slot except the last. For $1\leq j<N-1$, the maps on $A_{a_j}$ for $\rho_j^{i,b}$ and on $A_{a_{j+1}}$ for $\rho_{j+1}^{i,b}$ both have this average success probability and act on opposite sides. By \Cref{lem:bb84_fidelity}, both families have average Bell-state fidelity at least $1-4\gamma$. Let $\epsilon:=2\sqrt\gamma$ and $c:=1-2g(\epsilon)-f(\epsilon)$. For $\gamma=10^{-5}$, we have $c>0.74$.

As established before Step~1, fixing the complete classical record leaves the remaining quantum state pure. For each such record, its earlier part $\xi$ fixes the ordered reference qubits; let $O_h:=O_{r_{k,h}(\xi)}$. The channel $\Lambda_j^{i,b}$ preserves these qubits and their order, and the two families of maps are averaged over the same records and the same uniform $h$. Applying \Cref{cor:k-extraction-entanglement-consumption-distributed}, with the two sides interchanged when $j$ is odd, yields, for every $1\leq j<N-1$ and $(i,b)\in\mathcal S_j$,
\begin{equation}
\Esq(A_0:A_1)_{\rho_j^{i,b}}-\Esq(A_0:A_1)_{\rho_{j+1}^{i,b}}\geq R_Hc.
\end{equation}

\textbf{Step 4: Adding the entanglement decrease over one complete round}

We now add the decreases from Step~3 over one high-success round. We let $\omega_0$ be the adversarial state after the initialization phase has ended and no initialization qubit remains in transit, but before the variables of the first round are sampled. For $k\geq1$, we let $\omega_k$ be the state after all $s$ slots of round $k$, but before the fresh qubit for round $k+1$ is sent. We set $\mathcal W_k:=\Esq(A_0:A_1)_{\omega_k}$ and fix a high-success round $k$.

For $k>1$, we let $Q_k$ be the fresh qubit sent by $V_0$ before round $k$. After tracing out the reference qubit kept by the verifier, $Q_k$ is independent of the adversaries' previous state, so initially adding it to either side does not change the squashed entanglement. By condition~\ref{cond:dynamic-input-path}, it can cross from $A_0$ to $A_1$ at most once, increasing the squashed entanglement by at most one.\footnote{For every qubit $Q$, $E_{\rm sq}(A_0:QA_1)\leq E_{\rm sq}(A_0Q:A_1)+1$. For any extension $E$, the corresponding difference of conditional mutual information is $H(Q\mid A_1E)-H(Q\mid A_0E)\leq2$, and the definition of squashed entanglement includes a factor of $1/2$.} All its other possible changes before the slots begin are local and cannot increase the squashed entanglement. If $\widehat\omega_{k-1}$ is the state after this transmission has ended, then
\begin{equation}
\Esq(A_0:A_1)_{\widehat\omega_{k-1}}\leq\mathcal W_{k-1}+1.
\end{equation}
For $k=1$, we set $\widehat\omega_0:=\omega_0$, and equality holds without the additional one.

Since the (ordered) list changes after every round, the physical label in each list index depends on the earlier record $\xi$. For the analysis, we copy to both sides the parts of $\xi$ that determine $\mathbf r_k(\xi)$. Once $\xi$ is fixed, position $h$ corresponds to the fixed reference qubit $O_{r_{k,h}(\xi)}$, while $U_k$ remains uniform and independent of the earlier state. By classical copying, local deletion, and \Cref{lemma:direct-sum}, these additional copies do not change the squashed entanglement. Thus both sides use the same identification of list indices with physical reference qubits throughout Step~3.

Recall that $F_k=(I_k,B_k,U_k)$ contains the fresh choices for round $k$. It is independent of $\widehat\omega_{k-1}$, so including it with $A_0$ for the entanglement calculation does not change the squashed entanglement. At the end of the round, tracing out $F_k$ and the additional copies of the list-records locally gives $\omega_k$ and cannot increase the squashed entanglement. The only other possible increase comes from the verifier's test when a response arrives or from the additional records kept for the analysis when it does not. The following lemma bounds both cases.

\begin{lemma}\label{lem:verifier-entanglement-accounting}
Suppose that $O$ contains at most $R_H$ verifier reference qubits and that $OA_0A_1$ is pure after the earlier classical records are fixed. Let $\rho^-$ be the adversarial state immediately before the verifier acts, with any response waiting or traveling to the verifier still included on its sender's side. Let $\rho^+$ be the state after the respective verifier acts at the deadline of the challenge slot. Then $\Esq(A_0:A_1)_{\rho^+}\leq\Esq(A_0:A_1)_{\rho^-}+R_H$\footnote{This holds irrespective of whether a response is sent from either side or no response is sent.}.
\end{lemma}
The proof is given on \Cpageref{pf:lem:verifier-entanglement-accounting}.

We now fix a challenge choice $(i,b)$. Whenever $1\leq j<N-1$ and $(i,b)\in\mathcal S_j$, Step~3 shows that the transition from $\tau_j$ to $\tau_{j+1}$ decreases the squashed entanglement by at least $R_Hc$. Every other change during the round uses local operations and classical communication, except for the verifier's action and the additional records kept for the analysis. Receiving and discarding any additional responses cannot increase the squashed entanglement. The verifier tests at most one response, and \Cref{lem:verifier-entanglement-accounting} bounds the possible increase, including the analysis records, by $R_H$. Therefore, for this fixed challenge choice, the total decrease is at least
\begin{equation}
\sum_{\substack{1\leq j<N-1\\(i,b)\in\mathcal S_j}}R_Hc-R_H.
\end{equation}

Applying \Cref{lemma:direct-sum} to the classical register containing $(I_k,B_k)$, we average this inequality over the $2s$ equally likely choices of $(i,b)$. After also taking into account, the possible increase of squashed entanglement by one unit, caused by the replenishment qubit sent by $V_0$ before the start of round $k$, we obtain
\begin{equation}
\mathcal W_{k-1}-\mathcal W_k\geq\frac{R_Hc}{2s}\sum_{j=1}^{N-2}|\mathcal S_j|-R_H-1.
\end{equation}
We now lower-bound the sum. Since $t_1<\cdots<t_N$ are distinct elements of $[s]$, we have $t_{j+1}\leq s-N+j+1$ and hence $|\mathcal S_j|\geq2N-2j-1$. Therefore, $\sum_{j=1}^{N-2}|\mathcal S_j|\geq N(N-2)$.

We set $\kappa:=0.09>0.08$. Since $N\geq0.98s$ and $c>0.74$, we have
\begin{equation}
\frac{cN(N-2)}{2s}-0.18(s-2.5)>0.74(0.4802s-0.98)-0.18(s-2.5)=0.175348s-0.2752>1,
\end{equation}
where the last inequality holds for every integer $s\geq9$. Since $R_H\geq1$, this implies
\begin{equation}
R_H\left(\frac{cN(N-2)}{2s}-0.18(s-2.5)\right)>1.
\end{equation}
Hence, every high-success round satisfies
\begin{equation}
\mathcal W_{k-1}-\mathcal W_k>R_H\bigl(2\kappa(s-2.5)-1\bigr).
\end{equation}
For every integer $s\geq9$, the coefficient on the right is positive.

\textbf{Step 5: Finding consecutive high-success rounds} It remains to find a long block of high-success rounds to which the preceding bound applies. However, the original execution need not contain a long consecutive block in which every round is won with high probability. We therefore find a block that is jointly won with high probability after conditioning on the earlier rounds, and then apply Step~4 to an execution that starts from this conditioned state. This changes the starting state but not the adversaries' quantum-register size.

Partition the $m$ rounds into $\lceil\log_2^2(\secparam)\rceil$ consecutive blocks $\mathcal B_1,\ldots,\mathcal B_{\lceil\log_2^2(\secparam)\rceil}$ whose lengths differ by at most one. Thus every block has length at least $\lfloor m/\lceil\log_2^2(\secparam)\rceil\rfloor$. For endpoints $0=n_0<n_1<\cdots<n_{\lceil\log_2^2(\secparam)\rceil}=m$, where $\mathcal B_h=\{n_{h-1}+1,\ldots,n_h\}$, set
\begin{equation}
p_h:=\Pr[H_{n_h}\mid H_{n_{h-1}}].
\end{equation}
The chain rule gives $p_{\mathrm{win}}=\prod_{h=1}^{\lceil\log_2^2(\secparam)\rceil}p_h$. Hence some block $\mathcal B=\mathcal B_h$ satisfies
\begin{equation}
p_h\geq p_{\mathrm{win}}^{1/\lceil\log_2^2(\secparam)\rceil}\geq\left((1-10^{-7})^{\lceil\log_2^2(\secparam)\rceil}\right)^{1/\lceil\log_2^2(\secparam)\rceil}=1-10^{-7}.
\end{equation}

Fix this block and condition once on $H_{n_{h-1}}$. We let $\rho_{n_{h-1},\mathrm{win}}$ be the resulting state at the beginning of $\mathcal B$, and then run every round in $\mathcal B$ without further conditioning. Since we keep all classical records, fixing those records still leaves a pure quantum state, and the conditioning does not change any register dimension. Thus Step~4 applies to this execution. Moreover, winning the entire block implies winning each of its rounds, so every round in this execution is won with probability at least $p_h\geq1-10^{-7}$. Therefore, every round is high-success.

Let $\widetilde{\mathcal W}_q$ be $\Esq(A_0:A_1)$ after $q$ rounds of this execution. By Step~4, for every $q\in[|\mathcal B|]$,
\begin{equation}
\widetilde{\mathcal W}_{q-1}-\widetilde{\mathcal W}_q>R_H\bigl(2\kappa(s-2.5)-1\bigr).
\end{equation}
Summing these decreases and using the nonnegativity of squashed entanglement gives
\begin{align}
\Esq(A_0:A_1)_{\rho_{n_{h-1},\mathrm{win}}}
&=\widetilde{\mathcal W}_0 \nonumber\\
&>|\mathcal B|R_H\bigl(2\kappa(s-2.5)-1\bigr) \nonumber\\
&\geq\left\lfloor\frac{m}{\lceil\log_2^2(\secparam)\rceil}\right\rfloor R_H\bigl(2\kappa(s-2.5)-1\bigr).
\end{align}
By \Cref{lem:quantum_dim_bound}, this is also a strict lower bound on the total size of the adversarial coalition's quantum registers for the truncated execution with the post-selected state. Note that post-selection changes only the state stored in the adversarial coalition's fixed registers; it neither adds new registers nor changes their sizes. Moreover, any verifier qubit kept by an adversary must occupy one of these counted registers. Hence the same register-size lower bound applies to the original coalition. Thus every coalition succeeding with probability at least $\nu(\secparam)$ exceeds the bound stated in the theorem, which completes the proof.
\end{proof}

\begin{remark}[Concrete security]\label{rem:concrete-security}
The proof also gives a finite-parameter bound for any integer $1\leq k<m$. For positive integers $m,s,R_H$ with $s\geq9$, any finite LOCC adversarial coalition with no member at the claimed location that succeeds with probability at least $(1-10^{-7})^k$ must have total quantum-register size
\[
Q>\left\lfloor\frac{m}{k}\right\rfloor R_H\bigl(2\kappa(s-2.5)-1\bigr),
\qquad \kappa=0.09,
\]
as chosen in the proof. Indeed, partition the $m$ rounds into $k$ consecutive blocks whose lengths differ by at most one, and apply the same argument. The choice $k=\lceil\log_2^2\secparam\rceil$ in \Cref{thm:amplified_security-generalized} ensures that the success threshold is negligible.
\end{remark}

\paragraph{Generalizing the proof to \texorpdfstring{$d$}{d}-dimensions}
\label{sec:higher-dimensions}

We now give a generic lift from one-dimensional static-verifier security to $d$-dimensional mobile-verifier security.  Applying it to \Cref{thm:amplified_security-generalized} proves mobile-verifier security of the amplified protocol in every fixed dimension.

As specified in the adversary model after \Cref{def:qpv}, every outgoing quantum register remains counted with its sender and unavailable for reuse until the verifier receives or discards its state.

\begin{theorem}[Mobile-Verifier Security of the Protocol]
\label{thm:generic-reduction}
Let $d\geq2$ be fixed, and let $m$, $s$, and $R_H$ satisfy the assumptions of \Cref{thm:amplified_security-generalized}. The $d$-dimensional mobile-verifier version of the protocol in \Cref{fig:qpv-protocol-multi-round-unified}, defined in \Cref{def:qpv-d-dimensional-generalization}, has
\begin{equation}
\left\lfloor\frac{m}{\lceil\log_2^2(\secparam)\rceil}\right\rfloor R_H\bigl(2\kappa(s-2.5)-1\bigr)
\end{equation}
-qubit mobile-verifier security, where $\kappa>0.08$ is the universal constant from \Cref{thm:amplified_security-generalized}. In particular, it has $\Omega\!\left(smR_H/\log^2\secparam\right)$-qubit mobile-verifier security.
\end{theorem}
In the proof of \Cref{thm:generic-reduction}, the reduction from one-dimensional security to higher-dimensional security only requires the following two properties of the protocol.
\begin{enumerate}
\item\label{cond:dynamic-input-path-theorem-5} Each initialization or replacement qubit sent by $V_0$ is sent along the fixed one-way path from $V_0$, through the claimed location, to $V_1$. Before beginning each round, the verifiers wait until the time at which the qubits sent by $V_0$ would reach $V_1$. If $V_1$ receives any of these qubits, the verifiers abort, since the honest prover should have intercepted and stored them.
\item\label{cond:dynamic-response-order} The prescribed fixed-size responses are the only prover-to-verifier quantum messages that can affect acceptance, and every other prover-to-verifier quantum message is ignored or causes rejection. At the prescribed deadline, the selected verifier tests the response if it arrives then and otherwise rejects. If either verifier receives an additional quantum response during a timed round, the verifiers reject. The deadline action and any prescribed reveal finish before $V_0$ sends another input qubit and before the next timed round begins.
\end{enumerate}
\begin{remark}
More generally, let $\Pi$ be any one-dimensional two-verifier QPV protocol with claimed location $0$ and verifiers at $-R$ and $R$. If $\Pi$ satisfies conditions~\ref{cond:dynamic-input-path-theorem-5} and~\ref{cond:dynamic-response-order} and has $B(\secparam)$-qubit static-verifier security against arbitrary finite one-dimensional coalitions in the adversary model of \Cref{sec:adversary-model}, then its $d$-dimensional mobile-verifier version from \Cref{def:qpv-d-dimensional-generalization} has $B(\secparam)$-qubit mobile-verifier security against arbitrary finite $d$-dimensional coalitions in the same model for every fixed $d\geq2$.
\end{remark}

\begin{proof}[Proof of \Cref{thm:generic-reduction}]
We prove the theorem by contraposition. Let $\Pi$ be the one-dimensional protocol in \Cref{fig:qpv-protocol-multi-round-unified}, let $\Pi'$ be its $d$-dimensional mobile-verifier version, and let $B(\secparam)$ be the bound stated in the theorem. Suppose that there is a finite $d$-dimensional adversarial coalition $\widetilde{\mathcal A}$ against $\Pi'$ whose total quantum-register size is
\begin{equation}
Q\leq B(\secparam)
\end{equation}
and whose success probability is some non-negligible function $p=p(\secparam)$. We construct a finite one-dimensional adversarial coalition against $\Pi$ with the same quantum-register size $Q$ and success probability at least $p$.

Fix the claimed location for this execution. By translating and rotating coordinates, assume without loss of generality that the claimed location is the origin $\vec0\in\mathbb R^d$ and that the two verifiers lie on the first coordinate axis at
\begin{equation}
    V_0=(-R,\vec0)
    \qquad\text{and}\qquad
    V_1=(R,\vec0),
\end{equation}
where $R>0$. Write the location of every adversary $P_i$ as
\begin{equation}
    P_i=(x_i,\vec y_i)
    \in\mathbb R\times\mathbb R^{d-1}.
\end{equation}
Since $\widetilde{\mathcal A}$ is a valid coalition, $(x_i,\vec y_i)\neq(0,\vec0)$ for every $i$.

For $\varepsilon\in\mathbb R$, define
\begin{equation}
    F_\varepsilon(x,\vec y)
    :=x+\varepsilon\|\vec y\|_2^2.
\end{equation}
In particular,
\begin{equation}
    F_\varepsilon(V_0)=-R,
    \qquad
    F_\varepsilon(V_1)=R,
    \qquad
    F_\varepsilon(\vec0)=0.
\end{equation}
Thus, $F_\varepsilon$ preserves the two verifier locations and the claimed location.

We show that $\varepsilon$ can be chosen so that no distance between any two parties increases. Let $\mathcal S$ be the finite set consisting of the claimed location, the two verifier locations, and all adversarial locations. For any
\begin{equation}
    P=(x,\vec y),
    \qquad
    P'=(x',\vec y')
    \in\mathcal S,
\end{equation}
we have at $\varepsilon=0$
\begin{equation}
    |F_0(P)-F_0(P')|
    =
    |x-x'|
    \leq
    \sqrt{(x-x')^2+\|\vec y-\vec y'\|_2^2}
    =
    \|P-P'\|_2.
\end{equation}
If this inequality is strict, continuity shows that it remains strict for all sufficiently small $|\varepsilon|$. If equality holds at $\varepsilon=0$, then $\vec y=\vec y'$, and therefore
\begin{equation}
    F_\varepsilon(P)-F_\varepsilon(P')
    =
    x-x'
\end{equation}
for every $\varepsilon$. Since $\mathcal S$ is finite, there is a common $\varepsilon_*>0$ for which the distance inequality holds for every pair in $\mathcal S$ whenever $|\varepsilon|<\varepsilon_*$.

It remains to ensure that no adversary is mapped to the claimed location. If $\vec y_i=\vec0$, then $x_i\neq0$, and $F_\varepsilon(P_i)=x_i\neq0$. If $\vec y_i\neq\vec0$, then $F_\varepsilon(P_i)=0$ can hold only for
\begin{equation}
    \varepsilon=-\frac{x_i}{\|\vec y_i\|_2^2}.
\end{equation}
There are only finitely many such values. We may therefore choose a nonzero $\varepsilon$ with $|\varepsilon|<\varepsilon_*$ while avoiding all of them. Fix such a value. Then
\begin{equation}
    F_\varepsilon(P_i)\neq0
    \qquad\text{for every adversary }P_i,
\end{equation}
and no relevant distance has increased.

We now define the one-dimensional coalition $\mathcal A$. For every adversary $P_i$ of $\widetilde{\mathcal A}$, place a corresponding adversary on the verifier line at coordinate $F_\varepsilon(P_i)$. It holds the same quantum and classical registers, begins with the same joint state, and runs the same local strategy. If several adversaries are mapped to the same location, they may remain separate or be merged while keeping their registers separate. This does not increase the total quantum-register size. No member of $\mathcal A$ is at the claimed location.

The distance inequality allows $\mathcal A$ to reproduce the timed execution of $\widetilde{\mathcal A}$. Every message between adversaries is classical. The mapped sender transmits it at the original time, and if it arrives early, the receiver stores it without using it until its original arrival time. We treat every public classical verifier message in the same way. A directed slot instruction travels along the verifier axis, so every original recipient lies on that axis and is fixed by $F_\varepsilon$. An adversary newly mapped onto the instruction path simply ignores an instruction that its original strategy did not receive. Whenever the simulated strategy sends a classical message to a verifier, the mapped sender delays it by the decrease in travel time so that the verifier receives it at the original time.

For every $P_i$, let $A_i^{\mathrm q}$ contain its local and outgoing quantum registers, and let their total size be $Q_i$ qubits. Thus
\begin{equation}
    \sum_iQ_i=Q.
\end{equation}

Consider any verifier-to-prover quantum input, including every initialization input and every later replacement. Its path is the verifier-axis segment from $V_0$ to $V_1$. Every original adversary that acts on the traveling qubit therefore lies at $(x_i,\vec0)$ and satisfies $F_\varepsilon(x_i,\vec0)=x_i$. Its location, registers, interaction time, and operation on the qubit are consequently unchanged. If several adversaries act on the same qubit while allowing it to continue, they do so in the same order and at the same times in both executions. An off-axis adversary that is mapped onto the segment simply ignores the qubit. The qubit is stopped in the mapped execution exactly when it was stopped in the original execution; otherwise, it reaches $V_1$ at the same time and causes rejection in both.

Now consider a quantum message sent by $P_i$ to a verifier. Let $t_0$ be the time at which its state enters an outgoing register, let $t_1$ be its launch time, and let $t_2$ be the time at which that register is released. Let $\delta\geq0$ be the decrease in the distance from $P_i$ to the verifier after applying $F_\varepsilon$. The mapped adversary places the same state in the same outgoing register at time $t_0$, so its local source register is freed at the original time. If the original message reaches the verifier, the mapped adversary keeps the state unchanged until time $t_1+\delta$ and then launches it. The shorter path makes it reach the verifier and release the outgoing register at time $t_2$. If the original message is discarded without being received, the mapped adversary keeps it unchanged and discards it at time $t_2$. Thus the outgoing register is occupied and released at exactly the same times as in the original execution.

We apply this construction independently to every outgoing quantum message. Since each message uses the same counted register as before, no additional quantum register is required, even when several messages are waiting or traveling at the same time. The mapped execution reproduces every verifier receipt, discard, and decision, so the one-dimensional adversarial coalition succeeds with probability at least $p$.

Thus, the constructed adversarial coalition uses the same $Q$ qubits as the original adversarial coalition and has success probability at least $p$. Since $Q\leq B(\secparam)$, this contradicts the static-verifier security established in \Cref{thm:amplified_security-generalized}, which completes the proof of the theorem.
\end{proof}

\fi

\section{Security Proof of the QPV Protocol in the Continuous-Time Model}
\label{sec:continuous-time-model}

This section derives the continuous-time result (stated informally in \Cref{thm:continuous-main-result-informal}) from \Cref{thm:amplified_security-generalized}. The result does not follow merely by taking the limit $s\to\infty$, because for each finite $s$, \Cref{thm:amplified_security-generalized} considers a different protocol, and an attack on the continuous-time protocol need not arise as a limit of attacks on these finite protocols. 
Instead, we show that any attack on the continuous-time protocol gives, for every finite $s$, an attack with at least the same success probability on a suitable set of $s$ challenge times. Hence, given any fixed finite adversarial quantum-register size, by setting $s$ large enough, the lower bound in \Cref{thm:amplified_security-generalized} eventually exceeds the adversarial quantum-register size. 
 Throughout this section, $m$ and $R_H$ satisfy the hypotheses of that theorem.

\paragraph{The continuous-time protocol.} Recall from condition~\ref{cond:slot-timing} that a slot is specified by the time at which the selected verifier's instruction reaches the claimed location. We refer to this as the \emph{slot time}. The instruction's sending time and the response deadline are determined relative to this time, as in the original protocol. In particular, a slot time is not the beginning of the entire transmission and response interval, and the intervals for different slots may overlap.
\begin{itemize}
\item Each round $k$ has a fixed unit-length window $W_k=[a_k,a_k+1)$ of slot times. Every time in $W_k$ specifies a slot from which the verifiers may choose. All verifier inputs for round $k$ finish traveling along their fixed paths, and their round variables are sampled, before $a_k-D$, where $D=\max\{L,R\}$.
\item The verifiers choose the challenge verifier and the requested label exactly as in \Cref{fig:qpv-protocol-multi-round-unified}. They choose $T_k$ uniformly from $W_k$ and use the slot with time $T_k$. The selected verifier sends its instruction so that it reaches the claimed location at $T_k$, and the honest prover immediately returns the requested state. No return instruction is sent for any other slot in the round.
\item After the window ends, the protocol allows enough time for the response and the remaining actions of the round to finish, even when $T_k$ is arbitrarily close to the end of $W_k$. Any fresh input for round $k+1$ also finishes traveling before $a_{k+1}-D$.
\end{itemize}
Thus, each round has infinitely many possible challenge slots, but only one is selected.\footnote{The protocol requires exact sampling and use of real-valued slot times. It is therefore an ideal mathematical model rather than an efficient implementation under \Cref{def:qpv}.}

\paragraph{Adversarial strategies.} For each fixed $\secparam$, one finite circuit describes the complete adversarial strategy. The same circuit is used for every choice of challenge slots and every possible classical record. Earlier messages, outcomes, and arrival times may determine which operations are used and when they are used, but they do not add operations or communication wires. The circuit has finitely many quantum registers, and communication between adversaries remains classical. No upper bound on the finite size of the circuit or its quantum registers is fixed in advance. To define the average success probability, we assume that the success probability is a measurable function of the selected slot times.

We use the success threshold and register-size bound from \Cref{thm:amplified_security-generalized}. Let $\nu(\secparam):=(1-10^{-7})^{\lceil\log_2^2(\secparam)\rceil}$ and, for every finite integer $s\geq9$, let
\begin{equation}
B_s:=\left\lfloor\frac{m}{\lceil\log_2^2(\secparam)\rceil}\right\rfloor R_H\bigl(2\kappa(s-2.5)-1\bigr),
\end{equation}
where $\kappa>0.08$ is the constant from \Cref{thm:amplified_security-generalized}.

We first restrict each round $k$ to any fixed set of $s$ distinct, equally spaced slot times in $W_k$ and choose the challenge slot uniformly from them. Each selected slot uses the sending time, response deadline, and verifier operation in the continuous-time protocol. All other protocol actions remain unchanged.

\begin{lemma}[Security for equally spaced slots]
\label{lem:arbitrary-finite-schedules}
There exists $\secparam_0$, independent of $s$ and the chosen slot times, such that the following holds for every $\secparam\geq\secparam_0$ and every finite integer $s\geq9$. Any finite-circuit adversarial coalition that succeeds in the resulting protocol with probability at least $\nu(\secparam)$ must have total quantum-register size strictly greater than $B_s$.
\end{lemma}

\begin{proof}
The proof of \Cref{thm:amplified_security-generalized} applies to these equally spaced slots and gives the bound $B_s$ for every finite $s\geq9$. It does not require $s$ to be polynomially bounded.

\end{proof}

\begin{theorem}[Security in the ideal continuous-time model]
\label{thm:continuous-time-security}
For every $\secparam\geq\secparam_0$, every finite-circuit LOCC adversarial coalition against the continuous-time protocol has success probability smaller than $\nu(\secparam)$, regardless of its finite quantum-register size. Thus the protocol is secure against any finite amount of preshared entanglement, with no bound on that amount fixed in advance.
\end{theorem}

\begin{proof}
Fix $\secparam\geq\secparam_0$ and an adversarial coalition whose total quantum-register size is a finite number $Q$. For $t_k\in[0,1)$, let $f(t_1,\ldots,t_m)$ be the coalition's success probability when the selected slot in round $k$ has time $a_k+t_k$, averaged over all remaining randomness. Let $p$ be its success probability when each round's slot time is chosen independently and uniformly from $W_k$.

Fix any finite integer $s\geq9$. The key observation is that a uniform slot time in $W_k$ can be chosen in two steps.
\begin{enumerate}
\item Choose $U_k$ uniformly from $[0,1/s)$. This fixes $s$ equally spaced slots with times $a_k+U_k,a_k+U_k+1/s,\ldots,a_k+U_k+(s-1)/s$.
\item Independently choose $J_k$ uniformly from $\{0,\ldots,s-1\}$ and select the slot with time $a_k+U_k+J_k/s$.
\end{enumerate}
For each fixed $J_k=j$, the selected slot time is uniform in $[a_k+j/s,a_k+(j+1)/s)$. These $s$ intervals partition $W_k$ and are chosen with equal probability, so the selected slot time is uniform in the whole window. Applying this procedure independently in all $m$ rounds gives
\begin{equation}
p=\mathbb E_{U_1,\ldots,U_m}\left[\frac{1}{s^m}\sum_{j_1,\ldots,j_m=0}^{s-1}f\left(U_1+\frac{j_1}{s},\ldots,U_m+\frac{j_m}{s}\right)\right].
\end{equation}
For fixed values of $U_1,\ldots,U_m$, the expression inside the expectation is the coalition's success probability when each round is restricted to the corresponding $s$ slots. Its average is $p$, so some fixed choice $u_1,\ldots,u_m$, with each $u_k\in[0,1/s)$, gives success probability at least $p$. We fix this choice. In round $k$, the selected slot is now uniform among those with times
\begin{equation}
a_k+u_k,\ a_k+u_k+\frac{1}{s},\ \ldots,\ a_k+u_k+\frac{s-1}{s}.
\end{equation}
Only the choice of challenge slot has changed. For each selected slot, the instruction is sent at the same time, the response has the same deadline, and the verifier's operation occurs at the same time as in the continuous-time protocol. Replenishment and all actions between rounds remain unchanged. The same adversarial circuit therefore attacks this protocol with $s$ slots per round, using the same quantum registers and succeeding with probability at least $p$.

Suppose that $p\geq\nu(\secparam)$. By \Cref{lem:arbitrary-finite-schedules}, we must have $Q>B_s$. This argument applies for every finite $s\geq9$, although the chosen slots may differ with $s$. The register size $Q$ remains fixed, whereas $B_s$ grows without bound as $s$ grows. We may therefore choose a finite $s$ for which $B_s\geq Q$, giving a contradiction. Hence $p<\nu(\secparam)$.
\end{proof}

\begin{corollary}[Higher-dimensional continuous-time security]
\label{cor:continuous-time-higher-dimensional}
For every fixed dimension $d\geq2$, the mobile-verifier generalization of the ideal continuous-time protocol has the same security guarantee as \Cref{thm:continuous-time-security}.
\end{corollary}

\begin{proof}
The continuous-time protocol changes only the choice of challenge slot. The verifier-input paths, the timing of instructions and responses within each slot, and the order of all other messages remain as in \Cref{fig:qpv-protocol-multi-round-unified}. We therefore apply the entire argument from the proof of \Cref{thm:generic-reduction}, including the part that places adversaries away from the verifier line onto that line and adjusts the times at which they act. This gives a one-dimensional adversarial coalition with the same quantum-register size and at least the same success probability. The selected slot times and the response deadlines do not change. The claim follows from \Cref{thm:continuous-time-security}.
\end{proof}

\begin{remark}[Scope and relation to the general impossibility result]
The protocol requires the verifiers to sample and use an exact real-valued slot time. Also, arbitrary entanglement in \Cref{thm:continuous-time-security} means any finite amount, with no bound on that amount fixed in advance. The theorem does not cover an infinite-dimensional state or a strategy that performs infinitely many operations before a deadline. This result does not contradict the general impossibility result of Buhrman et al.~\cite{BCF+14}, which considers finite-dimensional protocols whose classical inputs, including the challenge time, come from fixed finite sets. In our ideal continuous-time protocol, the selected slot time is itself a real-valued input and may be any point of $W_k$. There are therefore uncountably many possible challenge slots, placing this ideal protocol outside the scope of that impossibility result.
\end{remark}

\section{Applications: Position-Based Authentication and Quantum Key Exchange}
\label{sec:position-authentication-appendix}
In this section, we extend our construction to a position-based authentication scheme and then to a position-based key exchange scheme.
\subsection{Definition}
\label{pba:sec:definition}
We adapt Unruh's definition~\cite[Definition 7]{Unr14} to one execution, using completeness and soundness as in~\cite[Section 7]{BCF+14}.

\begin{definition}[Position-based authentication in one execution]
\label{pba:def:adaptive}
A position-based authentication protocol has a prover at zero with input $M\in\mathcal M$ and verifiers at $-R,R$ with input $M'\in\mathcal M$, who jointly accept or reject. For $0\le\varepsilon_s\le1$, we require:
\begin{itemize}
\item \textbf{Completeness.} If $M=M'$, the honest execution accepts with probability one.
\item \textbf{Soundness.} For every adversarial strategy, initial state and preliminary classical procedure permitted by the adversary model, the probability that the verifiers accept although either the prover has started with a different message \(M\neq M'\) or has not started an execution by the verifiers' decision time is at most \(\varepsilon_s\).\footnote{This includes both impersonation and substitution attacks~\cite[Section 7.1]{BCF+14}.}
\end{itemize}
\end{definition}

\paragraph{Adversary Model.}
We use the adversary model of \Cref{sec:adversary-model}, with total quantum-register size $B$, including registers used to hold, delay or resend prover replies. The verifiers coordinate privately. Adversaries may replace, suppress or delay messages and delay honest operations, but cannot reorder or interrupt these operations.

The prover and verifiers each start at most once, in either order, with inputs and starting times depending on earlier interaction. The prover fixes $M$ at its start; the verifiers fix $M'$ before preparing any qubits or challenges.

Before starting, the prover may have any finite classical interaction, with records arbitrarily correlated with the adversaries' quantum state and an untouched reference system. It holds no qubits during this interaction, and no other honest device supplies quantum assistance. Fresh choices have their specified distributions conditional on the earlier interaction. Soundness is measured over the whole execution.

\subsection{Construction}
\label{pba:sec:protocol}
Let $\operatorname{ID}(u,v)$ be the least number of single-symbol insertions and deletions turning $u$ into $v$. Fix a public injective binary encoding $\operatorname{Enc}_0$ of length $m_0$, with distance at least $\delta m_0$ between distinct codewords, for a constant $\delta>0$. For $c=\operatorname{Enc}_0(M)$, insert a zero before each bit and define\footnote{Without these zeros, take $W=1$, with prover codeword $0^{m_0}$ and verifier codeword $1^{m_0}$. Before the verifiers start, the adversaries send the prover four qubits with its announced label, followed by a request to return one qubit. They intercept and discard the returned qubit. Thus the prover is at round $2$ when the verifiers begin round $1$. In each incoming tuple, they swap $D_{00}\leftrightarrow D_{11}$ and $D_{01}\leftrightarrow D_{10}$, and update the labels using the prover's announced random string. Forwarding the verifier requests makes prover round $j+1$ answer verifier round $j$ correctly, leaving only the last verifier round. After the prover finishes, the adversaries announce a replacement label themselves, causing the verifiers to send the final tuple. They measure its $D_{11}$ in the computational basis, share the outcome, and return the corresponding basis state from the requesting side. This passes with probability $3/4$. Two adversaries suffice, with no shared entanglement and $O(m_0)$ counted quantum registers.}
\[
\widetilde c=\operatorname{Enc}(M)=0c_1\,0c_2\cdots0c_{m_0},\qquad m=2m_0,\qquad N=2m.
\]
Fix the message space and lengths before execution. The protocol has $N=2m$ rounds, each with $s$ possible challenge slots; each final code bit governs two rounds. The prover stores $W\ge1$ tuples $(D_{00,r},D_{01,r},D_{10,r},D_{11,r})$, indexed by $r\in\{1,\ldots,W\}$, so $R_H=4W$. Let $k_{\mathrm{in}}\ge1$ be the number of random bits used to label each delivery of new qubits.

\begin{center}
\fbox{%
\begin{minipage}{0.97\linewidth}
\small
\textbf{Position-based authentication.} Use round numbers $j_P,j_V$ for the prover and verifiers. All random choices below are fresh, independent and uniform.
\begin{enumerate}
\setlength{\itemsep}{2pt}\setlength{\parsep}{0pt}
\item\label{pba:step:initialization} \textbf{Initialize.} At their respective starts:
\begin{enumerate}
\setlength{\itemsep}{0pt}\setlength{\parsep}{0pt}
\item Prover: set $\widetilde c\gets\operatorname{Enc}(M)$, $j_P\gets1$. Sample $\nu\in\{0,1\}^{k_{\mathrm{in}}}$; announce label $(\mathsf{init},j_P,\nu)$.
\item Verifiers: set $\widetilde c'\gets\operatorname{Enc}(M')$, $j_V\gets1$.
\end{enumerate}
\item\label{pba:step:input} \textbf{Load qubits.} Choose a verifier uniformly to send one message containing the announced label and $4W$ independent uniform BB84 qubits if $j_V=1$, otherwise $4$. The prover checks the first message's label and size, loads its qubits together into the bank or empty tuple, and waits for a request.
\item \textbf{Challenge.} Sample $i\in[s]$, $b,e\in\{0,1\}$, $r\in[W]$. At $T-R$, $V_b$ sends $(r,e)$, where $T$ is slot $i$'s time at zero.
\item\label{pba:step:reply} \textbf{Reply and update.} On receipt, select a qubit from tuple $r$ and perform the updates below in one operation. Use the table, with $\bar e=1-e$:
\begin{center}
\renewcommand{\arraystretch}{1.15}
\begin{tabular}{@{}ccc@{}}
\hline
$\widetilde c_{\lceil j_P/2\rceil}$ & $j_P$ odd & $j_P$ even\\
\hline
$0$ & $D_{00,r}$ & $D_{0e,r}$\\
$1$ & $D_{1\bar e,r}$ & $D_{11,r}$\\
\hline
\end{tabular}
\end{center}
\begin{enumerate}
\setlength{\itemsep}{0pt}\setlength{\parsep}{0pt}
\item Discard the other $3$ qubits; leave tuple $r$ empty. Set $j_P\gets j_P+1$.
\item If $j_P\le N$, invalidate the old label, sample $\nu\in\{0,1\}^{k_{\mathrm{in}}}$, and announce label $(j_P,r,\nu)$.
\item Immediately release the selected qubit towards the requesting side; if $j_P\le N$, wait for replacement qubits.
\end{enumerate}
\item \textbf{Verify.} At $T+R$, $V_b$ requires exactly $1$ qubit and checks it against the prepared BB84 state selected by the table using $\widetilde c'$, $j_V$, and the sampled $r,e$. Keep the result private. At round end, discard tuple $r$'s $4$ preparation records and set $j_V\gets j_V+1$. Repeat from step~\ref{pba:step:input} while $j_V\le N$. Accept iff all $N$ rounds pass.
\end{enumerate}
\end{minipage}%
}
\end{center}

\paragraph{Protocol Rules.}
The parties follow these steps in order and abort on any detectable violation of the prescribed inputs, labels or timing. Simultaneous arrivals are ordered independently of the random bits in the labels. The prover uses only its $R_H$-qubit bank and discards it and halts on completion or abort.

Choose $\tau\ge8R$ and let round $j$ end at $C_j=C_0+2j\tau$. Its $s$ possible challenge times at zero lie in $[C_{j-1}+\tau,C_{j-1}+5\tau/4]$. Send the new qubits at $C_{j-1}+R$, after the fixed deadline for the input label. Their direct path passes zero at $C_{j-1}+2R$ and reaches the other verifier by $C_{j-1}+3R$; if received there, they are discarded and the verifiers reject. The verifiers prepare no qubits for later rounds before $C_j$. Honest replies travel directly to the requesting verifier in time $R$. The travel intervals of different slots may overlap, but the complete round intervals do not.

\subsection{Security}
\label{pba:sec:security}
Let $1\le L\le m$ and $n=\lfloor m/L\rfloor$.
\begin{theorem}
\label{pba:thm:auth}
For the fixed code, there are constants $s_0,c,c_*>0$ such that, if $s\ge s_0$ and $B\le cnsR_H$, the protocol is perfectly complete and $\varepsilon_s$-sound, with
\[
\varepsilon_s\le e^{-c_*L}+2m2^{-k_{\mathrm{in}}}.
\]
Its quantum communication is $R_H+10m-4$ qubits. In ideal continuous time, the same error holds against every fixed finite coalition, irrespective of its quantum-register size.
\end{theorem}

For the analysis, continue after rejection while retaining the final rejection (convention~\ref{cond:nonoverlap-slot-structure}), and abort if qubits are loaded at the instant their input step begins. Since adversaries are at positive distance from the prover, such a message was sent before the fresh label was chosen. Only the first message is checked, with ties independent of the label, so its guessing probability is at most $2^{-k_{\mathrm{in}}}$. The $N$ input steps, including those before the verifiers start, change any event probability by at most $N2^{-k_{\mathrm{in}}}$. Every replacement is now loaded strictly after the preceding reply, so two replies cannot occur simultaneously.

Call a response independent of $e$ rule 1, and a response selected by $e$ from two cells rule 2. A final code bit zero uses rules $12$, and a final code bit one uses rules $21$. Since the encoding inserts a zero before each original bit, define
\[
H(0)=1212,\qquad H(1)=1221.
\]
For $c=c_1\cdots c_{m_0}$, put $H(c)=H(c_1)\cdots H(c_{m_0})$. Thus $H(c)$ lists the response rules for $\operatorname{Enc}(M)$.

Given two valid messages, pair some prover rounds with some verifier rounds, preserving their order in both sequences. Let $u$ count unpaired verifier rounds and $p_+$ count pairs with prover rule 1 and verifier rule 2. In an execution, set $X_j=1$ if, at the actual challenge time of verifier round $j$, no prover response occurs or the prover uses rule 1 while the verifier uses rule 2; otherwise set $X_j=0$. Count every prover response, even if its output is later discarded, and put $\gamma=\delta/64$.

\begin{samepage}
\begin{lemma}[Comparing the prover's and verifiers' rounds]
\label{pba:lem:count}
For distinct messages, $u+p_+\ge\delta N/32$. Every successful impersonation or substitution satisfies
\[
\sum_{j=1}^{N}X_j\ge\gamma N.
\tag{\thesection.1}\label{pba:eq:count}
\]
\end{lemma}
\end{samepage}
\begin{proof}
Let $c=\operatorname{Enc}_0(M)$ and $c'=\operatorname{Enc}_0(M')$. We first show that $\operatorname{ID}(H(c),H(c'))\ge\tfrac12\operatorname{ID}(c,c')$. Checking the four possible pairs of adjacent blocks shows that, in a concatenation of $1212$ and $1221$, an occurrence of $1212$ starts at offset $0$ or $2$ inside a block representing zero, while an occurrence of $1221$ starts only at offset $0$ inside a block representing one. Now take a common subsequence of $H(c)$ and $H(c')$ obtained by deleting $D$ symbols in total. Each deletion spoils at most one block of $H(c)$, so at least $m_0-D$ blocks of $H(c)$ match four consecutive symbols of $H(c')$. Assigning each such block to the block of $H(c')$ containing the start of its match gives distinct, increasing indices that encode the same bits. Hence $c$ and $c'$ have a common subsequence of length at least $m_0-D$, i.e., $\operatorname{ID}(c,c')\le2D$.
 
Both rule sequences have length $N$ and symbol sum $3m$, so each has $u$ unpaired positions. Let $p_-$ count the pairs in which the prover uses Rule~\ref{pba:rule:random} while the verifier uses Rule~\ref{pba:rule:fixed}. Every unequal pair differs by one, so equality of the total sums gives $p_-\le p_++u$. Deleting all unpaired positions and unequal pairs therefore gives
\[
\operatorname{ID}(H(c),H(c'))\le2(u+p_++p_-)\le4(u+p_+).
\]
On the other hand, the first part and the distance of $\operatorname{Enc}_0$ give $\operatorname{ID}(H(c),H(c'))\ge\delta m_0/2=\delta N/8$, proving $u+p_+\ge\delta N/32$.
 
For \eqref{pba:eq:count}, pair each verifier challenge with the prover round whose response occurs at that time, if any. Since the round numbers increase, this is a pairing, and every unpaired verifier round and every pair counted by $p_+$ has $X_j=1$. If the prover has not fixed a message, no round has a response, so $X_j=1$ for all $j$.
\end{proof}

\paragraph{Entanglement consumption.}
The next lemma applies \Cref{theorem:k-extraction-entanglement-consumption} through \Cref{cor:k-extraction-entanglement-consumption-distributed}, accounting for the prover's help. Let $O=O_1\cdots O_k$ be the reference qubits and $U,V$ the two sides. A recovery channel outputs one qubit; its BB84 error is the probability that this qubit and $O_j$ give different outcomes in the same uniform $X$ or $Z$ basis, averaged over uniform $j$ and the classical record. All entropies and entanglement below are averaged over that record.

\par\smallskip\noindent\begin{minipage}{\linewidth}
\begin{lemma}[Entanglement consumption]
\label{pba:lem:auth-consumption}
Consider a finite evolution of $OUV$ that leaves $O$ untouched and consists of LOCC across $U:V$ and $T$ qubit transfers, counting repetitions. Assume that the initial and final states are classically flagged mixtures of pure states, with the classical record providing the flags (\Cref{def:classically-flagged-mixture-pure-states}). For each $j$, suppose there are separate initial and final recovery channels acting on $U$ and $UV$, respectively. Suppose that their average failure probabilities in the BB84 verification test of \Cref{lem:bb84_fidelity} are at most $\varepsilon_0$ and $\varepsilon_1$, respectively, averaging over uniform $j$ and the classical record.

Write
\[
\Delta E=\EE\Esq(U:V)_{\mathrm{in}}-\EE\Esq(U:V)_{\mathrm{out}}.
\]
Then there is a universal nonnegative function $\varphi$, continuous, nondecreasing and concave in each argument, with $\varphi(0,0)=0$, for which the following bounds hold.
\begin{enumerate}
\setlength{\itemsep}{3pt}\setlength{\parsep}{0pt}
\item\label{pba:case:helper} \textbf{Recovery using at most $h$ qubits from $U$.} If final recovery uses only $VS$, where $S\subseteq U$ contains at most $h$ qubits, then
\[
\Delta E\ge k-h-k\varphi(\varepsilon_0,\varepsilon_1)-\EE T.
\]
\item\label{pba:case:noresp} \textbf{Responses without a new prover reply.} For each $j$ and $\theta\in\{X,Z\}$, suppose positive operators $E^{j,\theta}_0,E^{j,\theta}_1$ on $V$ satisfy $E^{j,\theta}_0+E^{j,\theta}_1\le I$. Within the same final response procedure, $E^{j,\theta}_b$ describes the joint event that no new prover reply is used and measuring the returned qubit in basis $\theta$ gives bit $b$. Let $\bar u$ be the probability that no new reply is used, averaged over $j$, basis and classical record. Then
\[
\Delta E\ge\frac{k(\bar u-2\varepsilon_1)_+^2}{4\ln2}-k\varphi(\varepsilon_0,\varepsilon_1)-\EE T,
\qquad x_+=\max\{x,0\}.
\]
\end{enumerate}
\end{lemma}

\end{minipage}\par
\begin{proof}
\clubpenalty=10000
By \Cref{lem:bb84_fidelity}, the average Bell infidelities of the recovered pairs are at most $2\varepsilon_0$ and $2\varepsilon_1$. We choose $\varphi$ so that $\varphi(\varepsilon_0,\varepsilon_1)/2$ bounds the per-qubit error terms of \Cref{cor:k-extraction-entanglement-consumption-distributed,lemma:lower_bound-closeness} for these infidelities, and $\varphi\ge2$ outside the range allowed by the corollary; concavity handles the averaging. We take $\varphi(\varepsilon_0,\varepsilon_1)=\varphi_0(\varepsilon_0)+\varphi_1(\varepsilon_1)$ with nonnegative, continuous, nondecreasing concave functions $\varphi_0,\varphi_1$ vanishing at zero; hence $\varphi$ is jointly concave.
Write $E_{\mathrm{in}},E_{\mathrm{out}}$ for the two expectations in $\Delta E$. Both parts follow from
\[
\Delta E\ge\tfrac12\EE[k-S(O\mid V)]_{\mathrm{out}}-\tfrac k2\varphi(\varepsilon_0,\varepsilon_1)-\EE T.
\tag{\thesection.2}\label{pba:eq:flow}
\]
 
\emph{Proof of \eqref{pba:eq:flow}.} If $\varphi\ge2$, the bound follows from $S(O\mid V)\ge-k$ and $\Delta E\ge-\EE T$. Otherwise, the final recovery need not be local to $V$, so \Cref{cor:k-extraction-entanglement-consumption-distributed} does not apply directly. Unlike in the QPV proof, we therefore first move the registers of $U$ needed for recovery to $V$. We run $q$ independent copies of the evolution, replacing each qubit transfer by teleportation, and then merge the final $U$ registers into $V$ by state merging~\cite[Theorem~2 and Remark~7]{HOW07}, which preserves their purification. After this auxiliary LOCC procedure, all final recovery maps are local on $V$, and \Cref{cor:k-extraction-entanglement-consumption-distributed} applies. The merging cost is $q\EE S(U\mid V)_{\mathrm{out}}+o(q)$, also with classical records: we partition the conditional states into finitely many sets of trace-distance diameter at most $\zeta$, merge the average state of each set, use uniform continuity of entropy, teleport the sets containing fewer than $\sqrt q$ copies, and let $q\to\infty$ and then $\zeta\to0$. Allocating a fixed Bell-pair supply and retaining the unused pairs, subadditivity, monogamy, \Cref{lemma:lower_bound-closeness} and \Cref{lemma:direct-sum} give
\[
E_{\mathrm{in}}+\EE T+\EE S(U\mid V)_{\mathrm{out}}\ge k\bigl(1-\tfrac12\varphi(\varepsilon_0,\varepsilon_1)\bigr).
\]
Since each final conditional state is pure, $S(U\mid V)=S(O)-S(V)$ and $\Esq(U:V)\le I(U:V)/2$. 

Therefore,
\[
\EE S(U\mid V)_{\mathrm{out}}+E_{\mathrm{out}}
\le\tfrac12\EE[S(O)+S(O\mid V)]_{\mathrm{out}}
\le\tfrac12\EE[k+S(O\mid V)]_{\mathrm{out}}.
\]

Subtracting $E_{\mathrm{out}}$ and using $S(O)\le k$ proves \eqref{pba:eq:flow}.
 
\emph{Part~\ref{pba:case:helper}.} Monogamy and \Cref{lemma:lower_bound-closeness}, applied to each recovered qubit, give $\EE I(O:VS)_{\mathrm{out}}\ge2k\bigl(1-\tfrac12\varphi(\varepsilon_0,\varepsilon_1)\bigr)$. Since $I(O:S\mid V)\le2h$ and $S(O)\le k$, we get $\EE[k-S(O\mid V)]_{\mathrm{out}}\ge2k\bigl(1-\tfrac12\varphi(\varepsilon_0,\varepsilon_1)\bigr)-2h$, and \eqref{pba:eq:flow} proves the claim.
 
\emph{Part~\ref{pba:case:noresp}.} Fix a final record and an index $j$, let $u_j$ and $v_j$ be the basis-averaged probabilities of the event and of the event with an incorrect bit, and put $A_{j,\theta}=E^{j,\theta}_0-E^{j,\theta}_1$. With the Pauli observables $\sigma_X,\sigma_Z$,
\[
\frac12\sum_{\theta\in\{X,Z\}}\operatorname{Tr}\!\left[\rho_{O_jV}(\sigma_\theta\otimes A_{j,\theta})\right]=u_j-2v_j,
\]
which vanishes on $I_{O_j}/2\otimes\rho_V$. Quantum Pinsker's inequality therefore gives $1-S(O_j\mid V)\ge(u_j-2v_j)_+^2/(2\ln2)$. By the chain rule, $k-S(O\mid V)\ge\sum_j[1-S(O_j\mid V)]$. Averaging, using convexity of $x_+^2$ and that the mean of $v_j$ is at most $\varepsilon_1$, we get $\EE[k-S(O\mid V)]_{\mathrm{out}}\ge k(\bar u-2\varepsilon_1)_+^2/(2\ln2)$, and \eqref{pba:eq:flow} proves the claim.
\end{proof}

For a sufficiently large universal constant $A$, put $F_4(\varepsilon)=A[\varepsilon+\varphi(A\varepsilon,A\varepsilon)]$ and $R_4(\varepsilon)=A\varphi(0,\varepsilon)$. These concave, continuous functions bound the recovery errors and vanish at zero.

\begin{proof}[Proof of \Cref{pba:thm:auth}]
Let $n:=\lfloor m/L\rfloor$. Fix a finite coalition with total quantum-register size $B\le cnsR_H$, and let $F$ be the event of a successful impersonation or substitution. We show that $\Pr(F)<e^{-c_*L}+2m2^{-k_{\mathrm{in}}}$. As in the proof of \Cref{thm:amplified_security-generalized}, we may take the initial state to be pure, and we use the Bell-pair description in which every BB84 input is half of a Bell pair whose other half is kept by the verifiers as a reference qubit. Under Rule~\ref{pba:rule:fixed}, this gives $W$ relevant reference qubits, and under Rule~\ref{pba:rule:random}, it gives $2W$.
 
\paragraph{Setting up the Framework}
\begin{itemize}
\item \textbf{Conventions.} We use the analysis conventions of \Cref{sec:security-analysis-conventions}, including the classical records of the prover's operations and replies, and we continue the execution after rejection (convention~\ref{cond:nonoverlap-slot-structure}). We also abort if qubits are loaded at the very instant their input step begins. Such a message was sent before the fresh label was chosen, and since only the first message is checked, with ties ordered independently of the label, it carries the correct label with probability at most $2^{-k_{\mathrm{in}}}$. Over the $N$ input steps, this changes any probability by at most $N2^{-k_{\mathrm{in}}}=2m2^{-k_{\mathrm{in}}}$. Afterwards, every replacement is loaded strictly after the preceding reply, so no two replies occur simultaneously.
\item \textbf{Registers.} We group the nonreference registers outside the prover by side into $A_L$ and $A_R$, assigning each reply to its destination at emission. Let $G$ be the prover's bank and $D_{uv}^{W}:=D_{uv,1}\cdots D_{uv,W}$. We track the potential
\[
\begin{aligned}
\Phi&:=P_L+P_R+Q_{00}+Q_{11},\qquad P_L:=\Esq(A_LG:A_R),\qquad P_R:=\Esq(A_L:A_RG),\\
Q_{00}&:=\Esq(A_LD_{00}^{W}:A_RD_{01}^{W}D_{10}^{W}D_{11}^{W}),\qquad
Q_{11}:=\Esq(A_LD_{11}^{W}:A_RD_{00}^{W}D_{01}^{W}D_{10}^{W}).
\end{aligned}
\]
All four quantities are averaged over the classical records, with the remaining reference qubits traced out. Prover operations are local for each of the four cuts, except for qubit transfers, and each transferred qubit increases each term by at most one.
\item \textbf{Blocks.} As in Step~5 of the proof of \Cref{thm:amplified_security-generalized}, we partition the $N$ rounds into $L$ consecutive blocks of sizes $d_b\in\{n',n'+1\}$, where $n'=\lfloor N/L\rfloor\ge2n$, and put $d_{\max}=\lceil N/L\rceil$, so that $N/d_{\max}\ge L/2$.
\end{itemize}
 
\textbf{Step 1: Mismatched rounds.} In the QPV proof, a highly successful round forces entanglement consumption because there is no honest prover. Here, a round may succeed with the honest prover's help, so high success alone is not enough. By \Cref{pba:lem:count}, every successful forgery satisfies \(\sum_jX_j\ge\gamma N\). In Step~5, we use this count to select a highly successful block with a constant (expected) fraction of \emph{mismatched rounds}: conditioned on the history before the block and averaging over all histories within it, each such round has probability at least a sufficiently small fixed positive constant of satisfying \(X_j=1\). In this proof, mismatched rounds in the selected highly successful block play the same role as high-success rounds in the proof of \Cref{thm:amplified_security-generalized}. We fix an integer $K\ge1$, chosen in Step~5, and, for a block $b$ with $d=d_b$ and a starting history $h$ with at most $Kd$ earlier direct replies still in transit, we stop the prover for the analysis after $Kd$ further responses, discard its bank, and count the block as failed. Let $X_j^{\mathrm{stop}}$ be the indicator of \Cref{pba:lem:count} in this stopped continuation, and let $\bar p$ be the probability, given $h$, that some round of the block fails or the limit is reached. All expectations in the following block calculation condition only on $h$. Initially, $\Phi\le C(B+W+Kd)$ by \Cref{lem:quantum_dim_bound}, counting the bank and the at most $Kd$ earlier replies.
 
\textbf{Step 2: Separating the two sides in one candidate slot.} Fix a round with candidate times $T_1<\cdots<T_s$, actual challenge index $I$, and endpoints $C_-<C_+$. Instead of the operation maps of Claim~\ref{claim:form-of-a-selected-slot}, we order events by $u(t,x):=t+|x|$, preserving causal dependencies and each party's operation order on ties; if an operation at $(t',x')$ can depend on one at $(t,x)$, then
\[
u(t',x')-u(t,x)\ge|x'-x|+|x'|-|x|\ge0.
\]
For candidate $i$, we run the events in $\{t\le C_-\}\cup\{u<T_i\}$, retaining messages in transit with their physical delivery times unchanged. These sets are closed under dependencies, and omitting the instruction that reaches $0$ at $T_I$ leaves every earlier state unchanged. We fix a requesting side and, among the operations with $u=T_i$, first run those on the opposite side; this is possible because a message crossing $0$ strictly increases $u$, as in the proof of Claim~\ref{claim:form-of-a-selected-slot}. If the prover is waiting for input, we temporarily load the first input arriving from the opposite side at $T_i$, or dummy qubits, keeping its round counter and whether it is waiting for qubits unchanged; this input is independent of $(r,e)$. The requesting side may replace it during recovery if its input arrives first; an invalid first input causes abort. The remaining operations at $0$ and on the requesting side then produce the response from that side's registers and the bank, while a response sent to the opposite side cannot return in time. These recovery operations are not performed. Perform the remaining idle operations, check the first input to resolve the temporary assignment, and advance to the next state. Alternating the requesting side gives two sequences of states with the actual endpoints of the round, and we keep the comparisons $i\to i+1$ with $i+1<I$. Resolve each temporary input once before the next round.
 
Two facts about the prover make the lemma applicable. First, under Rule~\ref{pba:rule:fixed}, only $D_{00}^W$ or $D_{11}^W$ is needed to reply, and altered requests cannot enlarge this set. Second, a response opens an input step that cannot load immediately, so the prover answers at most one round at a time, and without a response it neither reads the bank nor keeps outcomes depending on it. Hence the requesting side can measure the event that no new reply is used together with the returned bit, as required by part~\ref{pba:case:noresp} of \Cref{pba:lem:auth-consumption}. Missing or malformed responses are replaced by a fixed qubit. As in Step~3 of the proof of \Cref{thm:amplified_security-generalized}, direct opposite-side responses would also arrive if the challenge were later than both compared slots, forcing rejection. Including these events gives total weighted recovery error $O(s\varepsilon)$ in Step~4, where $\varepsilon$ is the round's rejection probability.
 
\textbf{Step 3: Comparing consecutive states.} For $\sigma\in\{00,11\}$, let $a_\sigma$ be the probability that, at the actual challenge, the prover gives the fixed reply $D_\sigma$ while the verifier uses Rule~\ref{pba:rule:random}, and let $Z_i^\sigma$ indicate that this reply is prescribed before $T_i$. When the verifier uses Rule~\ref{pba:rule:random}, we apply part~\ref{pba:case:helper} of \Cref{pba:lem:auth-consumption} to $Q_\sigma$ in the left-to-right comparisons, conditioned on $Z_i^\sigma=1$, with initial recovery from $A_LD_\sigma^W$ and final recovery from $A_RD_\sigma^W$. This gives $k=2W$ and $h=W$, and hence a decrease of $W\Pr[Z_i^\sigma=1]$ up to recovery errors; a change of $Z_i^\sigma$ from one to zero counts as an error. We apply part~\ref{pba:case:noresp} to $P_L$ from left to right and to $P_R$ from right to left, with $W\le k\le2W$. All other comparisons use monotonicity. Additional reference measurements and the verifier's measurements cost $O(W)$ per round by \Cref{lem:verifier-entanglement-accounting}, since at most $4W$ reference qubits remain. Measure unused references in their eventual BB84 bases and forget the additional outcomes at the round endpoint.
 
\textbf{Step 4: Adding the entanglement decrease over one block.} As in Step~4 of the proof of \Cref{thm:amplified_security-generalized}, a comparison $i\to i+1$ contributes only when the challenge is later, so we weight it by $w_i:=(s-i-1)/s$ for $i<s$, with $w_s:=0$. For $x_i\in[0,1]$ with mean $\bar x$ and $p\in\{1,2\}$,
\[
\sum_iw_ix_i^p\ge4^{-(p+1)}s\bar x^{p+1}-1,
\]
which is trivial if $s\bar x<4$ and otherwise follows by dropping the last $\lceil s\bar x/2\rceil$ terms and using convexity. With $p=1$, the comparisons of $Q_{00},Q_{11}$ give a constant times $sW(a_{00}^2+a_{11}^2)$. With $p=2$, those of $P_L,P_R$ give a constant times $sW(v-2\varepsilon)_+^3$, where $v$ is the no-response probability and $\varepsilon$ the rejection probability. Using $(v-2\varepsilon)_+^3\ge v^3-6\varepsilon$, the recovery errors total at most $sWF_4(\varepsilon)$. Since the counted events are disjoint,
\[
a_{00}^2+a_{11}^2+v^3\ge(a_{00}+a_{11}+v)^3/9=(\EE X_j^{\mathrm{stop}})^3/9.
\]
 
It remains to bound the losses over the block. For each variable $Z_i^\sigma$, let $\delta_{j,i}$ be its probability of changing from one to zero in comparison $i\to i+1$ when the instruction is omitted, conditional on the actual history before the round. Since the later challenge is independent, $\EE[w_i\delta_{j,i}]$ is the probability that $I_j>i+1$ and the indicator changes. Retained comparisons are disjoint within each sequence, and temporary assignments change neither the round counter nor whether the prover is waiting for qubits. Thus changes require a response, a start or a stop, so $\sum_{j,i,\sigma}\EE[w_i\delta_{j,i}]\le C(K+1)d$, and concavity gives
\[
\sum_{j,i,\sigma}\EE[w_iR_4(\delta_{j,i})]\le Cds\,R_4\!\left(\frac{C(K+1)}{s}\right).
\tag{\thesection.3}\label{pba:eq:flags}
\]
Resolve temporary inputs once and keep traveling qubits assigned to their side. The verifiers prepare at most $4W+4d$ qubits, and the prover loads at most $4W+4(Kd+1)$ and emits at most $Kd$, so transfers and replacements increase $\Phi$ by at most $C(W+(K+1)d)$ over the block. Each round rejects with probability at most $\bar p$, so by concavity the recovery errors total at most $dsWF_4(\bar p)$. Summing over the block, using $\Phi\ge0$ at its end, convexity of the cubic terms and \eqref{pba:eq:flags}, we obtain
\[
\frac1d\sum_{j\in b}\EE X_j^{\mathrm{stop}}
\le\left[\frac{C(B+W+Kd)}{c_0dsW}+\frac{F_4(\bar p)+\rho_{K,s}}{c_0}\right]^{1/3},
\qquad
\rho_{K,s}\le C\left[R_4\!\left(\frac{C(K+1)}s\right)+\frac{K+1}s\right],
\tag{\thesection.4}\label{pba:eq:block}
\]
for absolute constants $c_0,C>0$. For each fixed $K$, $\rho_{K,s}\to0$ as $s\to\infty$, uniformly in $m,L,W$.
 
\textbf{Step 5: Selecting a block.} As in the QPV proof, we now select a block after conditioning on the earlier history. Unlike there, the selected block must have both a small failure probability and a constant expected fraction of rounds with $X_j=1$. Choose $0<\delta_0\le1$ with $F_4(\delta_0)/c_0<(\gamma/8)^3/4$, and put
\[
t=\ln(1+\delta_0),\qquad
\beta=t+\ln\frac{32}{\gamma\delta_0},\qquad
K\ge\frac{128\beta}{\gamma\delta_0}.
\]
Suppose, for contradiction, that $F$, excluding the label guesses above, has probability at least $e^{-c_*L}$ with $c_*:=\gamma\delta_0/32$. For a block $b$, let $o_b$ count the earlier direct replies in transit at its start, $q_b$ its prover responses and $P_b$ the event that it passes, and put
\[
E_b:=\mathbf1_{\{o_b>Kd_b\ \text{or}\ q_b\ge Kd_b\}},\qquad
\bar P_b:=P_b(1-E_b),\qquad
Z_b:=\frac1{d_{\max}}\sum_{j\in b}X_j\in[0,1].
\]
For a history $h$ with $o_b\le Kd_b$, let $P_b^{\mathrm{stop}}$ and $Z_b^{\mathrm{stop}}$ be the analogues in the stopped continuation, and put $\bar p_b(h):=1-\EE[P_b^{\mathrm{stop}}\mid h]$ and $\mu_b(h):=\EE[Z_b^{\mathrm{stop}}\mid h]$. The two executions agree until the response limit, so
\[
\EE[\bar P_be^{tZ_b}\mid h]=\EE[P_b^{\mathrm{stop}}e^{tZ_b^{\mathrm{stop}}}\mid h],
\tag{\thesection.5}\label{pba:eq:coupling}
\]
and the left side vanishes if $o_b>Kd_b$. Each response lies in one block, and block starts are more than $R$ apart, so each reply is in transit at the start of at most one block. Hence $\sum_bq_b\le N$, $\sum_bo_b\le N$ and $\sum_bE_b\le2N/(Kn')\le4L/K$; adversarially delayed qubits count towards $B$ instead.
 
Let $Y_b:=P_be^{tZ_b-\beta E_b}$ and $Q:=1+\gamma\delta_0/8+e^{t-\beta}=1+5\gamma\delta_0/32$. On $F$, every $P_b=1$ and, by \eqref{pba:eq:count}, $\sum_bZ_b\ge\gamma L/2$. Since $t\ge\delta_0/2$ and $4\beta/K\le\gamma\delta_0/32$,
\[
\EE\prod_bY_b\ge\Pr(F)\,e^{t\gamma L/2-4\beta L/K}\ge e^{6\gamma\delta_0L/32}>Q^L.
\]
By iterated conditioning, some block $b$ and some actual history $h$, in which all earlier blocks passed, satisfy $\EE[Y_b\mid h]>Q$. Since $Y_b\le\bar P_be^{tZ_b}+e^{t-\beta}$, this gives $\EE[\bar P_be^{tZ_b}\mid h]>1+\gamma\delta_0/8$, so in particular $o_b\le Kd_b$. By \eqref{pba:eq:coupling} and $e^{tz}\le1+\delta_0z$ for $z\in[0,1]$,
\[
1+\gamma\delta_0/8<1-\bar p_b(h)+\delta_0\mu_b(h),
\qquad\text{so}\qquad
\mu_b(h)>\gamma/8\quad\text{and}\quad\bar p_b(h)<\delta_0.
\tag{\thesection.6}\label{pba:eq:selected-block}
\]
That is, the selected block fails with small probability and has a constant expected fraction of rounds with $X_j=1$.
 
\textbf{Concluding the storage bound.} We apply \eqref{pba:eq:block} to this block and history. Choose $c$ sufficiently small and then $s_0$ sufficiently large, independently of $m,L,W$. Since $d\ge2n$ and $B\le cnsR_H=4cnsW$, the first term of \eqref{pba:eq:block} is $O(c+(K+1)/s)$. Together with $\bar p_b(h)<\delta_0$ and the choice of $\delta_0$, \eqref{pba:eq:block} gives $\mu_b(h)\le\frac1d\sum_{j\in b}\EE[X_j^{\mathrm{stop}}\mid h]\le\gamma/8$, contradicting \eqref{pba:eq:selected-block}. Hence $\Pr(F)<e^{-c_*L}+2m2^{-k_{\mathrm{in}}}$, where the second term accounts for the label guesses. All bounds condition only on actual earlier histories, so averaging over the permitted starting histories, with zero acceptance when the verifiers never start, proves soundness for adaptive messages and starting times. For $L=\lceil\log_2^2\secparam\rceil$ and $k_{\mathrm{in}}\ge\lceil\log_2(2m)\rceil+L$, this error is negligible.
 
\emph{Intermediate start.}\phantomsection\label{pba:par:intermediate} The argument also applies from just before the first fresh authentication qubits are prepared, provided that the message, the completed responses, the next round, the waiting status and the bank follow the prover's program, while its quantum state may be correlated with the counted registers; that each uncounted reply in transit is one qubit from a distinct completed response with remaining direct travel time at most $R$; that these replies and the future responses number at most $N$; and that no other uncounted verifier qubits are in transit.
 
\emph{Completeness and communication.} In an honest execution, every prescribed qubit is returned at its deadline and passes the BB84 test. Initialization, replacements and responses communicate $4W+4(N-1)+2m=R_H+10m-4$ qubits, and the prover stores $4W=R_H$ qubits.
 
\emph{Continuous time.} As in the proof of \Cref{thm:continuous-time-security}, we write each uniform challenge time in $[a_j,a_j+\Delta)$ as $a_j+\Xi_j+J_j\Delta/q$, with $\Xi_j$ uniform on $[0,\Delta/q)$ and $J_j$ uniform on $\{0,\ldots,q-1\}$. For a finite coalition of quantum-register size $B$, we choose $q\ge s_0$ with $B\le cnqR_H$, apply the discrete-time bound conditioned on the shifts $\Xi_j$, and average.
\end{proof}

For security parameter $\lambda$ and polynomially bounded $m$, taking $L=\lceil(\log_2\lambda)^2\rceil\le m$ and $k_{\mathrm{in}}\ge\lceil\log_2(2m)\rceil+L$ gives negligible soundness error. The allowed quantum-register size $B\le c\lfloor m/L\rfloor sR_H$ has ratio $\Omega(\lfloor m/L\rfloor sR_H/(R_H+m))$ to honest quantum storage and communication. Increasing $s\ge s_0$ requires neither additional communication nor longer verifier windows.

\subsection{Position-based Quantum Key Exchange}
\label{pba:sec:application}
We follow Buhrman et al.~\cite[Section~7.3]{BCF+14} and Unruh~\cite[Section~4]{Unr14}, using one final authentication as in~\cite[Section~4]{TNT26}. For secrecy, we adapt Definition~II.1 of~\cite{TNT26} to the verifiers' output.

\begin{samepage}
\begin{definition}[Position-based quantum key exchange]
\label{pba:def:qkd-secrecy}
A position-based key exchange protocol outputs keys \(K_P,K_V\in\{0,1\}^{k_{\mathrm{key}}}\cup\{\bot\}\) for the prover and verifiers, respectively, where \(\bot\) denotes abort or noncompletion. We require:
\begin{itemize}
\item \textbf{Completeness.} In an honest execution (without any adversaries), \(\Pr[K_P=K_V\ne\bot]\ge1-\varepsilon_{\mathrm{comp}}\).
\item \textbf{Correctness.} For every permitted attack, the probability that both parties output keys but these keys differ is at most \(\varepsilon_{\mathrm{cor}}\).
\item \textbf{Secrecy.} For every permitted attack, let \(p_V=\Pr[K_V\ne\bot]\), and let \(\rho^{\mathrm{acc}}_{K_VE}\) be the normalized final state conditioned on the verifiers outputting a key. We require
\[
\frac{p_V}{2}\left\|\rho^{\mathrm{acc}}_{K_VE}-\tau_{K_V}\otimes\rho^{\mathrm{acc}}_E\right\|_1\le\varepsilon_{\mathrm{sec}},
\]
where \(E\) includes the adversary's information and public record, and \(\tau_{K_V}\) is the uniform state on \(k_{\mathrm{key}}\)-bit keys. The condition holds automatically if \(p_V=0\).
\end{itemize}
\end{definition}
\end{samepage}

Use the ideal-qubit BB84 protocol of~\cite[Sections~7--9]{TL17}, with $Q$ inputs, $n_{\mathrm{raw}}$ key-generation bits, $k_{\mathrm{test}}$ test bits, threshold $\delta_{\mathrm{test}}$, verification-hash length $t_{\mathrm{ver}}$, and key length $k_{\mathrm{key}}$ chosen according to its finite-key bound. Fix universal hash families and message lengths at setup. Write $h_{\mathrm{ver}},h_{\mathrm{key}}$ for the verification and privacy-amplification hashes selected below.

\begingroup
\small
\begin{enumerate}
\setlength{\itemsep}{2pt}\setlength{\parsep}{0pt}
\item \textbf{Distribute and measure.}
\begin{enumerate}
\setlength{\itemsep}{0pt}\setlength{\parsep}{0pt}
\item \textbf{Verifiers:} Send $Q$ independent uniform BB84 states on the fixed schedule, independently of their preparation data.
\item \textbf{Prover:} For each label in $[Q]$, immediately measure the first valid qubit separately in the independent private basis for that index. Discard other inputs; keep only classical records and send no qubits. Close preliminary input when all labels are filled.
\end{enumerate}
\item \textbf{Exchange $5$ classical messages.}
\begin{enumerate}
\setlength{\itemsep}{0pt}\setlength{\parsep}{0pt}
\item \textbf{Prover:} Send the basis vector.
\item \textbf{Verifiers:} Send the first $n_{\mathrm{raw}}+k_{\mathrm{test}}$ matching-basis indices in lexicographic order; abort if too few. Include a uniform $k_{\mathrm{test}}$-element test subset, independent seeds for $h_{\mathrm{ver}},h_{\mathrm{key}}$, and test bits. The remaining bits give $X$ (verifiers) and $Y$ (prover).
\item \textbf{Prover:} Send $\mathsf{pass}$ if the test-error fraction is at most $\delta_{\mathrm{test}}$; otherwise abort.
\item \textbf{Verifiers:} Send the constant $1$-bit syndrome and $h_{\mathrm{ver}}(X)$.
\item\label{pba:step:qkd-finish} \textbf{Prover:} Abort unless the received hash equals $h_{\mathrm{ver}}(Y)$. Send $\mathsf{pass}$; set $K_P\gets h_{\mathrm{key}}(Y)$, keeping $Y$ unchanged. \textbf{Verifiers:} On receiving $\mathsf{pass}$, set $K_V\gets h_{\mathrm{key}}(X)$.
\end{enumerate}
\item\label{pba:condition:preliminary-wait} \textbf{Authenticate the public record.} After the \hyperref[pba:step:qkd-finish]{last message} succeeds:
\begin{enumerate}
\setlength{\itemsep}{0pt}\setlength{\parsep}{0pt}
\item \textbf{Prover:} Clear the quantum state. Start authentication once and announce $M$: the record of sent and received messages, a success flag, and $k_f$ fresh independent uniform bits. Pad injectively to the fixed length; exclude private keys and undisclosed outcomes.
\item \textbf{Verifiers:} Fix one candidate $M'$; check its success flag and record.
\item \textbf{Verifiers:} Stop preliminary transmissions and discard unsent qubits. Wait until every sent qubit completes its prescribed $2R$-length path; discard arrivals at the opposite verifier. Then prepare authentication qubits and authenticate once.
\item \textbf{Outputs:} Verifiers: release $K_V$ only on acceptance before the fixed deadline beyond the honest schedule; otherwise output $\bot$. Prover: release $K_P$ on completing authentication.
\end{enumerate}
\end{enumerate}

\begin{itemize}
\setlength{\itemsep}{2pt}\setlength{\parsep}{0pt}
\item \textbf{Input rules.} Fix each input's label and index before its last transmission. Select, measure or discard inputs using only their classical format, indices, step numbers and independent randomness.
\item \textbf{Authentication.} Use distinct input labels and follow \Cref{pba:sec:protocol}, reading only public records and fresh independent state. All phases use the same paths and counted registers; add $Q$ BB84 transmissions.
\item \textbf{Classical messages.} Attach session, role and step labels and $k_c$ fresh independent uniform bits, revealed only on sending. Check one fixed candidate per expected step; abort if invalid and ignore unexpected steps.
\item \textbf{Aborts.} A local abort or valid abort notice stops normal progress. Each party sends at most $1$ abort notice with fresh bits, hence at most $7$ sends and $7$ fixed receives. A failed prover authenticates failure or halts, outputting $\bot$.
\end{itemize}
\endgroup

Let $\varepsilon_{\mathrm{QKD}}$ be the quantum key-distribution (QKD) error of this BB84 protocol, namely the sum of its correctness and secrecy errors~\cite[Section~4, Lemma~1]{TL17}.

\begin{corollary}[Position-based quantum key exchange]
\label{pba:cor:qkd}
In both the discrete-time and ideal continuous-time models, under the corresponding hypotheses of \Cref{pba:thm:auth}, including those on the LOCC coalition of adversaries, the protocol satisfies \Cref{pba:def:qkd-secrecy} with completeness error $\varepsilon_{\mathrm{comp}}=\Pr[\operatorname{Bin}(Q,1/2)<n_{\mathrm{raw}}+k_{\mathrm{test}}]$, where $\operatorname{Bin}(Q,1/2)$ is a binomial random variable, and
\[
\varepsilon_{\mathrm{cor}},\varepsilon_{\mathrm{sec}}
\le\varepsilon_{\mathrm{QKD}}+e^{-c_*L}+2m2^{-k_{\mathrm{in}}}+14\,2^{-k_c}+2^{-k_f}.
\]
\end{corollary}
\begin{proof}
The error bound has three sources: the BB84 protocol itself, contributing $\varepsilon_{\mathrm{QKD}}$; a successful attack on the final authentication, contributing $e^{-c_*L}+2m2^{-k_{\mathrm{in}}}$; and the adversaries guessing the fresh random bits in the classical messages, contributing $14\cdot2^{-k_c}+2^{-k_f}$. The only new difficulty is that authentication now runs after the BB84 phase, in which the prover received qubits that the adversaries may have prepared. We therefore first show that \Cref{pba:thm:auth} still applies, then bound the BB84 phase, and finally combine the two using~\cite{TNT26}.
 
\emph{Step 1: Authentication remains secure after the BB84 phase.} If the prover has already started authentication before the first authentication qubits are prepared, it satisfies the conditions of the \hyperref[pba:par:intermediate]{intermediate start} in the proof of \Cref{pba:thm:auth}. Otherwise, the prover's only quantum operations before authentication are single-qubit measurements on arrival, and we replace them by classical interaction. For the analysis, we reveal the prover's bases and move each measurement of an input coming from a counted register to its sender, before transmission, keeping that register occupied for the same period and delivering the outcome at the original arrival time; authentication qubits relabeled as inputs are measured after the last adversary. This does not change the execution, since measurements commute with operations on other registers and with classical selection, and measuring an unselected input and forgetting the outcome is the same as discarding it. Step~\ref{pba:condition:preliminary-wait} of the protocol ensures that no uncounted earlier input remains. The result is a classical preliminary procedure with the same message, timing, and storage, so \Cref{pba:thm:auth} applies.
 
\emph{Step 2: The BB84 phase.} Suppose the classical messages are unchanged and each is received after it is sent. Since inputs are selected independently of the private bases and outcomes, deferring the prover's measurements gives an arbitrary channel followed by BB84 measurements, so the analysis of~\cite{TL17} applies, with the arrival records included in the adversaries' information and incomplete runs counted as failures. Aborts and the deadline only erase keys as functions of public records, which commutes with replacing the key by an ideal one, so the error remains $\varepsilon_{\mathrm{QKD}}$.
 
\emph{Step 3: Authenticating the public record.} If the records match and each message is received after it is sent, the key data are fixed before authentication starts, and the proof of~\cite[Theorem~IV.1 and Appendix~A]{TNT26}, with the prover sending the transcript, preserves the bounds of Step~2. If the verifiers output a key and these conditions fail, then either a matching message or final announcement was fixed before it was sent, which requires guessing fresh random bits and has probability at most $14\cdot2^{-k_c}+2^{-k_f}$, or an impersonation or substitution attack succeeds, which by Step~1 has probability at most $e^{-c_*L}+2m2^{-k_{\mathrm{in}}}$. Summing these errors, without conditioning on acceptance, gives the bounds on $\varepsilon_{\mathrm{cor}}$ and $\varepsilon_{\mathrm{sec}}$. Finally, an honest execution (without any adversaries) fails only if fewer than $n_{\mathrm{raw}}+k_{\mathrm{test}}$ of the $Q$ independent basis pairs match, which gives $\varepsilon_{\mathrm{comp}}$.
\end{proof}

\bibliographystyle{alpha}
\bibliography{references}

\appendix

\section{Proof of Basic Quantum Information Lemmas}
\label{sec:appendix-basic-proofs}

\begin{proof}[Proof of \Cref{lemma:direct-sum}]\label{pf:lemma:direct-sum}
We prove this by establishing both the upper and lower bounds.

For the upper bound, take for each value $x$ an arbitrary extension $\omega^x_{ABE_x}$ of $\omega^x_{AB}$, embed all $E_x$'s into a common environment $E$, and add a classical copy of $x$ to the environment. This gives an extension for which the conditional mutual information equals
\[
    \sum_x p_x I(A:B|E)_{\omega^x}.
\]
Taking the infimum over these extensions gives
\[
    E_{\rm sq}(AX:B)_\omega
    \le
    \sum_x p_x E_{\rm sq}(A:B)_{\omega^x}.
\]

For the lower bound, let $\Omega_{XABE}$ be any extension of $\omega_{XAB}$. Dephase $X$ in its classical basis. Since this is a local operation on Alice's side, conditional mutual information cannot increase. For the dephased state,
\[
    I(AX:B|E)
    =
    I(X:B|E)+I(A:B|EX)
    \ge
    \sum_x p_x I(A:B|E)_{\Omega^x}.
\]
Each $\Omega^x_{ABE}$ is an extension of $\omega^x_{AB}$, so the right-hand side is at least $2\sum_x p_x E_{\rm sq}(A:B)_{\omega^x}$. Taking the infimum over all extensions, gives the reverse inequality.
\end{proof}

\begin{proof}[Proof of \Cref{lem:quantum_dim_bound}]\label{pf:lem:quantum_dim_bound}
Since the state $\rho_{A_1A_2}$ is a classical-quantum state with orthogonal classical flags in $C$, by \Cref{lemma:direct-sum}, the total squashed entanglement is exactly
\begin{equation}
    E_{sq}(A_1:A_2)_{\rho} = E_{sq}(C A_1':A_2)_{\rho} = \sum_{c} p_c E_{sq}(A_1':A_2)_{\rho^{(c)}}.
\end{equation}
Since squashed entanglement for any bipartite state $\omega_{XY}$ is at most the von Neumann entropy of either local subsystem, which is in turn upper-bounded by the logarithm of the dimensions of either local subsystem, we get that for every flag value $c$, the squashed entanglement across $A_1':A_2$ is bounded as
\begin{equation}
    E_{sq}(A_1':A_2)_{\rho^{(c)}} \le S(A_1')_{\rho^{(c)}} \le \log_2 \dim(\mathcal{H}_{A_1'}).
\end{equation}
Averaging over $c$ we get
\[
E_{\rm sq}(A_1:A_2)_\rho\le \log \dim \mathcal H_{A'_1},
\]
which completes the proof.
\end{proof}

\section{Proofs of the Entanglement-Consumption Results}
\label{sec:appendix-proofs-of-entanglement-lemmas}

\subsection{Entanglement Estimates}
\begin{proof}[Proof of \Cref{lemma:lower_bound-closeness}]\label{pf:lemma:lower_bound-closeness}
 It is known that squashed entanglement upper-bounds one-way distillable entanglement, see Ref.~\cite{christandl2004squashed}, and by the hashing inequality, see Ref.~\cite{devetak2005distillation}, one-way distillable entanglement upper-bounds coherent information. Thus, for any bipartite state $\tau_{OW}$, we have
$$E_{sq}(O:W)_\tau \ge E_D^{\to}(O:W)_\tau \ge I_c(O \rangle W)_\tau = S(W)_\tau - S(OW)_\tau.$$

We bound the two von Neumann entropy terms individually using the Fannes-Audenaert inequality~\cite{audenaert2007sharp}, which states that for any two $d$-dimensional (or, equivalently $\log(d)$-qubit) states $\rho$ and $\sigma$ with trace distance $T\in [0, 1 - 1/d]$,
\begin{equation}\label{eq:fannes-audenart}
|S(\rho) - S(\sigma)| \le T \log_2(d - 1) + h(T).
\end{equation}

Since $\operatorname{TD}(\tau_{OW},\Phi^+)\le \epsilon'$, by Eq.~\eqref{eq:fannes-audenart} on the two-qubit system $OW$, we get
\[
S(OW)_\tau \le \epsilon'\log 3+h(\epsilon').
\]
Next, since taking the partial trace over system $O$ does not increase the trace-distance,  by the data processing inequality, we get
\[
\operatorname{TD}(\tau_W,I/2)\le \epsilon',
\]
and hence by Eq.~\eqref{eq:fannes-audenart}, on the single-qubit system $W$ we get,
\[
S(W)_\tau \ge 1-h(\epsilon').
\]
Combining the last two equations completes the proof:
\[
E_{\rm sq}(O:W)_\tau
\ge
1-\epsilon'\log 3-2h(\epsilon')
=
1-f(\epsilon').
\]
\end{proof}

\begin{proof}[Proof of \Cref{lem:bb84_fidelity}]\label{pf:lem:bb84_fidelity}
For any specific state $\rho_{AB}^{(x)}$ in the ensemble, let $F_x = \langle \Phi^+ | \rho_{AB}^{(x)} | \Phi^+ \rangle$ be its fidelity with the target state. We express its diagonal elements in the standard Bell basis $\{|\Phi^+\rangle, |\Phi^-\rangle, |\Psi^+\rangle, |\Psi^-\rangle\}$ as a probability distribution, which sums to $1$: $F_x + P(\Phi^-) + P(\Psi^+) + P(\Psi^-) = 1$.

The random BB84 test accepts if the measurement outcomes match. The projectors for passing in the $ZZ$ and $XX$ bases are respectively $\Pi_Z = |\Phi^+\rangle\langle \Phi^+| + |\Phi^-\rangle\langle \Phi^-|$ and $\Pi_X = |\Phi^+\rangle\langle \Phi^+| + |\Psi^+\rangle\langle \Psi^+|$. The success probability $p_x$ is the average of passing in either basis:
\begin{equation}
    p_x = \frac{1}{2} \text{Tr}\Big(\rho_{AB}^{(x)} (\Pi_Z + \Pi_X)\Big) = \frac{1}{2} \Big( 2F_x + P(\Phi^-) + P(\Psi^+) \Big).
\end{equation}

From normalization, $P(\Phi^-) + P(\Psi^+) \le 1 - F_x$. Substituting this upper bound yields $p_x \le \frac{1}{2} (F_x + 1)$. Rearranging gives the fidelity bound for the state indexed by $x$: $F_x \ge 2p_x - 1$. Averaging this lower bound over the distribution $q_x$ yields the average fidelity of the ensemble: $F_{avg} \ge \sum_x q_x (2p_x - 1) = 2p - 1$.
\end{proof}

Next we prove the required entanglement upper bound for a pure state.
\begin{proof}[Proof of \Cref{lemma:pure-state-flow}]\label{pf:lemma:pure-state-flow}
For the pure state $\omega_{O A_1 A_2}$, $E_{sq}(O A_1 : A_2)_\omega = S(A_2)$. Let $E$ be an arbitrary extension of $A_1 A_2$. By Stinespring's Dilation, $E$ is generated by an isometry $V : O \to E \otimes O_{\text{rem}}$. Thus, the global state on $E O_{\text{rem}} A_1 A_2$ is pure.

We expand the conditional mutual information:
\begin{align}
    I(A_1 : A_2 | E) &= S(A_1 | E) + S(A_2 | E) - S(A_1 A_2 | E) \\
    &= \Big[S(A_1 E) - S(E)\Big] + S(A_2 | E) - \Big[S(A_1 A_2 E) - S(E)\Big] \\
    &= S(A_1 E) - S(A_1 A_2 E) + S(A_2 | E).
\end{align}
Since the global 4-part state is pure, subsystem entropies equal their complements: $S(A_1 E) = S(A_2 O_{\text{rem}})$ and $S(A_1 A_2 E) = S(O_{\text{rem}})$. Substituting these:
\begin{equation}
    I(A_1 : A_2 | E) = S(A_2 O_{\text{rem}}) - S(O_{\text{rem}}) + S(A_2 | E) = S(A_2 | O_{\text{rem}}) + S(A_2 | E).
\end{equation}
Rewriting conditional entropies as $S(X|Y) = S(X) - I(X:Y)$ yields:
\begin{equation}
    I(A_1 : A_2 | E) = 2S(A_2) - I(A_2 : O_{\text{rem}}) - I(A_2 : E).
\end{equation}
By the Data Processing Inequality under the isometry $V$, $I(A_2 : O_{\text{rem}}) \le I(O : A_2)$ and $I(A_2 : E) \le I(O : A_2)$. Thus,
\begin{equation}
    I(A_1 : A_2 | E) \ge 2S(A_2) - 2I(O : A_2) \implies \frac{1}{2} I(A_1 : A_2 | E) \ge S(A_2) - I(O : A_2).
\end{equation}
Taking the infimum over all extensions $E$ yields $E_{sq}(A_1 : A_2)_\omega \ge S(A_2) - I(O : A_2)$. Substituting $S(A_2) = E_{sq}(O A_1 : A_2)_\omega$ and rearranging completes the proof.
\end{proof}

Next we prove the mutual-information inequality used in the initial recovery bound.
\begin{proof}[Proof of \Cref{lemma:superadditivity}]\label{pf:lemma:superadditivity}
For each $j$, by definition of conditional mutual information, 
\[
I(O_j:A|O_{<j})
=
S(O_j|O_{<j})-S(O_j|AO_{<j})
\ge
S(O_j|O_{<j})-S(O_j|A),
\]
where the last inequality follows from strong subadditivity of von Neumann entropy. Therefore, by adding and subtracting $S(O_j)$, we conclude
\[
I(O_j:A|O_{<j})
\ge
I(O_j:A)-\bigl(S(O_j)-S(O_j|O_{<j})\bigr).
\]

Next, summing over $j$ and using $\sum_{j=1}^k S(O_j|O_{<j})=S(O)$, we get
\[\sum_{j=1}^k I(O_j:A|O_{<j})\geq \sum_{j=1}^k I(O_j:A)
-\left(\sum_{j=1}^k S(O_j)-S(O)\right)
=
\sum_{j=1}^k I(O_j:A)-\Delta_O.\]

Finally, by the chain rule of mutual information, we conclude,
\[
I(O:A)=\sum_{j=1}^k I(O_j:A|O_{<j})\ge \sum_{j=1}^k I(O_j:A)
-\left(\sum_{j=1}^k S(O_j)-S(O)\right)
=
\sum_{j=1}^k I(O_j:A)-\Delta_O.
\]

\end{proof}

Finally, we combine the preceding two bounds to control the entanglement of the initial state.

\begin{proof}[Proof of \Cref{lemma:bounding-initial-state-k-extraction}]\label{pf:lemma:bounding-initial-state-k-extraction}
Let $\rho^{(c)} = |\psi_c\rangle\langle \psi_c|_{OA_1'A_2}$. Since each conditional state $\rho^{(c)}$ is pure, exact entropy conservation holds: $I(O : A_2)_{\rho^{(c)}} = 2S(O)_{\rho^{(c)}} - I(O : A_1')_{\rho^{(c)}}$. Taking the expectation over $C$ yields the conditional entropy conservation:
\begin{equation}
    I(O : A_2 | C)_\rho = 2S(O | C)_\rho - I(O : A_1' | C)_\rho.
\end{equation}
Since $C$ is a classical subregister of $A_1$, conditioning on $C$ ensures $I(O : A_1' | C)_\rho = I(O : A_1 | C)_\rho$. Thus
\begin{equation}\label{eq:cond-entropy}
    I(O : A_2 | C)_\rho = 2S(O | C)_\rho - I(O : A_1 | C)_\rho.
\end{equation}

By \Cref{lemma:superadditivity}, we can lower bound $I(O : A_1)_\rho$ as 
\begin{equation}
    I(O : A_1)_\rho \ge \sum_{j=1}^k I(O_j : A_1)_\rho - \Delta_O,
\end{equation}
where the correlation penalty $\Delta_O$ is as defined in \Cref{lemma:superadditivity}, i.e., $\Delta_O:= \sum_{j=1}^k S(O_j)_\rho - S(O)_\rho$. 

Since $A_1 = C \otimes A_1'$, the mutual information expands via the chain rule as $I(X : A_1)_\rho = I(X : C)_\rho + I(X : A_1' | C)_\rho$. Furthermore, conditioning on the classical flag $C$ ensures $I(X : A_1' | C)_\rho = I(X : A_1 | C)_\rho$. Thus, $I(X : A_1)_\rho = I(X : C)_\rho + I(X : A_1 | C)_\rho$. Applying this to both the joint and marginal terms gives:
\begin{equation}
    I(O : C)_\rho + I(O : A_1 | C)_\rho \ge \sum_{j=1}^k I(O_j : C)_\rho + \sum_{j=1}^k I(O_j : A_1 | C)_\rho - \Delta_O.
\end{equation}
Rearranging this isolates our conditional mutual information:
\begin{equation}\label{eq:isolated-cond-mi}
    I(O : A_1 | C)_\rho \ge \sum_{j=1}^k I(O_j : A_1 | C)_\rho + \sum_{j=1}^k I(O_j : C)_\rho - I(O : C)_\rho - \Delta_O.
\end{equation}

Substitute \eqref{eq:isolated-cond-mi} into the entropy conservation equation \eqref{eq:cond-entropy}:
\begin{equation}
    I(O : A_2 | C)_\rho \le 2S(O|C)_\rho - \sum_{j=1}^k I(O_j : A_1 | C)_\rho - \sum_{j=1}^k I(O_j : C)_\rho + I(O : C)_\rho + \Delta_O.
\end{equation}

We now expand the classical mutual informations using $I(X : C)_\rho = S(X)_\rho - S(X|C)_\rho$. By grouping the non-conditional entropies, the correlation penalty $\Delta_O$ perfectly cancels the mismatch:
\begin{align}
    I(O : A_2 | C)_\rho
    &\le 2S(O|C)_\rho - \sum_{j=1}^k I(O_j : A_1 | C)_\rho \nonumber \\
    &\quad - \left(\sum_{j=1}^k S(O_j)_\rho - \sum_{j=1}^k S(O_j|C)_\rho\right)
      + \left(S(O)_\rho - S(O|C)_\rho\right) + \Delta_O \nonumber \\
    &= S(O|C)_\rho + \sum_{j=1}^k S(O_j|C)_\rho - \sum_{j=1}^k I(O_j : A_1 | C)_\rho \nonumber \\
    &\quad + \underbrace{\left[ -\sum_{j=1}^k S(O_j)_\rho + S(O)_\rho \right]}_{-\Delta_O} + \Delta_O \nonumber \\
    &= S(O|C)_\rho + \sum_{j=1}^k S(O_j|C)_\rho - \sum_{j=1}^k I(O_j : A_1 | C)_\rho. 
    \label{eq:almost-there}
\end{align}

For each specific classical value $c$, by applying subadditivity on $\rho^{(c)} = |\psi_c\rangle\langle \psi_c|_{OA_1'A_2}$, followed by averaging over all values of $C$, we get $S(O|C)_\rho = \sum_c p_c S(O)_{\rho^{(c)}} \le \sum_c p_c \sum_{j=1}^k S(O_j)_{\rho^{(c)}} = \sum_{j=1}^k S(O_j|C)_\rho$. Therefore, the term $S(O|C)_\rho - \sum_{j=1}^k S(O_j|C)_\rho \le 0$, which in combination with the previous equation (Eq.~\eqref{eq:almost-there}), yields, 
\begin{equation}\label{eq:final-mi-bound}
    I(O : A_2 | C)_\rho \le 2\sum_{j=1}^k S(O_j|C)_\rho - \sum_{j=1}^k I(O_j : A_1 | C)_\rho =\sum_{j=1}^k \left( 2S(O_j|C)_\rho - I(O_j : A_1 | C)_\rho \right).
\end{equation}

Next, we lower bound $I(O_j : A_1 | C)_\rho$. For any specific classical value $c$, the state on $O A_1' A_2$ is $\rho^{(c)}$. The prover applies the local unitary $U_c^{(j)}$ to the subsystem $A_1'$, resulting in an intermediate state $\hat{\rho}^{(c,j)}$. Since mutual information is invariant under local unitary operations on the respective subsystems, we have the exact equality:
\begin{equation}
    I(O_j : A_1')_{\rho^{(c)}} = I(O_j : A_1')_{\hat{\rho}^{(c,j)}}.
\end{equation}
The target register $R_j$ is a subregister of $A_1'$. Discarding the remaining subsystem of $A_1'$ constitutes a completely positive trace-preserving (CPTP) map (the partial trace), resulting in the final extracted state $\tilde{\rho}^{(c,j)}$ on $O_j R_j$. By the quantum data processing inequality, mutual information is monotonically non-increasing under local CPTP maps. Therefore,
\begin{equation}
    I(O_j : A_1')_{\hat{\rho}^{(c,j)}} \ge I(O_j : R_j)_{\tilde{\rho}^{(c,j)}}.
\end{equation}
Chaining these gives $I(O_j : A_1')_{\rho^{(c)}} \ge I(O_j : R_j)_{\tilde{\rho}^{(c,j)}}$. By the Araki-Lieb triangle inequality~\cite{araki1970entropy} applied to the bipartite state on $O_j R_j$, we get
\begin{equation}
    I(O_j : R_j)_{\tilde{\rho}^{(c,j)}} \ge 2S(O_j)_{\tilde{\rho}^{(c,j)}} - 2S(O_j R_j)_{\tilde{\rho}^{(c,j)}}.
\end{equation}

Since $U_c^{(j)}$ and the partial trace act strictly on $A_1'$, the local entropy of $O_j$ is invariant, meaning $S(O_j)_{\tilde{\rho}^{(c,j)}} = S(O_j)_{\rho^{(c)}}$. Taking the expectation over $C$, we get
\begin{equation}
    I(O_j : A_1 | C)_\rho \ge 2S(O_j | C)_\rho - 2\EE_c \left[S(O_jR_j)_{\tilde{\rho}^{(c,j)}}\right].
\end{equation}

Substituting this lower bound into \eqref{eq:final-mi-bound}, we get
\begin{align}
    I(O : A_2 | C)_\rho &\le \sum_{j=1}^k \left[ 2S(O_j|C)_\rho - \left( 2S(O_j | C)_\rho - 2\EE_c \left[S(O_jR_j)_{\tilde{\rho}^{(c,j)}}\right] \right) \right] \nonumber \\
    &= 2 \sum_{j=1}^k \EE_c \left[S(O_jR_j)_{\tilde{\rho}^{(c,j)}}\right]. \label{eq:very-final-mi-bound}
\end{align}

For any $j$ and possible classical value $c$ of $C$, let $\delta^c_j$ denote the trace distance of $\tilde{\rho}^{(c,j)}_{O_jR_j}$ from the Bell state. Applying the Fannes-Audenaert theorem (\cite{audenaert2007sharp}), we get that for every $j\in [k]$ and every possible classical value $c$ of $C$, if $\delta^c_j\leq 3/4$, $|S(O_jR_j)_{\tilde{\rho}^{(c,j)}}|=|S(O_jR_j)_{\tilde{\rho}^{(c,j)}}-S(\Phi^+)|\leq g(\delta^c_j)$, and hence by definition of $g_{safe}$ and $g$, $S(O_jR_j)_{\tilde{\rho}^{(c,j)}}\leq g(\delta^c_j)=g_{safe}(\delta^c_j)$ for $0\leq \delta^c_j\leq 3/4$. Moreover note that for $\delta^c_j> 3/4$ by definition, $g_{safe}(\delta^c_j)=g(3/4)$ and it is easy to see $g(3/4)=2$. Since the $O_jR_j$ is a two-qubit system $S(O_jR_j)_{\tilde{\rho}^{(c,j)}}\leq 2=g_{safe}(\delta^c_j)$ for $\delta^c_j> 3/4$. Therefore, we conclude that for any $j\in [k]$, and for any value of $\delta^c_j\in [0,1]$, 
\[S(O_jR_j)_{\tilde{\rho}^{(c,j)}}\leq g_{safe}(\delta^c_j).\]

Averaging over $c$ and using Jensen's inequality for the concave nondecreasing function $g_{\rm safe}$, in combination with the last equation, yields
\begin{equation}
    \EE_c \left[S(O_jR_j)_{\tilde{\rho}^{(c,j)}}\right]  \le \EE_c \left[g_{safe}(\delta^c_j)\right]\leq g_{safe}(\delta_j).
\end{equation}
where the last inequality holds since $g_{safe}$ is globally concave and monotonically non-decreasing, and $\mathbb{E}_c[\delta^c_j] \le \delta_j$.

Substituting this into \eqref{eq:very-final-mi-bound} yields
\begin{equation}
     I(O : A_2 | C)_\rho \le 2 \sum_{j=1}^k \EE_c \left[S(O_jR_j)_{\tilde{\rho}^{(c,j)}}\right] \leq \sum_{j=1}^k 2g_{safe}(\delta_j).
\end{equation}

Combined with \Cref{lemma:pure-state-flow} for the conditional state, we conclude that for any possible classical value $c$ of $C$,
\[E_{sq}(O A_1' : A_2)_{\rho^{(c)}} \le E_{sq}(A_1' : A_2)_{\rho^{(c)}} +I(O:A_2)_{\rho^{(c)}}.\]

For a fixed value $c$, the classical register $C$ is in the pure, deterministic, and unentangled product state $|c\rangle\langle c|_C$. Appending a pure product state to a local subsystem does not change the squashed entanglement across any bipartition. Therefore, appending $C$ to $A_1'$ yields the exact equalities $E_{sq}(O A_1' : A_2)_{\rho^{(c)}} = E_{sq}(O A_1 : A_2)_{\rho^{(c)}}$ and $E_{sq}(A_1' : A_2)_{\rho^{(c)}} = E_{sq}(A_1 : A_2)_{\rho^{(c)}}$. Taking the expectation over $C$, we conclude
\begin{align}
    &\EE_c E_{sq}(O A_1 : A_2)_{\rho^{(c)}} \leq \EE_c E_{sq}(A_1 : A_2)_{\rho^{(c)}} +\EE_c I(O:A_2)_{\rho^{(c)}} \\
    &\implies \EE_c E_{sq}(O A_1 : A_2)_{\rho^{(c)}} \leq \EE_c E_{sq}(A_1 : A_2)_{\rho^{(c)}} + \sum_{j=1}^k 2g_{safe}(\delta_j).
\end{align}

Next, by \Cref{lemma:direct-sum}, $\EE_c E_{sq}(O A_1 : A_2)_{\rho^{(c)}}=E_{sq}(O A_1 : A_2)_\rho$, and $\EE_c E_{sq}(A_1 : A_2)_{\rho^{(c)}}=E_{sq} (A_1:A_2)_\rho$, substituting which in the last equation yields,
\begin{equation}
    E_{sq}(O A_1 : A_2)_\rho \leq E_{sq} (A_1:A_2)_\rho + \sum_{j=1}^k 2g_{safe}(\delta_j),
\end{equation}
which completes the proof.
\end{proof}

\subsection{Proofs of the entanglement-loss theorem and its corollaries}

\begin{proof}[Proof of \Cref{theorem:k-extraction-entanglement-consumption}]\label{pf:theorem:k-extraction-entanglement-consumption}
By the hypothesis that $U_j$ does not scramble the classical flag register, i.e., $U_j$ can be written as $\sum_{c}\ket{c}\bra{c}_C\otimes U^{(c)}_j$, the preconditions required for \Cref{lemma:bounding-initial-state-k-extraction} are satisfied.

Since $C$ is contained in $A_1$, the squashed entanglement decomposes as a direct sum over the conditional pure states: $E_{sq}(A_1 : A_2 | C)_\rho = \sum p_c E_{sq}(A_1' : A_2)_{\psi_c} = E_{sq}(A_1 : A_2)_\rho$ and $E_{sq}(O A_1 : A_2 | C)_\rho = \sum p_c E_{sq}(O A_1' : A_2)_{\psi_c} = E_{sq}(O A_1 : A_2)_\rho$.

Since the target Bell state $\Phi^+$ is pure, fidelity is linear over convex combinations. Thus, the global fidelity is exactly the expected conditional fidelity. For a fixed classical record $c$ and target pair $j$, let $F_{j,c}$ be the fidelity of the recovered state. The expected trace distance is bounded by $\delta_j = \mathbb{E}_c[TD_{j,c}] \le \sum_c p_c \sqrt{1 - F_{j,c}}$. Applying \Cref{lemma:bounding-initial-state-k-extraction} to the initial state $\rho$ and substituting the direct sums gives:
\begin{equation}
    E_{sq}(A_1 : A_2)_\rho \ge E_{sq}(O A_1 : A_2)_\rho - 2k \left[ \frac{1}{k} \sum_{j=1}^k g_{safe}(\delta_j) \right].
\end{equation}

Note that the composition of a globally concave, non-decreasing function with another globally concave decreasing function remains globally concave and non-increasing. Since $y(F) = \sqrt{1 - F}$ is concave and monotonically decreasing with respect to $F$ for $F\in (-\infty,1]$, and $g_{safe}(y)$ is concave and non-decreasing in $[0,1]$, their composition $H(F) = g_{safe}(\sqrt{1 - F})$ is a concave and non-increasing function of $F\in [0,1]$. Applying Jensen's inequality moves the summation inside the function due to the concavity of $H$.
\begin{equation}
\begin{aligned}
    \frac{1}{k} \sum_{j=1}^k g_{safe}(\delta_j)
    &\le \frac{1}{k} \sum_{j=1}^k g_{safe}\Bigg(\sum_c p_c \sqrt{1 - F_{j,c}}\Bigg) \\
    &\le g_{safe}\Bigg(\sqrt{1 - \sum_c p_c \frac{1}{k} \sum_{j=1}^k F_{j,c}}\Bigg).
\end{aligned}
\end{equation}

By the theorem's initial hypothesis and the linearity of fidelity, the globally averaged fidelity is lower-bounded by $1 - \epsilon^2$. Since $H(F)$ is monotonically decreasing with respect to $F$, substituting $F \ge 1 - \epsilon^2$ yields an upper bound on the penalty:
\begin{equation}
    g_{safe}\Big(\sqrt{1 - (1 - \epsilon^2)}\Big) = g_{safe}(\epsilon).
\end{equation}
Since $\epsilon \le 3/4$, $g_{safe}(\epsilon) = g(\epsilon)$. Combining the preceding three displays, we obtain that the initial state is bounded by $E_{sq}(A_1 : A_2)_\rho \ge E_{sq}(O A_1 : A_2)_\rho - 2k \cdot g(\epsilon)$.

By LOCC monotonicity, $E_{sq}(O A_1 : A_2)_\rho \ge E_{sq}(O A_1 : A_2)_\tau$. By monogamy, $E_{sq}(O A_1 : A_2)_\tau \ge E_{sq}(A_1 : A_2)_\tau + E_{sq}(O : A_2)_\tau$. Chaining these provides the flow 
\begin{equation}
\Delta E_{sq} \ge E_{sq}(O : A_2)_\tau - 2k \cdot g(\epsilon).
\end{equation}

To bound the final target entanglement, we use monogamy across the $k$ qubit registers in $O$ to get, 
\begin{equation}
E_{sq}(O : A_2)_\tau \ge \sum_j E_{sq}(O_j : A_2)_\tau.
\end{equation}
 Let $j\in [k]$ be arbitrary. By local unitary invariance under $V_j$ and partial trace monotonicity, we get that 
\[E_{sq}(O_j : A_2)_\tau \ge E_{sq}(O_j : W_j)_{\tilde{\tau}^{(j)}}.\]

Let $F'_{j}$ be the final fidelity of $\tilde{\tau}^{(j)}_{O_jW_j}$ with the Bell-pair state. Applying \Cref{lemma:lower_bound-closeness} via the trace distance bound $\sqrt{1 - F'_{j}}$ yields a lower bound as 
\[ E_{sq}(O_j : W_j)_{\tilde{\tau}^{(j)}}\geq 1 - f\Big(\sqrt{1 - F'_{j}}\Big) ,\]
 provided $\sqrt{1 - F'_{j}}\leq 1/2.$ Since $f_{safe}(\epsilon')=f(\epsilon')$ for $\epsilon'\leq 1/2$,    squashed entanglement is always non-negative, and $1-f_{safe}(\epsilon')=1-f(1/2)<0$ for $\epsilon'>1/2$, we can substitute the non-decreasing safe extension $f_{safe}$ in the last lower bound instead of $f$ without violating the inequality for any value of $\sqrt{1 - F'_{j}}\in [0,1]$, i.e.,  for any value of $F'_{j}$. Hence we conclude that for any value of $F'_{j}$, 
 \[E_{sq}(O_j : A_2)_\tau \ge 1 - f_{safe}\Big(\sqrt{1 - F'_{j}}\Big).\]
 Since $j\in [k]$ was arbitrary, combining the last equation with the displayed monogamy bound gives
\begin{equation}
    E_{sq}(O : A_2)_\tau \ge \sum_{j=1}^k \Big[ 1 - f_{safe}\Big(\sqrt{1 - F'_{j}}\Big) \Big] = k - \sum_{j=1}^k f_{safe}\Big(\sqrt{1 - F'_{j}}\Big).
\end{equation}

By the exact same calculus principle that we applied for $g_{safe}$ earlier in the proof, $f_{safe}(\sqrt{1 - F})$ is a globally concave function of fidelity. Applying Jensen's inequality using the concavity of $f_{safe}$ we get,
\begin{align}
    E_{sq}(O : A_2)_\tau
    &\ge k - \sum_{j=1}^k f_{safe}\Big(\sqrt{1 - F'_{j}}\Big) \nonumber \\
    &\ge k-kf_{safe}\left(\sqrt{1-\frac{\sum_{j}F'_{j}}{k}}\right) \nonumber \\
    &\ge k-kf_{safe}\left(\sqrt{1-(1-{\epsilon'}^2)}\right)
     = k(1-f_{safe}(\epsilon')).
\end{align}
where the last inequality follows from the hypothesis of the theorem that the final average fidelity (after applying $V_j$'s) with the Bell pair state is at least $ 1 - (\epsilon')^2$ where $\epsilon'\leq 1/2$.

Since $\epsilon' \le 1/2$, $f_{safe}(\epsilon') = f(\epsilon')$. Combining this with the last inequality and the displayed flow inequality completes the proof, i.e.,
\begin{equation}
    \Delta E_{sq} \ge k \Big[ 1 - 2g(\epsilon) - f_{safe}(\epsilon') \Big]= k \Big[ 1 - 2g(\epsilon) - f(\epsilon') \Big].
\end{equation}
\end{proof}

\begin{proof}[Proof of \Cref{cor:k-extraction-entanglement-consumption}]\label{pf:cor:k-extraction-entanglement-consumption}
We prove the corollary by reducing the local CPTP recovery maps to local unitary recovery maps on an extended Hilbert space.

For each initial recovery map $\mathcal E_j$ acting on $A_1$, let $E_j$ be a local Stinespring environment, initialized to a pure state $\ket{0}_{E_j}$. For each final recovery map $\mathcal F_j$ acting on $A_2$, let $F_j$ be a local Stinespring environment, initialized to a pure state $\ket{0}_{F_j}$. Define
\[
    E_{\mathrm{init}}:=E_1\cdots E_k,
    \qquad
    F_{\mathrm{final}}:=F_1\cdots F_k.
\]
The important point is that $E_{\mathrm{init}}$ is appended on the $A_1$ side, whereas $F_{\mathrm{final}}$ is appended on the $A_2$ side. Consider the extended state
\[
    \rho^{\mathrm{ext}}
    :=
    \rho_{OA_1A_2}
    \otimes
    \ket{0}\bra{0}_{E_{\mathrm{init}}}
    \otimes
    \ket{0}\bra{0}_{F_{\mathrm{final}}}.
\]
Since these are pure product ancillas appended locally to the two sides,
\begin{equation}
    \Esq(A_1E_{\mathrm{init}}:A_2F_{\mathrm{final}})_{\rho^{\mathrm{ext}}}
    =
    \Esq(A_1:A_2)_\rho .
\end{equation}

For every $j$, the Stinespring dilation of $\mathcal E_j$ gives a local unitary $U_j$ acting on $A_1E_j$, and hence on $A_1E_{\mathrm{init}}$ by tensoring with the identity on the other $E$-registers. Similarly, the Stinespring dilation of $\mathcal F_j$ gives a local unitary $V_j$ acting on $A_2F_j$, and hence on $A_2F_{\mathrm{final}}$ by tensoring with the identity on the other $F$-registers. If the CPTP maps are classically controlled by the flag registers and leave the original flag registers unchanged, the dilations can be chosen to be classically controlled in the same way.

Let
\[
    \Lambda^{\mathrm{ext}}
    :=
    \Lambda\otimes \mathrm{id}_{E_{\mathrm{init}}F_{\mathrm{final}}}.
\]
This is still an LOCC channel across the extended bipartition $A_1E_{\mathrm{init}}:A_2F_{\mathrm{final}}$. Let
\[
    \tau^{\mathrm{ext}}
    :=
    (\mathcal I_O\otimes \Lambda^{\mathrm{ext}})(\rho^{\mathrm{ext}}).
\]
Since the ancillas are not touched by $\Lambda^{\mathrm{ext}}$, the ancillae on each side remain in a product state with the rest of the registers, and hence
\begin{equation}
    \Esq(A_1E_{\mathrm{init}}:A_2F_{\mathrm{final}})_{\tau^{\mathrm{ext}}}
    =
    \Esq(A_1:A_2)_\tau .
\end{equation}

The unitary maps $\{U_j\}_{j=1}^k$ on $A_1E_{\mathrm{init}}$ reproduce the same initial recovered states as the CPTP maps $\{\mathcal E_j\}_{j=1}^k$ after tracing out the Stinespring environments. Hence the initial average fidelity condition is unchanged. Likewise, the unitary maps $\{V_j\}_{j=1}^k$ on $A_2F_{\mathrm{final}}$ reproduce the same final recovered states as the CPTP maps $\{\mathcal F_j\}_{j=1}^k$, so the final average fidelity condition is unchanged.

Thus the hypotheses of \Cref{theorem:k-extraction-entanglement-consumption} hold for the extended state and the extended bipartition $A_1E_{\mathrm{init}}:A_2F_{\mathrm{final}}$. Therefore,
\[
\begin{aligned}
&\Esq(A_1E_{\mathrm{init}}:A_2F_{\mathrm{final}})_{\rho^{\mathrm{ext}}}
-
\Esq(A_1E_{\mathrm{init}}:A_2F_{\mathrm{final}})_{\tau^{\mathrm{ext}}} \\
&\qquad\ge
k\bigl[1-2g(\epsilon)-f(\epsilon')\bigr].
\end{aligned}
\]
Finally, combining this lower bound with the two displayed invariance identities completes the proof:
\begin{align}
    \Esq(A_1:A_2)_\rho
    -
    \Esq(A_1:A_2)_\tau&=\Esq(A_1E_{\mathrm{init}}:A_2F_{\mathrm{final}})_{\rho^{\mathrm{ext}}}
-
\Esq(A_1E_{\mathrm{init}}:A_2F_{\mathrm{final}})_{\tau^{\mathrm{ext}}}\\
&\ge
    k\bigl[1-2g(\epsilon)-f(\epsilon')\bigr].
\end{align}

\end{proof}

\begin{proof}[Proof of \Cref{cor:k-extraction-entanglement-consumption-distributed}]\label{pf:cor:k-extraction-entanglement-consumption-distributed}
Append to $A_0$ a classical register $\widehat C_1$ containing a copy of $C_1$, and append to $A_1$ a classical register $\widehat C_0$ containing a copy of $C_0$. Denote the resulting state by $\widehat\rho$, and put
\begin{equation}
    \widehat A_0:=C_0\widehat C_1A_0',
    \qquad
    \widehat A_1:=C_1\widehat C_0A_1'.
\end{equation}
Copying a classical record across the cut is LOCC, while locally discarding the copy recovers the original state. The direct-sum identity therefore gives
\begin{equation}
    \Esq(\widehat A_0:\widehat A_1)_{\widehat\rho}
    =
    \Esq(A_0:A_1)_\rho.
\end{equation}

For the analysis, relabel the reference qubits on each branch $c$ so that $O_{r_j(c)}$ occupies the fixed reference register $O_j$. This leaves the reduced state on $\widehat A_0\widehat A_1$ unchanged at both endpoints and commutes with $\Lambda$, which preserves the records and reference qubits. For every $j$, define a recovery map on $\widehat A_0$ that reads $c$ and applies $\mathcal E_j^c$, and define the recovery map on $\widehat A_1$ similarly using $\mathcal F_j^c$. These are local quantum channels because the records are classical. Their average fidelities are the averages over $c$ and $j$ in the statement.

Thus \Cref{cor:k-extraction-entanglement-consumption} applies to $\widehat\rho$ and to the channel that acts as $\Lambda$ on the original registers and as the identity on $\widehat C_0\widehat C_1$. If $\widehat\tau$ is the output, then
\begin{equation}
    \Esq(\widehat A_0:\widehat A_1)_{\widehat\rho}
    -
    \Esq(\widehat A_0:\widehat A_1)_{\widehat\tau}
    \geq
    k\bigl[1-2g(\epsilon)-f(\epsilon')\bigr].
\end{equation}
Locally discarding $\widehat C_0,\widehat C_1$ gives $\tau$, so
\begin{equation}
    \Esq(\widehat A_0:\widehat A_1)_{\widehat\tau}
    \geq
    \Esq(A_0:A_1)_\tau.
\end{equation}
Combining the last two inequalities with the initial equality proves the claim.
\end{proof}

\section{Proofs Used in the Security Analysis}
\label{sec:appendix-security-helpers}

\subsection{From Overlapping to Nonoverlapping Slots}
\label{sec:proof-overlap-to-nonoverlap}

\begin{proof}[Proof of \Cref{prop:overlap-to-nonoverlap}]
\label{pf:prop:overlap-to-nonoverlap}
Let $\mathcal A$ be the original coalition, whose quantum-register size is $Q$ and whose success probability is $p$. Set $D:=\max\{L,R\}$, choose any $\Delta^\star>\max\{\Delta,2D\}$, and set $G:=\Delta^\star-\Delta>0$. The added waiting time $G$ is otherwise arbitrary, and we use the same value between consecutive slots in every round. Let $\delta_0:=0$ and $\delta_j:=(j-1)G$ for $j\geq1$. If $T_{r,j}$ is the original logical time of slot $j$ in round $r$, define
\begin{equation}
g_{r,j}:=(r-1)\delta_s+\delta_j,\qquad T_{r,j}^\star:=T_{r,j}+g_{r,j}.
\end{equation}
Thus every round is stretched in the same way, and the accumulated delay at the end of round $r$ is $r\delta_s=g_{r+1,1}$. We apply this delay to every action between rounds. All locations, verifier roles, message contents, and timing within each slot remain unchanged.

The idea is to run the same finite strategy while delaying each operation only as much as the information and registers it uses require. The new coalition uses the same circuit, local operations, and quantum registers as $\mathcal A$. When it sends a quantum state, we distinguish the local transfer into a counted outgoing register from the later launch of that register toward the verifier. The local source register is free after the transfer, while the outgoing register remains unavailable until the verifier receives or discards its state. This uses only the registers already counted for $\mathcal A$.

Assign delay $g_{r,j}$ to every instruction or idle value from slot $j$ of round $r$ and to the launch and release of every quantum message whose original arrival is the prescribed response deadline for that slot. Once a return instruction determines later idle values, assign those values the same delay. Assign delay $(r-1)\delta_s$ to information and actions fixed before the timed part of round $r$, and delay $r\delta_s$ to actions between rounds $r$ and $r+1$. A quantum message that cannot reach any prescribed response deadline cannot help an accepting execution, so the new coalition may discard it locally while keeping its outgoing register unavailable until the scheduled version of its original release event.

For every circuit event $e$, let $\mathcal P(e)$ contain the earlier events whose outputs or timing can affect whether, when, or how $e$ is performed. This includes every message or null value used in a wait or timeout check and every preceding use and release of a local or outgoing register needed by $e$. Let $h(e)$ be the largest delay assigned to $e$ or to any event on which it directly or indirectly depends, and set $h(e):=0$ if no such delay exists. Let $t(e)$ be its original time, and let $q(f,e)\geq0$ be the required time between $f\in\mathcal P(e)$ and $e$. Schedule $e$ at
\begin{equation}
t^\star(e):=\max\left\{t(e)+h(e),\max_{f\in\mathcal P(e)}\bigl(t^\star(f)+q(f,e)\bigr)\right\},
\end{equation}
where the second maximum is omitted when $\mathcal P(e)$ is empty. Thus an operation occurs only after its shifted time, after all information it uses has arrived, and after every register it needs is free.

Each adversary follows this rule using its own classical record and a copy of the original clock. If a classical message arrives early, the receiver stores it but does not use it before its scheduled delivery. An unused message wire provides its null value at the same event, and a wait or timeout check occurs only after every delivery event on which it depends. Induction over the finite circuit therefore shows that every operation receives the same information as in the original execution.

We now consider a response accepted in slot $i$ of round $r$, where $V_b$ sends the return instruction. By condition~\ref{cond:nonoverlap-ordering}, information from a later round is unavailable before the response deadline. Information from a later slot $j>i$ of $V_{1-b}$ becomes available no earlier than $T_{r,j}-d_{1-b}$, while the response is due at $T_{r,i}+d_b$, where $d_0:=L$ and $d_1:=R$. The available time is
\begin{equation}
(T_{r,i}+d_b)-(T_{r,j}-d_{1-b})=d_0+d_1-(T_{r,j}-T_{r,i})<d_0+d_1,
\end{equation}
which is shorter than the distance between the verifiers. Hence no information first available in a later slot can affect the response. Since the execution is accepted, its outgoing register is not occupied by another quantum message assigned the same or a later delay. Every event needed to prepare the response, including the preceding release of that register, therefore has delay at most $g_{r,i}$.

The scheduling rule now shows that every event preparing the response occurs no later than its original time plus $g_{r,i}$. Its state may enter the outgoing register earlier and remain there unchanged. The launch has delay $g_{r,i}$ and no prerequisite with a larger delay, so it occurs exactly $g_{r,i}$ later than in the original execution. The response follows the same path and reaches the same verifier at the shifted deadline in the same state. Its release is delayed by the same amount, so every later use of the outgoing register still occurs after that register becomes free.

Since the schedule covers the complete circuit, the same argument applies to every round, including a response prepared before the round in which it is received. Initialization and replenishment are delayed with the corresponding round boundaries without changing their paths or order. Every execution accepted by $\mathcal A$ therefore remains accepted.

No additional quantum register is introduced. Each delayed response uses the same counted outgoing register as before, and the register assigned to any locally discarded quantum message remains unavailable until its scheduled release. Only classical information may need to be stored for longer. Hence the new coalition has quantum-register size $Q$ and success probability at least $p$.

Finally, slot $j$ of round $r$ has interval $[T_{r,j}^\star-D,T_{r,j}^\star+D]$. Consecutive centers within a round are separated by $\Delta^\star>2D$, so these intervals are disjoint. The verifier and adversary locations remain unchanged, and no adversary is introduced at $0$.
\end{proof}

\subsection{Loss Due to the Verifiers' Messages}
\label{sec:proof-verifier-entanglement-accounting}

\begin{proof}[Proof of \Cref{lem:verifier-entanglement-accounting}]
\label{pf:lem:verifier-entanglement-accounting}
We first copy the complete earlier classical record to both sides. Locally discarding these copies recovers the original state, so the squashed entanglement does not change. We may therefore fix a value $c$ of the record and average over $c$ at the end using \Cref{lemma:direct-sum}.

Suppose first that the verifier receives a response at the prescribed deadline in an outgoing register $W$ sent by side $A_b$. The register $W$ already belongs to $A_b$, so the proof neither adds nor copies a response register. If no response arrives at the deadline, choose either side as $A_b$ and let $W$ be a one-dimensional trivial register on that side. Set $X:=A_b$ and $Y:=A_{1-b}$. The conditional state $\ket{\psi_c}_{OXY}$ is pure, and hence
\begin{equation}
\Esq(X:Y)_{\psi_c}\geq S(Y)_{\psi_c}-S(XY)_{\psi_c}=S(Y)_{\psi_c}-S(O)_{\psi_c}\geq S(Y)_{\psi_c}-R_H.
\label{eq:measurement-local-entropy}
\end{equation}

Let $z$ contain the verifier's measurement values, the resulting decision, and any additional outcome kept for the analysis. When a timely response is sent, this includes the verifier's chosen basis and both measurement outcomes. When no response is present and the round is rejected, the proof records the rank-one outcome representing the discarded reference qubit.
 Let $\tau_{c,z}^{A_0A_1}$ be the conditional state after the verifier acts and removes its systems. If we ignore $z$, the verifier acts only on $O$ and $W$, so the reduced state of $Y$ is unchanged. Since squashed entanglement is at most the entropy of either side and entropy is concave, we obtain
\begin{equation}
\sum_z p(z|c)\Esq(A_0:A_1)_{\tau_{c,z}}\leq\sum_z p(z|c)S(Y)_{\tau_{c,z}}\leq S(Y)_{\psi_c}\leq\Esq(A_0:A_1)_{\psi_c}+R_H.
\end{equation}
Since squashed entanglement is symmetric, $\Esq(X:Y)=\Esq(A_0:A_1)$ for either value of $b$. Thus the same bound holds regardless of which side sends the response and also when no response is sent.

Averaging over $c$ using \Cref{lemma:direct-sum} proves the result with the complete analysis record. The physical protocol is obtained by deleting any outcome kept only for the analysis, which cannot increase the squashed entanglement.
\end{proof}

\end{document}